\documentclass[journal,twoside,a4paper]{IEEEtran}

\makeatletter
\def\ps@IEEEtitlepagestyle{%
	\def\@oddfoot{\mycopyrightnotice}%
	\def\@oddhead{\myieeepaper}\relax
}
\def\mycopyrightnotice{%
	\begin{minipage}{\textwidth}
		\centering \scriptsize%
			\textcopyright~2022 IEEE. Personal use of this material is permitted.
			Permission from IEEE must be obtained for all other uses, in any current or\\
			future media, including reprinting/republishing this material for advertising
			or promotional purposes, creating new collective works,\\
			for resale or redistribution to servers or lists, or reuse of any copyrighted
			component of this work in other works.
	\end{minipage}
}
\def\myieeepaper{%
	\hfil{\centering\scriptsize%
	This paper has been published in \emph{IEEE Transactions on Information Theory},
	vol.~68, no.~7, pp.~4491--4517, July~2022
	[DOI: \href{https://doi.org/10.1109/TIT.2022.3151061}{10.1109/TIT.2022.3151061}].}%
	\@IEEEheaderstyle\leftmark\hfil\thepage%
}
\makeatother

\usepackage{footmisc}
\usepackage{algorithm}
\usepackage{algpseudocode}
\usepackage{enumitem}
\usepackage{cite}
\usepackage{hyperref}
\hypersetup{hidelinks}
\usepackage{amsmath,amssymb,amsthm}
\usepackage{array}
\usepackage{diagbox}
\usepackage[pdftex]{graphicx}
\usepackage{xcolor}

\newtheorem{thm}{Theorem}
\newtheorem{dfn}{Definition}
\newtheorem{cor}{Corollary}
\newtheorem*{remark}{Remark}

\algnewcommand\algorithmicswitch{\textbf{switch}}
\algnewcommand\algorithmiccase{\textbf{case}}
\algdef{SE}[SWITCH]{Switch}{EndSwitch}[1]{\algorithmicswitch\ #1\ \algorithmicdo}
	{\algorithmicend\ \algorithmicswitch}%
\algdef{SE}[CASE]{Case}{EndCase}[1]{\algorithmiccase\ #1}
	{\algorithmicend\ \algorithmiccase}%
\algtext*{EndCase}%

\renewcommand{\i}{\ensuremath{\mathrm{i}}}
\renewcommand{\j}{\ensuremath{\mathrm{j}}}
\renewcommand{\k}{\ensuremath{\mathrm{k}}}

\newcommand{\Z}{\ensuremath{\mathbb{Z}}}
\newcommand{\R}{\ensuremath{\mathbb{R}}}
\newcommand{\C}{\ensuremath{\mathbb{C}}}
\newcommand{\Ha}{\ensuremath{\mathbb{H}}}

\newcommand{\G}{\ensuremath{\mathcal{G}}}
\newcommand{\E}{\ensuremath{\mathcal{E}}}
\renewcommand{\L}{\ensuremath{\mathcal{L}}}
\newcommand{\Hu}{\ensuremath{\mathcal{H}}}

\renewcommand{\Re}[1]{\ensuremath{\mathrm{Re}\{#1\}}}
\renewcommand{\Im}[1]{\ensuremath{\mathrm{Im}\{#1\}}}

\newcommand{\T}{\ensuremath{\mathsf{T}}}
\renewcommand{\H}{\ensuremath{\mathsf{H}}}

\newcommand{\Q}{\ensuremath{\mathrm{Q}}}
\renewcommand{\mod}{\ensuremath{\mathrm{mod}}}

\DeclareMathOperator*{\argmin}{argmin}

\newcommand{\ve}[1]{{\mathchoice{\mbox{\boldmath$\displaystyle #1$}}%
		{\mbox{\boldmath$\textstyle #1$}}%
		{\mbox{\boldmath$\scriptstyle #1$}}%
		{\mbox{\boldmath$\scriptscriptstyle #1$}}}}
	
\def\ua{{u^{(1)}}}
\def\ub{{u^{(2)}}}
\def\uc{{u^{(3)}}}
\def\ud{{u^{(4)}}}
\def\va{{v^{(1)}}}
\def\vb{{v^{(2)}}}
\def\vc{{v^{(3)}}}
\def\vd{{v^{(4)}}}
\def\ma{{\ve{M}^{(1)}}}
\def\mb{{\ve{M}^{(2)}}}
\def\mc{{\ve{M}^{(3)}}}
\def\md{{\ve{M}^{(4)}}}
\def\ma{{\ve{M}^{(1)}}}
\def\mb{{\ve{M}^{(2)}}}
\def\mc{{\ve{M}^{(3)}}}
\def\md{{\ve{M}^{(4)}}}

\newcommand{\Lat}{\ensuremath{\ve{\Lambda}}}
\newcommand{\I}{\ensuremath{\ve{\mathbb{I}}}}
\newcommand{\vol}{\ensuremath{\mathrm{vol}}}
\newcommand{\rank}{\ensuremath{\mathrm{rank}}}

\newcommand{\dB}{\ensuremath{\mathrm{dB}}}
\newcommand{\corr}{\ensuremath{\mathrel{\widehat{=}}}} 

\newcommand{\commentcolor}{green!50!black}

\begin{document}

\title{\vspace*{-1mm}
	Algorithms and Bounds\\
	for Complex and Quaternionic Lattices\\
	With Application to MIMO Transmission}%
\author{Sebastian~Stern,~\IEEEmembership{Member,~IEEE,}
        Cong~Ling,~\IEEEmembership{Member,~IEEE,}
        and~Robert~F.H.~Fischer,~\IEEEmembership{Senior Member,~IEEE}%
\thanks{Sebastian Stern and Robert F.H.~Fischer are with the Institute of Communications
		Engineering, Ulm University, 89081 Ulm, Germany
		(e-mail: sebastian.stern@uni-ulm.de; robert.fischer@uni-ulm.de).
}%
\thanks{Cong Ling is with the Department of Electrical and Electronic Engineering,
		Imperial College London, London SW7 2AZ, United Kingdom
		(e-mail: cling@ieee.org).
}%
\thanks{
	The work of Sebastian Stern and Robert F.H.\ Fischer has been supported in part by
	the Deutsche Forschungsgemeinschaft (DFG) under grant \mbox{Fi 982/13-1} and the work
	of Cong Ling by the Engineering and Physical Sciences Research Council (EPSRC) under
	grant EP/S021043/1.
}%
\thanks{Parts of this work have been discussed in the conference paper
	\cite[DOI: \href{https://doi.org/10.1109/ICT.2018.8464898}{10.1109/ICT.2018.8464898}]{Stern:18}
	presented at the 25\textsuperscript{th} International Conference on Telecommunications,
	Saint Malo, France, 2018, and in Sebastian Stern's dissertation \cite[DOI: \href{https://doi.org/10.18725/OPARU-32407}{10.18725/OPARU-32407}]{Stern:19}.
}%
\vspace*{-4.75mm}%
}

\maketitle

\begin{abstract}
Lattices are a popular field of study in mathematical research,
but also in more practical areas like cryptology or
multiple-input/multiple-output (MIMO) transmission. In
mathematical theory, most often lattices over real numbers
are considered. However, in communications, complex-valued
processing is usually of interest. Besides, by the use of
dual-polarized transmission as well as by the combination of two time
slots or frequencies, four-dimensional (quaternion-valued) approaches
become more and more important. Hence, to account for this fact,
well-known lattice algorithms and related concepts are generalized in
this work. To this end, a brief review of complex arithmetic, including
the sets of Gaussian and Eisenstein integers, and an introduction to
quaternion-valued numbers, including the sets of Lipschitz and Hurwitz
integers, are given. On that basis, generalized variants of two important
algorithms are derived: first, of the polynomial-time LLL algorithm,
resulting in a reduced basis of a lattice by performing a special variant
of the Euclidean algorithm defined for matrices, and second, of an
algorithm to calculate the successive minima---the norms of the shortest
independent vectors of a lattice---and its related lattice points.
Generalized bounds for the quality of the particular results are
established and the asymptotic complexities of the algorithms are
assessed. These findings are extensively compared to conventional
real-valued processing. It is shown that the generalized approaches
outperform their real-valued counterparts in complexity and/or quality
aspects. Moreover, the application of the generalized algorithms to MIMO
communications is studied, particularly in the field of
lattice-reduction-aided and integer-forcing equalization.
\end{abstract}
\begin{IEEEkeywords}
Lattices, lattice reduction, LLL algorithm, successive minima, Gaussian integers,
	Eisenstein integers, quaternions, Lipschitz integers, Hurwitz integers,
	MIMO, lattice-reduction-aided equalization, integer-forcing equalization.
\end{IEEEkeywords}


\section{Introduction}

\noindent
\IEEEPARstart{T}{he} concept of lattices has been studied for almost two
centuries. Initial work was, e.g., published by Hermite~\cite{Hermite:50},
by Korkine and Zolotareff \cite{Korkine:73}, and by
Minkowski \cite{Minkowski:91}. Nevertheless, lattices remained a topic of
theoretical mathematical studies for quite a long time.

This situation dramatically changed with the advent of the digital
revolution in the late 20\textsuperscript{th} century. Suddenly,
enough computational power was available to implement and run
particular algorithms for lattice problems. The most prominent one was
proposed by Lenstra, Lenstra and Lov{\'a}sz \cite{Lenstra:82}. The LLL
algorithm calculates a \emph{reduced basis} of a lattice, i.e., a more
suited mathematical description of the lattice w.r.t.\ some quality
criteria, with only polynomial-time complexity. More powerful strategies
for lattice basis reduction were addressed in the sequel, e.g., the
concepts of Hermite-Korkine-Zolotareff (HKZ) reduction
\cite{Kannan:83,Schnorr:94,Zhang:12} or Minkowski reduction
\cite{Helfrich:85,Zhang:12}, that, however, demand an
exponentially-growing computational complexity for calculating the reduced
basis. All above-mentioned algorithms operate over real numbers, i.e.,
the lattices are defined over the integer ring~$\Z$.

Apart from the application of lattices in cryptological
schemes~\cite{Micciancio:09}, lattices gained popularity in the field of
multiple-input/multiple-output (MIMO) communications \cite{Tse:05}.
In particular, maximum-likelihood (ML) detection was enabled by the
sphere decoder \cite{Agrell:02}, however, with the burden of
a large computational complexity. An alternative, low-complexity
strategy was given with the concept of lattice-reduction-aided (LRA) 
equalization \cite{Yao:02,Windpassinger:03,Windpassinger:04,
	Wuebben:04,Wuebben:11,Stern:19,Fischer:19}.
Here, the channel equalization is performed in a \emph{more suited basis}
which is obtained by one of the above-mentioned lattice-basis-reduction
algorithms---most often, by the polynomial-time LLL algorithm.
Since, for block-fading channels, this calculation has only to be
done once in the beginning, the computational complexity is
dramatically decreased when compared with ML detection.
Besides, in comparison to straight-forward linear equalization of
the MIMO channel, the noise enhancement can significantly be lowered,
even resulting in the optimum diversity behavior as shown
in~\cite{Taherzadeh:07}.

A few years ago, the concept of integer-forcing (IF) linear (MIMO)
equalization has been introduced \cite{Zhan:14}. The LRA and IF
approaches share the philosophy of performing the channel equalization
in a more suited representation of the channel matrix such that
the noise enhancement inherently caused by equalization is lowered.
However, it was found out that the restriction to lattice basis
reduction---described by a unimodular integer transformation
matrix---is actually not required. Instead, it is sufficient that
this integer matrix has full rank---the lattice-basis-reduction
problem is weakened to the so-called \emph{successive-minima problem}.
For a more detailed insight into the topic, see, e.g., \cite{Fischer:19}.
These successive minima are also quite important for the derivation of
bounds for lattice-basis-reduction schemes, as they serve as lower bounds
for the norms of the basis vectors.

In MIMO transmission, the channel matrix is usually assumed to be
complex-valued due to representation in the equivalent complex-baseband
domain \cite{Fischer:02,vanTrees:04}. Since the algorithms available for
lattice-basis reduction have initially been real-valued, equalization was
performed with an equivalent real-valued representation of the complex
channel, resulting in doubled dimensions. Nevertheless, the concept can be
extended to the complex case. Then, the lattice is not considered over
$\Z$ any more, but over the complex integers---the so-called
\emph{Gaussian integers} $\G$ \cite{Huber:94a,Conway:99}. In particular,
the high-complexity HKZ and Minkowski reduction algorithms were adapted in
order to run over Gaussian integers \cite{Jiang:13,Ding:17}---mainly due to
the reason that the complex-valued algorithms lower the computational
complexity in comparison to their real-valued counterparts. Furthermore,
it was found out that the use of another complex integer ring may be
beneficial \cite{Tunali:15,Stern:16}---of the \emph{Eisenstein integers}~$\E$
\cite{Huber:94b,Conway:99}, forming the hexagonal lattice over the complex
numbers. To be precise, signal constellations that form subsets of the
Eisenstein integers may enable a \emph{packing gain} \cite{Stern:19}, i.e.,
they cover less space in the complex plane than those based on Gaussian
integers (quadrature-amplitude modulation (QAM)) with the same cardinality
while the minimum distance between the signal points stays the same. Hence,
the constellation's variance and the related required transmit
\mbox{power can be lowered}.

\markboth{Stern \MakeLowercase{\textit{et al.}}: Algorithms and Bounds for
	Complex and Quaternionic Lattices}%
{Stern \MakeLowercase{\textit{et al.}}: Algorithms and Bounds for Complex
	and	Quaternionic Lattices}%

Moreover, in recent years, four-dimensional signaling techniques have
become more and more popular. In particular, in the field of optical
communications, it is already quite common to employ both polarization
planes of electromagnetic waves \cite{Karlsson:16}, resulting in a
\emph{dual-polarized} transmission. In wireless (MIMO) communications,
\emph{dual-polarized antennas} have been designed, e.g.,
in \cite{Cui:14,Li:16,Ghaedi:20}. Besides, diversity schemes that are
suited to combine the transmit symbols of two different time steps or
frequencies are known for some time, e.g., the famous Alamouti
scheme \cite{Alamouti:98}. Given such a four-dimensional signal space,
its representation over the set of \emph{quaternion numbers}
\cite{Conway:99,Conway:03} is quite obvious \cite{Isaeva:95,Wysocki:06}.
Thereby, it has to be taken into account that quaternion-valued (scalar)
multiplication is not commutative any more, i.e., the quaternionic numbers
do not form a field but only a \emph{skew field}. Related integer rings
are given by the (non-Euclidean) set of \emph{Lipschitz integers}
$\L$---forming the quaternionic integer lattice (integers in each
component)---and by the \emph{Hurwitz integers} $\Hu$---enabling the
Euclidean property by additionally including all quaternions with only 
\emph{half-integer} components \cite{Conway:99,Conway:03}.

Concerning the low-complexity (polynomial-time) LLL reduction
algorithm \cite{Lenstra:82}, it is known for quite some time that it
can be interpreted to form some kind of \emph{Euclidean algorithm}
for matrices \cite{Napias:96}. Consequently, it can generally be applied
w.r.t.\ any integer ring with Euclidean property, i.e., over rings for
which a Euclidean algorithm can uniquely be defined. In \cite{Napias:96},
the possibility to apply LLL reduction over $\G$, $\E$, and $\Hu$
has briefly been discussed, however, only providing vague definitions
of bases that can be called reduced over these particular rings. In the
sequel, these criteria have been shown to be too unspecific \cite{Gan:09}.
Hence, in \cite{Gan:09}, a complex-valued (Gaussian-integer) adaption of the
LLL reduction algorithm has been proposed. To this end, the real-valued
quantization operation inherently performed in the initial LLL algorithm
has been replaced by a respective complex one, resulting in a stronger
reduction criterion compared to \cite{Napias:96}. In \cite{Stern:15}, an
algorithm operating over the set of Eisenstein integers has been employed
which is based on a further adaption of the quantization operation. Recently,
in \cite{Lyu:20}, LLL reduction over (other) imaginary quadratic
(i.e., complex-valued) fields has been studied, again based on the adjustment
of the corresponding quantization operation in the reduction algorithm. However,
all those publications cover special cases rather than providing generalized
criteria and algorithms. Moreover, besides the brief consideration
in \cite{Napias:96}, hardly anything is known about the possibility to define
an LLL reduction that operates over quaternion numbers.

Concerning algorithms that solve the successive minima problem, i.e.,
algorithms that calculate the matrices which are optimal w.r.t.\ noise
enhancement for the LRA/IF receiver concepts, the situation is similar:
Given the low-dimensional case, efficient algorithms have been proposed a few
years ago \cite{Ding:15,Fischer:16,Wen:19}. They are, though, either completely
restricted to the real-valued case or cover the complex-valued (Gaussian-integer)
case by employing its equivalent real-valued representation. Besides, in
\cite{Stern:16}, initial results for lattices over the Eisenstein integers have been
provided that have been obtained by a particular adaption of the sphere-decoding
algorithm \cite{Agrell:02}. A bound on the first successive minimum for the case
of imaginary quadratic fields has been considered in~\cite{Lyu:20}. Beyond that,
the definition and assessment of a \emph{generalized} successive-minima
algorithm---also w.r.t.\ quaternionic lattices---is still an open point.

Hence, apart from a review of complex lattices and the introduction to
quaternionic arithmetic, the aims and contributions of this work can be divided
into three main aspects. The first point is the generalization of the LLL
reduction approach based on its interpretation as Euclidean algorithm for matrices.
Given the fact that the related operations are actually modulo reductions defined
over Euclidean rings, generalized variants of the LLL algorithm are derived.
Theoretical analysis of the related properties and parameters is provided for
all real-, complex-, and quaternion-valued integer rings mentioned above. Beyond
that, the application of LLL reduction given a \emph{non-Euclidean ring}---in
particular the Lipschitz integers---is discussed. Moreover, the list-based
successive-minima algorithm from \cite{Fischer:16} is generalized in order to
determine the successive minima for all above-mentioned complex and quaternionic
(Euclidean/non-Euclidean) integer rings. All algorithms are provided in such a
way that the non-commutative behavior of quaternionic multiplication is adequately
taken into account.

The second main aspect deals with the generalization of quality bounds which have
originally been derived for real-valued LLL reduction \cite{Lenstra:82,Nguyen:09}
and/or the real-valued successive-minima problem \cite{Hermite:50,Zhang:12,Nguyen:09}.
This particularly concerns the norms of the basis vectors (and the respective
successive minima), as well as the orthogonality defect of a lattice basis.
It is shown that the quaternion-valued and/or complex-valued approaches may
outperform their real-valued equivalents---especially if lattices over the
Eisenstein or the Hurwitz integers are considered.  Moreover, the asymptotic
computational complexities of the different approaches are established. Concerning
the (polynomial-time) LLL lattice-basis-reduction approach, these derivations
reveal that the complexity can considerably be decreased if the respective
complex- or quaternion-valued variants are employed. By providing additional
results from numerical simulations, it is shown that the quality bounds and
complexity evaluations reflect the behavior that can be observed when
i.i.d.\ Gaussian stochastic models are applied in practice.

Finally, the application of the derived approaches in MIMO communications,
particularly in the case of (multi-user) MIMO uplink transmission \cite{Tse:05}
based on the concepts of LRA and IF equalization, is extensively studied.
This includes a discussion on how the quaternion-valued concept can be employed
in dual-polarized transmission, as well as  in the Alamouti-like combination of
two time steps or frequencies. Respective system models are derived and evaluated
by means of numerical simulations for particular transmission scenarios. These
results show that the theoretical derivations and bounds also reflect the behavior
in practical MIMO schemes. Beyond that, the combination of the proposed strategies
with advanced complex and quaternionic signal constellations---also
w.r.t.\ soft-decision decoding approaches---is briefly discussed.

The paper is structured as follows: In Sec.~\ref{sec:extensions}, complex
integer rings are briefly reviewed and an introduction to quaternions including
the sets of Lipschitz and Hurwitz integers is given. In Sec.~\ref{sec:algorithms},
the LLL algorithm as well a list-based algorithm for the determination of the
successive minima of a lattice are generalized. Related quality bounds and the
assessment of the computational complexities are provided in Sec.~\ref{sec:bounds}.
In Sec.~\ref{sec:mimo}, the particular application of the generalized algorithms
is regarded in the field of MIMO communications. The paper is closed by a brief
summary and an outlook in Sec.~\ref{sec:sum}.

\section{Two- and Four-Dimensional Extensions of the Real Numbers
	and Related Lattices}	\label{sec:extensions}

\noindent In this section, the sets of \emph{complex numbers} and
\emph{quaternions} that form a two- and four-dimensional extension
of the real numbers, respectively, are reviewed. The related algebras
are presented and important subsets, particularly integer rings, 
are discussed. On that basis, generalized \emph{lattices} are defined.

\subsection{Complex Numbers and Quaternions}

First, the extension of the real numbers $\R$ to complex numbers and
quaternions, respectively, is reviewed. For a deeper insight into the
topic, see \cite{Neumann:94,Conway:99,Conway:03}.

\subsubsection{Complex Numbers}

The set of complex numbers%
\begin{equation}	\label{eq:complex}
	\C=\{c = \underbrace{c^{(1)}}_{\Re{c}} +
		\underbrace{c^{(2)}}_{\Im{c}} \i 
		\mid c^{(1)},c^{(2)} \in \R \}
\end{equation}
forms a \emph{field extension} of the real numbers. It is obtained by
extending the first, real component $c^{(1)}$ (\emph{real part} $\Re{c}$)
by a second component $c^{(2)}$ which is multiplied by the
\emph{imaginary unit} $\i=\sqrt{-1}$ (\emph{imaginary part} $\Im{c}$).

The complex conjugate of $c\in\C$ reads
$c^{*} = c^{(1)} - c^{(2)}\,\i$
and the absolute value of $c$ is given as
$|c|=\sqrt{(c^{(1)})^2+ (c^{(2)})^2}$.
Scalar additions (and subtractions) over complex numbers are
performed individually per component. The multiplication of
two complex numbers $u,v \in \C$ can be expressed as%
\begin{equation}	\label{eq:compl_mul}
	\begin{aligned}
		w&=(u^{(1)}+u^{(2)}\,\i)\cdot (v^{(1)}+v^{(2)}\,\i)\\
	&= \underbrace{(u^{(1)} v^{(1)} - u^{(2)}v^{(2)})}_{w^{(1)}}+
	\underbrace{(u^{(1)} v^{(2)} + u^{(2)}v^{(1)})}_{w^{(2)}}\,\i \; .
	\end{aligned} 
\end{equation}
Hence, four multiplications and two additions/subtractions are
required. Following the concept of the Karatsuba algorithm
\cite{Karatsuba:62}, this multiplication can alternatively be realized
by three multiplications and five additions/subtractions.
The scalar division of $u$ by $v$ is performed by the scalar
multiplication $u\cdot v^{-1}$ with the element
$v^{-1}=v^{*}/|v|^2=v^{*}\cdot(v^{*}v)^{-1}$.

Based on~\eqref{eq:compl_mul}, an \emph{equivalent real-valued}
representation of complex matrices can be given. An $N \times K$
matrix $\ve{C}\in\C^{N \times K}$ may be represented via its
equivalent $2N\times 2K$ real matrix%
\begin{equation}	\label{eq:equiv_real_compl}
	\ve{C}_\mathsf{r} =
	\begin{bmatrix} 
		\ve{C}^{(1)} & -\ve{C}^{(2)}\\
		\ve{C}^{(2)} & \phantom{-}\ve{C}^{(1)}
	\end{bmatrix}
	\in \R^{2N \times 2K} \; ,
\end{equation}
where $\ve{C}^{(1)}$ and $\ve{C}^{(2)}$ denote the real and
imaginary part of $\ve{C}$, respectively. Hence, the dimensions are
increased by a factor of $D_\mathsf{r}=2$. This variable also
represents the number of independent real-valued components
in~\eqref{eq:complex}. If $N \geq K$,%
\begin{equation}	\label{eq:det_real_compl}
	{\det}(\ve{C}^\H\ve{C}) =
		\sqrt{\det(\ve{C}_\mathsf{r}^\T\ve{C}_\mathsf{r}^{})}
\end{equation}
is valid \cite{Neumann:94}, where $\ve{C}^\H$ denotes the Hermitian of
$\ve{C}$, i.e., the conjugated transpose.
Utilizing~\eqref{eq:equiv_real_compl}, the matrix addition (and subtraction)
$\ve{S}=\ve{U}+\ve{V}$, where $\ve{U}$ and $\ve{V}$ denote complex matrices,
as well as the related complex matrix multiplication (and division)
$\ve{W}=\ve{U}\cdot\ve{V}$, can isomorphically be represented
by the real-valued addition
$\ve{S}_\mathsf{r}=\ve{U}_\mathsf{r}+\ve{V}_\mathsf{r}$
and the real-valued multiplication
$\ve{W}_\mathsf{r}=\ve{U}_\mathsf{r}\cdot\ve{V}_\mathsf{r}$,
respectively.

\subsubsection{Quaternions}

The set of quaternions\footnote{%
	In honor of Sir William Rowan Hamilton,
	the set of quaternions is denoted by $\Ha$.%
	}
\cite{Conway:99,Conway:03}%
\begin{equation}	\label{eq:quat}
	\begin{aligned}
		\Ha &= \{q = \underbrace{q^{\{1\}}}_{q^{(1)}+q^{(2)}\i} +
			\underbrace{q^{\{2\}}}_{q^{(3)}+q^{(4)}\i}
			\j\mid q^{\{1\}},q^{\{2\}}\in\C\}\\
			& = \{q = q^{(1)} + q^{(2)}\,\i + q^{(3)}\,\j+q^{(4)}\,\k \mid
			\\ & \qquad\qquad 
			q^{(1)},q^{(2)},q^{(3)},q^{(4)}\in\R \}
	\end{aligned}
\end{equation}
extends the set of complex numbers by an \emph{additional} complex-valued
component which is multiplied by the imaginary unit~$\j$. Hence,
\emph{four real-valued} components are present, where the real part of a
quaternion reads $\Re{q}=q^{(1)}$ and its imaginary part is represented by
the 3-tuple $\Im{q}=(q^{(2)},q^{(3)},q^{(4)})$. The related imaginary
\emph{quaternion units} are given as $\i$, $\j$, and $\k=\i\,\j$. The
relations between these units are described by the \emph{Hamilton equations}
\cite{Conway:03} which are stated in Table~\ref{tab:hamilton}.

\begin{table}
	\caption{\label{tab:hamilton} Hamilton Equations \cite{Conway:03}
		for the Product $u \cdot v$, where $u$ and $v$ are
		Quaternion Units, i.e., $u,v\in\{1,\i,\j,\k\}$.}%
	\centerline{
	\begin{tabular}{c||c|c|c|c}
		\hline
	\diagbox{$u$}{$v$} & $1$ & $\i$ & $\j$ & $\k$ \\
	\hline\hline
		$1$ & $+1$	& $+\i$ & $+\j$ & $+\k$	\\
		\hline
		$\i$ & $+\i$ & $-1$ & $+\k$ & $-\j$ \\
		\hline
		$\j$ & $+\j$ & $-\k$ & $-1$ & $+\i$ \\
		\hline
		$\k$ & $+\k$ & $+\j$ & $-\i$ & $-1$ \\
		\hline
	\end{tabular}
	}%
\end{table}

From Table~\ref{tab:hamilton}, it becomes apparent that the multiplication of
two quaternions is---in general---not commutative. Consequently, the quaternions
do not form a field but only a \emph{skew field}, i.e., they fulfill all
conditions which are required to form a field---except for the commutativity of
the multiplication.

By analogy with complex numbers, the conjugate of a quaternion $q\in\Ha$ is given
as $q^{*}=q^{(1)}-q^{(2)}\,\i-q^{(3)}\,\j-q^{(4)}\,\k$. The absolute value of $q$
is uniquely defined by
$|q|=\sqrt{q q^{*}}=\sqrt{q^{*} q}=
	\sqrt{(q^{(1)})^2+(q^{(2)})^2+(q^{(3)})^2+(q^{(4)})^2}$.
Moreover, additions (and subtractions) are performed individually per component.
The (non-commutative) multiplication of two quaternions $u,v\in\Ha$ is expressed
as \cite{Conway:03}%
\begin{equation}	\label{eq:quat_mul}
\begin{aligned}
	u \cdot v =\;& (\ua \va - \ub \vb - \uc \vc - \ud \vd) \\
	+\;&  (\ua \vb + \ub \va + \uc \vd - \ud \vc) \,\i \\
	+\;& (\ua \vc - \ub \vd + \uc \va + \ud \vb) \,\j \\
	+\;& (\ua \vd + \ub \vc - \uc \vb + \ud \va) \,\k \; ,
\end{aligned}
\end{equation}
i.e., 16 multiplications and twelve additions/subtractions are required.
Alternatively, this multiplication can be realized using eight multiplications and 28
additions/subtractions \cite{Cariow:15}. The division can be implemented via the
multiplication with the inverse element $v^{-1}=v^{*}\cdot(v^* v)^{-1}$, where this
choice ensures that $v v^{-1}=1$ (right inverse) and $v^{-1} v=1$ (left inverse).

Similar to the real-valued representation of complex matrices, the quaternion-valued
arithmetic defined in \eqref{eq:quat_mul} can be realized by the
\emph{equivalent complex- or real-valued matrix representation}. In particular, an
$N \times K$ matrix $\ve{M}\in\Ha^{N \times K}$ can be represented as $2N \times 2K$
complex-valued matrix\footnote{%
	\label{foot:repr}%
	The complex- and real-valued representations
	\eqref{eq:equiv_compl_quat} and~\eqref{eq:equiv_real_quat},
	respectively, are not unique. There exist several
	representations that differ in the positions
	of the minus signs within the matrices
	$\ve{M}_\mathsf{c}$ and $\ve{M}_\mathsf{r}$,
	see, e.g., \cite{Aslaksen:96,Conway:99,Mueller:12}.
	However, all these representations isomorphically
	express the quaternion-valued multiplication (division)
	according to~\eqref{eq:quat_mul}.%
}%
\begin{equation}		\label{eq:equiv_compl_quat}
		\begin{aligned}
	\ve{M}_\mathsf{c}&=
	\begin{bmatrix}
		\ve{M}^{\{1\}}  & -\ve{M}^{\{2\}} \\
		(\ve{M}^{\{2\}})^{*} & \phantom{-}(\ve{M}^{\{1\}})^{*}\\
	\end{bmatrix}\\
	&= 
	\begin{bmatrix}
		\ma + \mb  \i  & -\mc - \md  \i  \\
		\mc - \md  \i  & \phantom{-}\ma - \mb  \i \\
	\end{bmatrix} \; ,
	\end{aligned}
\end{equation}
where \eqref{eq:equiv_compl_quat} directly corresponds
to~\eqref{eq:equiv_real_compl}, i.e., the dimensions are increased by a factor
of $D_\mathsf{c}=2$, with the only difference that an additional conjugation has to
be performed in the second row. In~\eqref{eq:equiv_real_compl}, this step is not
required since only real numbers are present. By plugging
\eqref{eq:equiv_compl_quat} into \eqref{eq:equiv_real_compl}, i.e., by forming
the real-valued representation of the complex matrix $\ve{M}_\mathrm{c}$, one would
obtain one particular real-valued $4N \times 4K$ representation of the
quaternion-valued matrix $\ve{M}$. However, for the subsequent system model, it is
more convenient to form a real-valued representation according to\footref{foot:repr}%
\begin{equation}	\label{eq:equiv_real_quat}
	\ve{M}_\mathsf{r}
	= 
	\begin{bmatrix}
		\ma & -\mb & -\mc & -\md\\
		\mb & \phantom{-}\ma & -\md & \phantom{-}\mc\\
		\mc & \phantom{-}\md & \phantom{-}\ma & -\mb\\
		\md& -\mc & \phantom{-}\mb & \phantom{-}\ma\\
	\end{bmatrix} \; ,
\end{equation}
in which the four components are directly stacked in the left-most column.
The dimensions are increased by a factor of $D_\mathsf{r}=4$, representing
the four independent real-valued components in~\eqref{eq:quat}.
Here, if $N \geq K$, we have\footnote{%
	Due to the skew-field property, quaternion-valued determinants
	do not necessarily posses all properties which are known from
	real or complex ones. However, for Hermitian matrices (here,
	the Gramian $\ve{M}^\H\ve{M}$), they can, to a large extent, be
	employed just like real or complex determinants,
	cf.~\cite{Aslaksen:96}.%
}%
\begin{equation}	\label{eq:det_real_quat}
	{\det}(\ve{M}^\H\ve{M}) =
	\sqrt{\det(\ve{M}_\mathsf{c}^\H\ve{M}_\mathsf{c}^{})} = 
	\sqrt[4]{{\det}(\ve{M}_\mathsf{r}^\T\ve{M}_\mathsf{r}^{})}
	\;,
\end{equation}
$\ve{M}^\H$ denoting the Hermitian (conjugated transpose) of $\ve{M}$.

\subsection{Integer Rings}

The set of \emph{integers} $\Z = \{\dots,-2,-1,0,1,2,\dots\}$ forms a subset of
the real numbers $\R$ and additionally a \emph{Euclidean ring}. Hence, for
$u,\sigma,v,\rho\in\Z$ and $v \neq 0$, a division $u/v$ with
\emph{small remainder} according to\footnote{%
	We assume that negative remainders may occur, i.e.,
	the modulo operation defined in~\eqref{eq:mod_Z}
	is assumed to be symmetric w.r.t.\ the origin.%
}%
\begin{equation}	\label{eq:division}
u = \sigma \cdot v +\rho
\end{equation}
is possible, where the term small remainder implicates that $|\rho| < |v|$
is valid \cite[Def.~2.5]{Weintraub:08}. Consequently, the \emph{Euclidean algorithm}
\cite{Cormen:09} can be used to calculate the \emph{greatest common divisor} (gcd)
of two numbers $u,v\in\Z$. The minimum squared distance between the elements of $\Z$
reads $d_\mathrm{min,\Z}^2=1$.

A real number $r\in\R$ is quantized to its nearest integer via%
\begin{equation} \label{eq:quant_Z}
	\Q_{\Z}\{r\} = \lfloor r\rceil \in\Z\; ,
\end{equation}
i.e., by a simple rounding operation $\lfloor\cdot\rceil$, where w.l.o.g.\ we assume
that ties are resolved towards~$+\infty$. The maximum squared quantization error
occurs for all half-integer values ($\Z+\frac{1}{2}$) and is given as
$\epsilon^2_\Z=|\Q_{\Z}\{\frac{1}{2}\} - \frac{1}{2}|^2=\frac{1}{4}$.
The related (non-squared) error corresponds with the maximum of the remainder
in~\eqref{eq:division}, i.e., we have $|\rho|\leq\frac{1}{2} < |v|$,
since $|v|\geq 1~\forall v\in \Z \setminus \{0\}$. Based on the
quantization~\eqref{eq:quant_Z}, the \emph{modulo function}%
\begin{equation}	\label{eq:mod_Z}
	\mod_\Z\{r\}=r - \Q_{\Z}\{r\}
\end{equation}
yields a congruent point $r+\lambda$, $\lambda\in\Z$, located within
$[-\frac{1}{2},\frac{1}{2})$, forming the Voronoi cell of $\Z$ w.r.t.\ the
origin \cite{Conway:99,Fischer:02}.

\subsubsection{Complex-Valued Integer Rings}

Integers in the complex plane are represented by the
\emph{Gaussian integers} \cite{Huber:94a,Bossert:99,Conway:99}%
\begin{equation}
	\G = \{c = {c^{(1)}} + {c^{(2)}} \i \mid c^{(1)},c^{(2)} \in \Z \}
	= \Z + \Z\,\i \; .
\end{equation}
They are illustrated in Fig.~\ref{fig:GEint} (left). The minimum squared distance
between the elements reads $d_\mathrm{min,\G}^2=1$. The quantization of a complex
number $c\in\C$ to $\G$ is performed as%
\begin{equation} \label{eq:quant_G}
	\Q_{\G}\{c\} = \lfloor c^{(1)}\rceil + \lfloor c^{(2)}\rceil\,\i\in\G\; ,
\end{equation}
where the maximum squared quantization error sums up to
$\epsilon^2_\G=|\Q_{\G}\{\frac{1}{2}+\frac{1}{2}\i\} - 
	(\frac{1}{2}+\frac{1}{2}\i)|^2=\frac{1}{2}$.
The modulo operation $\mod_\G\{c\}=c-\Q_\G\{c\}$ reduces a complex number~$c$ to the
Voronoi cell of $\G$ (w.r.t.\ the origin), which forms a square in the complex plane
where all values are located within the range $[-\frac{1}{2},\frac{1}{2})$ per component.

\begin{figure}[t]
	\centerline{%
		\includegraphics{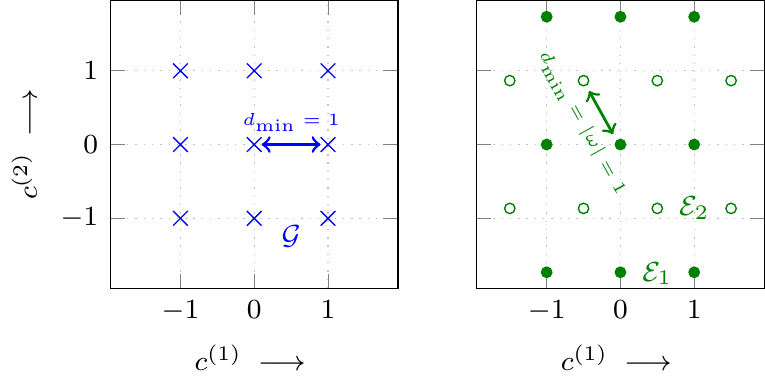}%
	}%
	\caption{\label{fig:GEint}%
		Illustration of the Gaussian integers $\G$ (left) and the Eisenstein integers $\E$
		(right). For the Eisenstein integers, the two subsets $\E_1$ (filled circles) and
		$\E_2$ (hollow circles) are shown.}%
\end{figure}

The \emph{Eisenstein integers} \cite{Huber:94b,Conway:99}%
\begin{equation}
	\E = \{c = {c^{(1)}} + {c^{(2)}} \omega \mid c^{(1)},c^{(2)} \in \Z \}
	= \Z + \Z\,\omega \; ,
\end{equation}
with the Eisenstein unit $\omega=\mathrm{e}^{\frac{2\pi}{3}\i}$ (third root of unity),
represent the hexagonal numbers ($\ve{A}_2$ lattice \cite{Conway:99}) in the complex
plane, cf.\ Fig.~\ref{fig:GEint} (right). Without a decrease in minimum distance
($d_\mathrm{min,\E}^2=1$), the elements are more densely packed (particularly, the
densest packing in two-dimensions is achieved \cite{Conway:99,Fischer:02}).
The quantization is realized as \cite{Conway:82,Conway:99}
\begin{align}	\label{eq:quant_E}
	\Q_{\E}\{c\} &= \argmin_{\Q_{\E_1}\{c\},\Q_{\E_2}\{c\}}
		\{|c-\Q_{\E_1}\{c\}|,|c -\Q_{\E_2}\{c\}|\} \; ,\\
	\Q_{\E_1}\{c\} &= \Q_{\Z}\left\{c^{(1)}\right\} +
		\sqrt{3}\,\Q_{\Z}\left\{\frac{c^{(2)}}{\sqrt{3}}\right\} \i\; ,\\
	\Q_{\E_2}\{c\} &= \Q_{\Z}\left\{c^{(1)}-\frac{1}{2}\right\} +\frac{1}{2}
		+ \nonumber	\\&
	\left(\sqrt{3}\,\Q_{\Z}\left\{\frac{c^{(2)}-
		\frac{\sqrt{3}}{2}}{\sqrt{3}}\right\}+\frac{\sqrt{3}}{2}\right) \i \; ,
\end{align}
i.e., by performing a quantization to the subsets $\E_1$ (filled circles in
Fig.~\ref{fig:GEint} (right)) and $\E_2$ (hollow circles) and a subsequent
decision to the point which is located closer to the original value $c\in\C$.
The maximum squared quantization error reads $\epsilon_{\E}^2=\frac{1}{3}$,
cf.~\cite{Conway:99,Fischer:02,Stern:19}. The modulo operation
$\mod_\E\{c\}=c-\Q_{\E}\{c\}$ calculates a point $c+\lambda$, $\lambda\in\E$,
located within the \emph{hexagonal} Voronoi cell of $\E$ w.r.t.\ the origin.

Both Gaussian and Eisenstein integers form Euclidean rings \cite{Napias:96}.
Due to the maximum squared quantization errors, for the former,
$|\rho|\leq\frac{1}{\sqrt{2}} < |v|$, with $|v|\geq 1~\forall v\in \G \setminus \{0\}$,
is valid if the division with remainder according to~\eqref{eq:division} is performed
over $\G$. For the latter, $|\rho|\leq\frac{1}{\sqrt{3}} < |v|$, with
$|v|\geq~1\ \forall v\in \E \setminus \{0\}$, holds. A division with small remainder can
be performed by analogy with~\eqref{eq:division} and, thus, it is possible to define a
Euclidean algorithm over the Gaussian and the Eisenstein integers.

\subsubsection{Quaternion-Valued Integer Rings}

With regard to the set of quaternions $\Ha$, two important subsets, in particular
integer rings, can be defined. The first type are the \emph{Lipschitz integers}%
\begin{equation}
	\begin{aligned}
		\L &= \{q = q^{(1)} + q^{(2)}\i + q^{(3)}\j + q^{(4)} \k 
		\mid\\& \qquad\qquad
		q^{(1)},q^{(2)},q^{(3)},q^{(4)} \in \Z \}\\
		&	= \Z + \Z\,\i + \Z\,\j + \Z\,\k \; ,
	\end{aligned}
\end{equation}
i.e., following the philosophy of the Gaussian integers, integer values are present
in each of the four components. A two-dimensional projection of the Lipschitz integers
is illustrated in Fig.~\ref{fig:LHint} (left). Again, the minimum squared distance reads
$d_{\mathrm{min},\L}^2=1$. The quantization of a quaternion $q\in\Ha$ to the closest
Lipschitz integer is realized by%
\begin{equation} \label{eq:quant_L}
	\Q_{\L}\{q\} = \lfloor q^{(1)}\rceil + \lfloor q^{(2)}\rceil\,\i
		+ \lfloor q^{(3)}\rceil\,\j + \lfloor q^{(4)}\rceil\,\k \in\L\; .
\end{equation}
The modulo operation reads $\mod_\L\{q\}=q-\Q_{\L}\{q\}$, where the Voronoi region
constitutes a \emph{hypercube} with the range $[-\frac{1}{2},\frac{1}{2})$ per component.
The maximum squared quantization error is---in comparison to the Gaussian
integers---increased to
$\epsilon^2_\L=|\frac{1}{2}+\frac{1}{2}\i+\frac{1}{2}\j+\frac{1}{2}\k|^2={1}$.
A direct consequence thereof is that the division with remainder according
to~\eqref{eq:division}, with $u,\sigma,v,\rho\in\L$, is not a \emph{Euclidean one} any
more \cite{Conway:03}: Here, the case $u v^{-1}\in\L+(1+\i+\j+\k)/2$ may occur.
Then, for the absolute value of the remainder,
$|\rho|=|(1+\i+\j+\k)/2|=1 \leq |v|$,  with $|v|\geq 1~\forall v\in \L \setminus \{0\}$,
is obtained. Hence, $|\rho|=|v|$ may be present, i.e., the \emph{in}equality required
to ensure a \emph{small} remainder is, in general, not achieved. Thus, it is \emph{not}
possible to define a Euclidean algorithm in order to calculate the gcd of two Lipschitz
integers $u,v\in\L$.

\begin{figure}[t]
	\centerline{%
		\includegraphics{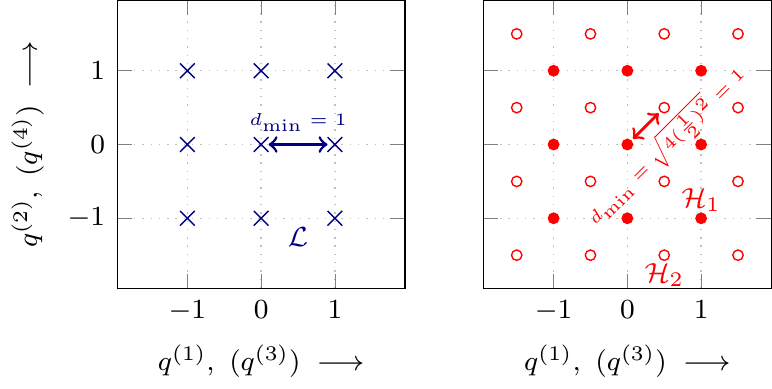}%
	}%
	\caption{\label{fig:LHint}%
		Two-dimensional projection of the Lipschitz integers $\L$ (left) and the Hurwitz
		integers~$\Hu$ (right). The components $q^{(3)}$ and $q^{(4)}$ are projected onto
		$q^{(1)}$ and $q^{(2)}$, respectively. For the Hurwitz integers, the two subsets
		$\Hu_1$ (filled circles) and $\Hu_2$ (hollow circles) are shown.}%
\end{figure}

Nevertheless, the non-Euclidean property of the Lipschitz integers can be ``cured'' by
the insertion of additional points at the problematic coordinates---the half-integer
values located at $\L+(1+\i+\j+\k)/2$. Then, as depicted in Fig.~\ref{fig:LHint} (right),
the so-called \emph{Hurwitz integers} \cite{Conway:99,Conway:03}%
\begin{equation}	\label{eq:hurwitz}
	\begin{aligned}
		\Hu &= \Big\{q = q^{(1)} + q^{(2)}\i + q^{(3)}\j + q^{(4)} \k \mid 
		\\ &\qquad\quad
		(q^{(1)},q^{(2)},q^{(3)},q^{(4)}) \in \Z^4 \cup \Big(\Z+\frac{1}{2}\Big)^4 \Big\} \\
	\end{aligned}
\end{equation}
are obtained, where \emph{all} components are---at the same time---\emph{either} integers
\emph{or} half-integers. In particular, the number of points (in quaternion space) is
doubled within the same hypervolume---without a decrease in minimum squared distance,
which still reads $d_\mathrm{min,\Hu}^2=1$. These points form the densest packing in
four-dimensional space \cite{Conway:99}, as will be discussed in
Sec.~\ref{subsec:generalized_lattice}. The maximum squared quantization error reads
$\epsilon^2_\Hu=|\frac{1}{4}+\frac{1}{4}\i+\frac{1}{4}\j+\frac{1}{4}\k|^2=\frac{1}{2}$,
i.e., for the division with remainder according to~\eqref{eq:division}, we have
$|\rho|\leq \frac{1}{\sqrt{2}} < |v|$,  with $|v|\geq 1~\forall v\in \Hu \setminus \{0\}$.
As a consequence, a Euclidean ring is present and a related Euclidean algorithm can be
applied. Thereby, similar to the Eisenstein integers in~\eqref{eq:quant_E}, the quantization
is performed as \cite{Conway:82,Conway:99}
\begin{align}
	\Q_{\Hu}\{q\} &= \argmin_{\Q_{\Hu_1}\{q\},\Q_{\Hu_2}\{q\}}
		\{|q-\Q_{\Hu_1}\{q\}|,|q -\Q_{\Hu_2}\{q\}|\} \; ,\\
	\Q_{\Hu_1}\{c\} &= \Q_{\L}\left\{q\right\}\; , \\\
	\Q_{\Hu_2}\{c\} &= \Q_{\L}\{q - o_\Hu\} + o_\Hu\; , \;  o_\Hu=\frac{1+\i+\j+\k}{2} \; .
\end{align}
Hence, both a quantization to all Lipschitz integers (filled circles in Fig.~\ref{fig:LHint}
(right)) and a quantization to all Lipschitz integers shifted by $\frac{1}{2}$ per component
(hollow circles) is done; subsequently, the closest result is chosen as the quantized value.
The modulo reduction is, again, described as $\mod_\Hu\{q\}=q-\Q_{\Hu}\{q\}$. The points
are reduced to the Voronoi cell located around the origin which forms a 24-cell with
24~vertices, 96~edges, and 96~faces \cite{Conway:99}.

\subsection{Generalized Definition of Lattices}	\label{subsec:generalized_lattice}

A lattice $\Lat$ is an \emph{infinite} set of points that are distributed over the Euclidean
space in such a way that an \emph{Abelian group} w.r.t.\ addition is present
\cite{Conway:99,Fischer:02,Zamir:14}. Most often, lattices are considered over the real
numbers. Then, a discrete (infinite) additive subgroup of $\R^N$ (lattice with
\emph{dimension} $N$), or equivalently the set of all integer linear combinations of a set
of $K$ independent vectors (\emph{lattice basis} with \emph{rank} $K$; assembled in the
\emph{generator matrix} $\ve{G}\in\R^{N \times K}$) is obtained.

In a more general setting, an $N$-dimensional lattice of rank $K$, with $N \geq K$, can be
defined by%
\begin{equation}	\label{eq:lattice}
	\Lat(\ve{G}) = \left\{\ve{G} \ve{\zeta} = 
		\sum\nolimits_{k=1}^{K} \ve{g}_k \zeta_k \mid \zeta_k \in \I\right\} \; ,
\end{equation}
where $\ve{G}=[\ve{g}_1,\dots,\ve{g}_K]$ denotes the $N\times K$ generator matrix of the
lattice, and $\ve{\zeta}=[\zeta_1,\dots,\zeta_K]^\T$ an \emph{integer vector} that describes
linear combinations of the basis vectors based on elements drawn from the constituent integer
ring $\I$.

In the real-valued case ($\ve{G}\in\R^{N \times K}$), those linear combinations are expressed
by real integers. The generalized definition of lattices in~\eqref{eq:lattice} enables the
construction of lattices over complex numbers, where $\ve{G}\in\C^{N \times K}$. Here, both
$\I=\G$ and $\I=\E$ are suited to express linear combinations. Moreover, lattices over
quaternions, with $\ve{G}=\Ha^{N \times K}$, can be defined based on the Lipschitz integers
($\I=\L$) or the Hurwitz integers ($\I=\Hu$) as the constituent integer ring.

The basis of a lattice is not unique. Instead, there exists an infinite number of generator
matrices that span the same lattice. The transformation to an alternative generator matrix
can be realized by the multiplication with a \emph{unimodular} integer matrix
$\ve{T}\in\I^{K \times K}$ according to%
\begin{equation}	\label{eq:lattice_reduction}
	\ve{G}_\mathrm{red} = \ve{G}\ve{T} \; , \qquad \textrm{with }
	{\det(\ve{T}^\H\ve{T})}=1 \; .
\end{equation}
A generator matrix is called \emph{reduced} matrix if, by means of
\emph{lattice-reduction algorithms}, this matrix is constructed in such a way that it
fulfills some desired quality criteria. Noteworthy, if the unimodularity constraint is
relaxed to the full-rank constraint $\rank(\ve{T})=K$, \emph{sublattices} of the original
lattice with the order $\sqrt{\det(\ve{T}^\H\ve{T}})$ may be obtained
\cite{Stern:19,Fischer:19}.

To evaluate the quality of a lattice basis, the lengths~(norms) of the basis vectors
$\ve{g}_1,\dots,\ve{g}_K$ are often assessed. The~Euclidean norm of a vector
$\ve{v}$ over $\R$, $\C$, or $\Ha$ is given as $\|\ve{v}\| = \sqrt{\ve{v}^\H\ve{v}}$.
Another criterion is the \emph{orthogonality defect}~\cite{Conway:99,Fischer:02}%
\begin{equation}	\label{eq:orth_defect}
	\Omega(\ve{G}) = \frac{\prod_{k=1}^{K}\|\ve{g}_k\|}{\vol(\Lat(\ve{G}))}
\end{equation}
of a lattice basis which is given as the product of the basis vectors' norms over the
volume of the lattice%
\begin{equation}	\label{eq:volume}
	\vol(\Lat(\ve{G}))=\sqrt{\det(\ve{G}^\H\ve{G})}
\end{equation}
given as the volume of a related $K$-dimensional fundamental parallelotope. Thereby,
the volume is the same for all generator matrices that span the same lattice
\cite{Conway:99,Fischer:02}.

\subsubsection{Real Representation of Complex Lattices}

The Gaussian integers are isomorphic to the two-dimensional real-valued integer
lattice $\Z^2$ (generator matrix\footnote{%
	 $\ve{I}_K$ denotes the $K \times K$ identity matrix
	 and $\ve{0}_K$ the $K \times K$ all-zero matrix
	 (all elements are $0$).%
}
$\ve{G}=\ve{I}_2$), cf.~\cite{Conway:99}. Hence, a complex lattice with generator
matrix $\ve{G}\in\C^{N \times K}$ defined over $\G$ is isomorphically expressed by a
real lattice (over $\Z$) with the generator matrix
$\ve{G}_\mathsf{r,\G}=\ve{G}_\mathsf{r}\ve{I}_{2K}=\ve{G}_\mathsf{r}$,
cf.~\eqref{eq:equiv_real_compl}.

The Eisenstein integers are isomorphic to the real hexagonal lattice $\ve{A}_2$. Thus,
a real lattice representation is obtained via~\cite{Stern:19}%
\begin{equation}	\label{eq_equiv_real_E}
			\ve{G}_\mathsf{r,\E} = 
				\underbrace{\ve{G}_\mathsf{r}}_{\text{\eqref{eq:equiv_real_compl}}}
				\underbrace{\begin{bmatrix}
					\ve{I}_K & -\frac{1}{2}\ve{I}_K\\
					\ve{0}_K & \frac{\sqrt{3}}{2} \ve{I}_K
				\end{bmatrix} }_{\ve{G}_\E}
\end{equation}
where, for $K=1$, the right-hand-side matrix in~\eqref{eq_equiv_real_E} represents a
generator matrix of the $\ve{A}_2$ lattice \cite{Conway:99}.

\subsubsection{Real and Complex Representation of Quaternion-Valued Lattices}

The Lipschitz integers are isomorphic to the four-dimensional (real-valued) integer
lattice $\Z^4$ (generator matrix $\ve{I}_4$). Quaternionic lattices over $\L$ can
equivalently be expressed by complex-valued lattices over $\G$ via the generator matrix
$\ve{G}_\mathsf{c,\L}=\ve{G}_\mathsf{c}\ve{I}_{2K}= \ve{G}_\mathsf{c}\in\C^{2N \times 2K}$
according to~\eqref{eq:equiv_compl_quat}, or by real-valued lattices (over $\Z$) with the
corresponding generator matrix
$\ve{G}_\mathsf{r,\L}=\ve{G}_\mathsf{r}\ve{I}_{4K}= \ve{G}_\mathsf{r}\in\R^{4N \times 4K}$
from~\eqref{eq:equiv_real_quat}.

Regarding the Hurwitz integers, an isomorphism to the four-dimensional
\emph{checkerboard lattice} $\ve{D}_4$ is present \cite{Conway:99}. This isomorphism can be
exploited in order to define an equivalent real-valued representation (over $\Z$) for
lattices over $\Hu$. Here, the equivalent real-valued generator matrix is obtained as%
\begin{equation}	\label{eq:equiv_real_hur}
		\ve{G}_\mathsf{r,\Hu}=
		\underbrace{\ve{G}_\mathsf{r}}_{\text{\eqref{eq:equiv_real_quat}}}
	\underbrace{\begin{bmatrix}
		 \ve{I}_K &\ve{0}_K &\ve{0}_K & \frac{1}{2}\ve{I}_K \\
		 \ve{0}_K &\ve{I}_K &\ve{0}_K & \frac{1}{2}\ve{I}_K \\
		 \ve{0}_K &\ve{0}_K &\ve{I}_K & \frac{1}{2}\ve{I}_K \\
		 \ve{0}_K &\ve{0}_K &\ve{0}_K & \frac{1}{2}\ve{I}_K \\
	\end{bmatrix}}_{\ve{G}_\Hu}
	\in \R^{4N \times 4K}
\end{equation}
by incorporating the generator matrix of the $\ve{D}_4$ lattice.\footnote{%
	In particular, the generator matrix of the lattice \emph{dual} to
	$\ve{D}_4$ is employed that actually corresponds to a version of
	the original $\ve{D}_4$ lattice which is scaled by a factor of
	$\frac{1}{2}$ \cite{Conway:99}. Using that strategy, the resulting
	points directly correspond to the set of Hurwitz integers (with
	half-integer values).%
}
In the same way, an equivalent complex-valued representation (over~$\G$) is realized
by using the generator matrix%
\begin{equation}	\label{eq:equiv_compl_hur}
	\ve{G}_\mathsf{c,\Hu}=
	\underbrace{\ve{G}_\mathsf{c}}_{\text{\eqref{eq:equiv_compl_quat}}}
	\begin{bmatrix}
		\ve{I}_K & (\frac{1}{2}+\frac{1}{2}\,\i)\ve{I}_K \\
		\ve{0}_K & (\frac{1}{2}+\frac{1}{2}\,\i)\ve{I}_K \\
	\end{bmatrix}
	\in \C^{2N \times 2K} \; .
\end{equation}

\section{Algorithms for Generalized Lattice Problems}	\label{sec:algorithms}

\noindent In this section, algorithms that are suited to (approximately) solve
particular lattice problems are reviewed and generalized.

\subsection{Lattice Basis Reduction and the LLL Algorithm}

The task of lattice basis reduction is to find a \emph{more suited basis}
$\ve{G}_\mathrm{red}=\big[\ve{g}_{\mathrm{red},1},\dots,\ve{g}_{\mathrm{red},K}\big]$
for the representation of the lattice spanned by the (unreduced) generator matrix
$\ve{G}$. For generalized lattices, a unimodular integer matrix
$\ve{T}=\big[\ve{t}_1,\dots,\ve{t}_K\big]\in\I^{K \times K}$ according
to~\eqref{eq:lattice_reduction} has to be found in such a way that particular quality
criteria are fulfilled.

To assess when a lattice basis is reduced, several different criteria can be defined.
One is the minimization of the orthogonality defect~\eqref{eq:orth_defect}, i.e.,
the optimality criterion reads%
\begin{equation}	\label{eq:minOrth}
	\ve{T} = \argmin_{\substack{\ve{T}\in\I^{K \times K}\\ 
			\det(\ve{T}^\H \ve{T})=1}} \!\Omega (\ve{G}\ve{T})
		   = \argmin_{\substack{\ve{T}\in\I^{K \times K}\\ \det(\ve{T}^\H \ve{T})=1}} 
		   \!\frac{\prod_{k=1}^{K}\|\ve{g}_{\mathrm{red},k}\|}{\vol(\Lat(\ve{G}))} \; .
\end{equation}
Hence, the norms of all basis vectors are incorporated. An alternative approach is to
consider the lengths of these vectors individually, e.g., $\|\ve{g}_{\mathrm{red},1}\|$.
The minimization of the \emph{maximum squared norm} among the basis vectors is particularly
known under the name \emph{shortest basis problem} (SBP) and described~by%
\begin{equation}	\label{eq:sbp}
	\ve{T} = \argmin_{\substack{\ve{T}\in\I^{K \times K}\\ 
			\det(\ve{T}^\H \ve{T})=1}} \max_{k=1,\dots,K} \|\ve{G}\ve{t}_k\|^2 \; .
\end{equation} 

The LLL algorithm \cite{Lenstra:82} is suboptimal w.r.t.\ the above
lattice-basis-reduction criteria, i.e., it only approximates these optimization
problems. Nevertheless, a polynomial asymptotic complexity is ensured (over $K$),
whereas for the exact solutions to \eqref{eq:minOrth} and \eqref{eq:sbp} an exponential
complexity is required, cf., e.g., \cite{Daude:94,Agrell:02}. Moreover, it is possible
to derive certain performance guarantees (i.e., bounds) for LLL reduction. More details
will be provided in Sec.~\ref{sec:bounds}. In the following, the concept of LLL
reduction will be generalized.

\subsubsection{Gram--Schmidt Orthogonalization}

The LLL reduction and its related criteria operate on the Gram--Schmidt
orthogonalization (GSO) \cite{Press:07} of the generator matrix\footnote{%
	The matrices $\ve{Q}$ and $\ve{R}$ often describe the \emph{QR decomposition}
	of a matrix, in which $\ve{Q}$ is usually assumed to be a unitary matrix
	($\ve{Q}^\H\ve{Q}=\ve{I}$). Given the GSO (without normalization), the column
	norms of $\ve{Q}$ are then absorbed in $\ve{R}$ (non-unit main diagonal). In
	this work, we assume that the matrix $\ve{Q}$ does not have to be unitary
	but that $\ve{R}$ is a matrix with unit main diagonal.%
}%
\begin{equation}
	\ve{Q}\ve{R} = \ve{G} \ve{P}\; .
\end{equation}
In particular, $\ve{Q}=\big[\ve{q}_1,\dots,\ve{q}_K\big]$ forms an $N \times K$ matrix
with \emph{orthogonal columns}, and
$\ve{R}=\big[r_{l,k}\big]$, $l=1,\dots,K$, $k=1,\dots,K$,
an upper triangular $K \times K$ matrix with unit main diagonal
($r_{k,k}=1$, $k=1,\dots,K$). The $K \times K $ matrix $\ve{P}$ (\emph{permutation matrix}
with a single $1$ per column and row and all other elements equal to $0$) can be used to
sort the Gram--Schmidt vectors in $\ve{Q}$ according to their length during the
orthogonalization process (known as \emph{pivoting}).

The GSO procedure\footnote{%
	To simplify the notation, we denote the selection of the elements with index
	$k,\dots,l$ of a vector $\ve{g}$ by $\ve{g}_{k:l}$. Given a matrix $\ve{G}$,
	the selection of a submatrix composed of the rows $k,\dots,l$ and the columns
	$m,\dots,n$ is denoted as $\ve{G}_{k:l,m:n}$.%
}
is presented in Algorithm~\ref{alg:gso}. In every step $k=1,\dots,K$, the pivoting is
performed first, i.e., the shortest of the remaining columns $\ve{q}_k,\dots,\ve{q}_K$
is inserted at position $k$ (and removed at its original position). Afterwards, the
remaining columns are projected onto the orthogonal complement of $\ve{q}_k$. In
contrast to other implementations, e.g., \cite{Press:07,Gan:09}, all multiplications are
defined in such a way that a non-commutative behavior, e.g., for the case of quaternion
numbers, is taken into account. Hence, the procedure can be used for real-, complex-, or
quaternion-valued numbers.

\subsubsection{Generalized LLL Reduction Criteria}

In the initial publication on real-valued LLL (RLLL) reduction, the matrices $\ve{R}$ and
$\ve{Q}$ have been evaluated w.r.t.\ two conditions:
i) the property $|r_{l,k}| \leq \frac{1}{2}$, $1 \leq l < k \leq K$, has to be fulfilled
for the upper triangular part of $\ve{R}$, known as \emph{size-reduction} condition, and
ii) postulating a size-reduced matrix $\ve{R}$, the Lov{\'a}sz condition
$\|\ve{q}_k\|^2 \geq (\delta - |r_{k-1,k}|^2) \cdot \|\ve{q}_{k-1}\|^2$ has to hold for
the norms of the columns of $\ve{Q}$. The parameter $\delta$ is used as a quality
parameter and enables a trade-off between reduction quality and the runtime of the
algorithm.

Taking a closer look at the size-reduction condition, it becomes apparent that the LLL
algorithm \cite{Lenstra:82} can be interpreted to form some kind of Euclidean algorithm
for matrices, see also~\cite{Napias:96}. In particular, this condition is forced by a
modulo reduction according to~\eqref{eq:mod_Z}, i.e., by a (Euclidean) division as defined
in~\eqref{eq:division} with the divisor $v=1$, where the resulting remainder
$r_{l,k}=\rho$ is a \emph{small remainder} since $|r_{l,k}|\leq \frac{1}{2} < |v|=1$.
The size-reduction condition is alternatively expressed by $\Q_\Z\{r_{l,k}\}=0$, i.e., all
non-zero, non-diagonal values in $\ve{R}$ have to be located within the Voronoi cell of
$\Z$ w.r.t.\ the origin.

Taking advantage of the interpretation that the size-reduction operation forms a modulo
reduction over the particular Euclidean ring and the fact that the Lov{\'a}sz condition is
generally valid for Euclidean rings as it compares norms of vectors,
\emph{generalized LLL reduction criteria} can be defined.
\begin{dfn}[Generalized LLL Reduction]	\label{eq:generalLLL}
	A generator matrix $\ve{G}=\ve{Q}\ve{R}$ that spans a lattice over the Euclidean
	integer ring $\I$ is LLL-reduced over~$\I$, if
	\begin{enumerate}[label=\roman*)]
		\item $\ve{R}$ is size-reduced according to%
			\begin{equation}
			\Q_\I \{r_{l,k}\} = 0 \; , \qquad 1 \leq l < k \leq K \; ,
			\end{equation}
	\item the respective Lov{\'a}sz condition%
	\begin{equation}	\label{eq:lovasz_general}
		\|\ve{q}_k\|^2 \geq (\delta - |r_{k-1,k}|^2) \cdot \|\ve{q}_{k-1}\|^2
	\end{equation}
	is fulfilled. The quality parameter can be chosen according to%
	\begin{equation}
		\delta \in (\epsilon_{\I}^2,1] \; ,
	\end{equation}
	where $\epsilon_{\I}^2$ denotes the maximum squared quantization~error for elements
	quantized to the Euclidean integer ring~$\I$.
	\end{enumerate}
\end{dfn}
\begin{proof}
	Since, after size reduction, $\Q_\I\{r_{l,k}\}=0$ is valid, $r_{l,k}$ forms the 
	remainder $\rho$ of the division with the divisor $v=1$ as defined
	in~\eqref{eq:division}. If $\I$ is a Euclidean ring,
	$|r_{l,k}|\leq \epsilon_\I < |v| = 1$, i.e., a \emph{small remainder} is present.
	Hence, the Lov{\'a}sz condition becomes operative if
	$0<(\delta-\epsilon_{\I}^2)\leq 1$, i.e., if $\delta \in (\epsilon_{\I}^2,1]$.
\end{proof}
\begin{remark}
	For the integer lattice $\Z$, we have $\delta\in(\frac{1}{4},1]$ since
	$\epsilon_\Z^2=\frac{1}{4}$. This result is identical to the range initially derived
	in \cite{Lenstra:82,Akhavi:03}.	Given $\I=\G$, $\delta\in(\frac{1}{2},1]$ is valid as
	$\epsilon_{\G}^2=\frac{1}{2}$, cf.~\cite{Gan:09}. For lattices over $\E$,
	$\delta\in(\frac{1}{3},1]$ is obtained due to $\epsilon_{\E}^2=\frac{1}{3}$.
	Considering lattices over the Hurwitz integers $\Hu$, $\delta\in(\frac{1}{2},1]$ as
	$\epsilon_{\Hu}^2=\frac{1}{2}$. Since the Lipschitz integers do	not form a Euclidean
	ring, an LLL reduction over $\L$ can---in general---not be defined. This can be seen
	if the maximum squared quantization error $\epsilon_{\L}^2$ is inserted into the
	Lov{\'a}sz condition: even when $\delta=1$,
	$\delta - |r_{k-1,k}|^2 = 0$ if $|r_{k-1,k}|=1$, i.e., the reduction may become
	inoperative.
\end{remark}

A generalization of the size-reduction operation from \cite{Lenstra:82,Gan:09} is provided
in Algorithm~\ref{alg:size}---given the (reduced) basis $\ve{G}_\mathrm{red}$, its related
matrix $\ve{R}$, the transformation matrix $\ve{T}$, the indices $l$ and $k$, and the integer
ring $\I$ as input variables. For lattices over $\Z$ and $\G$, the operations are equivalent
to the ones defined in the RLLL algorithm \cite{Lenstra:82} and the complex LLL (CLLL)
algorithm~\cite{Gan:09}, respectively. However, the generalized definition also enables an
LLL reduction over the Eisenstein integers (ELLL algorithm), over the Hurwitz integers
(QLLL algorithm, ``Q'' for quaternions), and over other Euclidean rings.

Given complex and quaternion-valued matrices, the reduction can alternatively be performed
w.r.t.\ their equivalent real- and real-/complex-valued matrix representations, as defined
in \eqref{eq:equiv_real_compl}, \eqref{eq:equiv_real_quat}, and \eqref{eq:equiv_compl_quat}.
Then, the reduction has to be done with real-valued and/or complex-valued algorithms. However,
during the GSO, the particular structure of these matrices (and the isomorphism) is destroyed,
i.e., the resulting reduced basis and the related integer matrix cannot be reconverted into
equivalent complex/quaternion-valued representations, see also~\cite{Stern:19}. Consequently,
the quality of the reduction will not necessarily be the same. More details will be given in
Sec.~\ref{sec:bounds}.

\begin{algorithm}[t]{
		\small
		\caption{\label{alg:gso} Gram--Schmidt Orthogonalization with Pivoting.} %
		$[\ve{Q},\ve{R},\ve{P}] = \textsc{GSO}(\ve{G})$
		\begin{algorithmic}[1]
			\State {$\ve{Q}=\ve{G}$, $\ve{R}=\ve{I}_K$, $\ve{P}=\ve{I}_K$ }
			\For {$k=1,\dots,K$}
			\State {$k_\mathrm{m}=\argmin_{l=k,\dots,K} \| \ve{q}_l\|$}
			\If {$k_\mathrm{m} \neq k$}			{\color{\commentcolor}\Comment{pivoting}}
			\State {$\ve{Q} = [\ve{q}_{1:k-1},\ve{q}_{k_\mathrm{m}},
				\ve{q}_{k:k_\mathrm{m}-1},\ve{q}_{k_\mathrm{m}+1:K}]$}
			\State {$\ve{P} = [\ve{p}_{1:k-1},\ve{p}_{k_\mathrm{m}},
				\ve{p}_{k:k_\mathrm{m}-1},\ve{p}_{k_\mathrm{m}+1:K}]$}
			\State {In the upper $k-1$ rows of $\ve{R}$ (index $1:k-1$):}
			\Statex {\qquad\quad$\ve{R} = [\ve{r}_{1:k-1},\ve{r}_{k_\mathrm{m}},
				\ve{r}_{k:k_\mathrm{m}-1},\ve{r}_{k_\mathrm{m}+1:K}]$}
			\EndIf
			\For {$l=k+1,\dots,K$}	{\color{\commentcolor}\Comment{orthogonalization}}
			\State {$r_{k,l}=\ve{q}_k^\H \ve{q}_l \cdot\|\ve{q}_k\|^{-2}$}
			\State {$\ve{q}_l = \ve{q}_l - \ve{q}_k r_{k,l}$}
			\EndFor
			\EndFor
		\end{algorithmic}
	}
\end{algorithm}
\subsubsection{Generalized LLL Reduction Algorithm}

In Algorithm~\ref{alg:lll}, a generalized variant of the LLL algorithm is provided. In
contrast to other implementations, e.g., in \cite{Gan:09}, all multiplications are performed
in the right order to account for non-commutative behavior. In particular, in the first line,
a GSO with pivoting is calculated. The pivoting is not necessarily required, but it is
well-known that sorted Gram--Schmidt vectors may speed up the following reduction process
\cite{Fischer:10}. In the loop, the size reduction for the element $r_{k-1,k}$ (over the
particular ring $\I$) is done first. Then, the respective Lov{\'a}sz condition can be checked.
If it is not fulfilled, the columns $k-1$ and $k$ are \emph{swapped} in $\ve{G}_\mathrm{red}$
and the (unimodular) transformation matrix $\ve{T}$. As a consequence, the matrices $\ve{Q}$
and $\ve{R}$ have to be updated. This can either be done by a complete recalculation of the
GSO, or by the procedure in Algorithm~\ref{alg:qr} that restricts the recalculation to all
elements which have to be updated. It is an adapted variant of the procedure in \cite{Gan:09},
additionally taking the right order of all multiplications into account. If the check of the
Lov{\'a}sz condition is successful, the elements $r_{l,k}$, $l=k-2,k-3,\dots,1$, are finally
reduced and the algorithm continues with the next reduction step until $k=K$.

\begin{algorithm}[t]{
		\small
		\caption{\label{alg:size} Generalized Size Reduction.} %
		$[\ve{G}_\mathrm{red},\ve{R},\ve{T}] = 
		\textsc{SizeRed}(\ve{G}_\mathrm{red}, \ve{R}, \ve{T}, l, k, \I)$
		\begin{algorithmic}[1]
			\State {$r_\mathrm{q} = \Q_{\I}\{r_{l,k}\}$ }
			\If {$r_\mathrm{q}\neq 0$}
			\State {$\ve{g}_{\mathrm{red},k} = 
				\ve{g}_{\mathrm{red},k} - \ve{g}_{\mathrm{red},l} \,r_\mathrm{q}$}
			\State {$\ve{t}_{k} = \ve{t}_{k} - \ve{t}_{l}\, r_\mathrm{q}$}
			\State {In the upper $l$ rows of $\ve{R}$ (index $1:l$):}
			\Statex {\quad\,~$\ve{r}_{k} = \ve{r}_{k} - \ve{r}_{l} \,r_\mathrm{q}$}
			\EndIf
		\end{algorithmic}
	}
\end{algorithm}
\begin{algorithm}[t]{
		\small
		\caption{\label{alg:lll} Generalized LLL Reduction.} %
		$[\ve{G}_\mathrm{red},\ve{Q},\ve{R},\ve{T}] = 
			\textsc{LLL}(\ve{G}, \delta, \I)$
		\begin{algorithmic}[1]
			\State {$[\ve{Q},\ve{R},\ve{T}] = \textsc{GSO}(\ve{G})$}
				{\color{\commentcolor}\Comment{initial GSO with pivoting}}
			\State{$\ve{G}_\mathrm{red}=\ve{G}\ve{T}$, $k=2$}
			\While {$k\leq K$}
			\State {$[\ve{G}_\mathrm{red},\ve{R},\ve{T}] = 
				\textsc{SizeRed}(\ve{G}_\mathrm{red}, \ve{R}, \ve{T}, k-1, k, \I)$}
			\If{$ \|\ve{q}_k\|^2 < (\delta - |r_{k-1,k}|^2) \cdot \|\ve{q}_{k-1}\|^2$}
				{\color{\commentcolor}\Comment{swap}}
			\State {$\ve{G}_\mathrm{red}=[\ve{g}_{\mathrm{red},1:k-2},\ve{g}_{\mathrm{red},k},
				\ve{g}_{\mathrm{red},k-1},\ve{g}_{\mathrm{red},k+1:K}]$}
			\State {$\ve{T}=[\ve{t}_{1:k-2},\ve{t}_{k},	\ve{t}_{k-1},\ve{t}_{k+1:K}]$}
			\State {$[\ve{Q},\ve{R}] = \textsc{UpdateQR}(\ve{Q},\ve{R},k)$}
			\State{$k=\max(2,k-1)$}
			\Else	{\color{\commentcolor}\Comment{Lov{\'a}sz condition fulfilled}}
			\For {$l=k-2,k-3,\dots,1$}
			\State{$[\ve{G}_\mathrm{red},\ve{R},\ve{T}] = 
				\textsc{SizeRed}(\ve{G}_\mathrm{red}, \ve{R}, \ve{T}, l, k, \I)$}
			\EndFor
			\State {$k = k + 1$}
			\EndIf
			\EndWhile
		\end{algorithmic}
	}
\end{algorithm}
\begin{algorithm}[t]{
		\small
		\caption{\label{alg:qr} GSO update if columns $k$ and $k-1$ are swapped.} %
		$[\ve{Q},\ve{R}] = \textsc{UpdateQR}(\ve{Q},\ve{R},k)$
		\begin{algorithmic}[1]
			\State {$\ve{\tilde{q}}_{k-1}= \ve{q}_{k-1}$, 
				$\ve{\tilde{q}}_{k}= \ve{q}_k$, $\ve{\tilde{R}}=\ve{R}$ }
					{\color{\commentcolor}\Comment{temporal variables}}
			\State {$\ve{{q}}_{k-1}= \ve{q}_k + \ve{q}_{k-1} r_{k-1,k}$}
			\State {$\ve{{r}}_{k-1,k}= {r}_{k-1,k}^* \cdot
				\|\ve{\tilde{q}}_{k-1}\|^2 \cdot \|\ve{q}_{k-1}\|^{-2}$}
			\State {$\ve{q}_k = \ve{\tilde{q}}_{k-1} - \ve{q}_{k-1} r_{k-1,k}$}
			\For{$l=k+1,\dots,K$}
			\State{$r_{k-1,l}=r_{k-1,k} r_{k-1,l} + 
				r_{k,l} \cdot \|\ve{\tilde{q}}_{k}\|^2 \cdot \| \ve{q}_{k-1}\|^{-2}$}
			\State{$r_{k,l}=\tilde{r}_{k-1,l} - \tilde{r}_{k-1,k} r_{k,l}$}
			\EndFor
			\For{$l=1,\dots,k-2$}
			\State{$r_{l,k-1}=r_{l,k}$}
			\State{$r_{l,k}=\tilde{r}_{l,k-1}$}
			\EndFor
		\end{algorithmic}
	}
\end{algorithm}
\subsubsection{Pseudo-QLLL Reduction}

Even though an LLL reduction over the Lipschitz integers can---due to the non-existent
Euclidean property---not be defined in general, a \emph{pseudo-QLLL reduction} can be
defined instead. In particular, the reduction only becomes inoperative if $|r_{k-1,k}|=1$.
Hence, choosing $\delta=1$, the LLL algorithm can be applied if the probability that
$|r_{k-1,k}|=1$ tends to zero.

Such a case is, e.g., present if the elements of the generator matrix are drawn i.i.d.ly
from a continuous distribution (e.g., an i.i.d. Gaussian one). Then, if the quantization
in the size-reduction steps is performed w.r.t.\ $\I=\L$ and if additionally the parameter
$\delta=1$ is chosen, a pseudo-QLLL-reduced basis is obtained. Obviously, no general
performance guarantees or bounds can be derived in that case---however, it is at least
ensured that the Lov{\'a}sz condition is still operative and that, thus, some kind of
``optimized'' basis is produced.

\subsection{Shortest Independent Vectors in Lattices}	\label{subsec:smp}

Another important lattice problem is the determination of the shortest linearly independent
vectors in lattices. They are related to the so-called \emph{successive minima}. In
particular, the $k$\textsuperscript{th} successive minimum of a lattice with $N\times K$
generator matrix $\ve{G}$, $k=1,\dots,K$, is defined as \cite{Minkowski:91,Lagarias:90}%
\begin{equation}
	\mu_{k} = \inf \{\mu \mid \dim \{\mathrm{span}
		\{\Lat(\ve{G})\cap \mathrm{B}_N(\mu)\}\}=k\} \; ,
\end{equation}
where $\mathrm{B}_N(\mu)$ denotes the $N$-dimensional ball with hyperradius $\mu$ centered
at the origin. In words, $\mu_{k}$ denotes the smallest radius in which $k$ linearly
independent vectors can be found within the hyperball. The related lattice points with%
\begin{equation}	
	\mu_{k} = \|\ve{\lambda}_{\mathrm{m},k}\| \; ,\qquad k=1,\dots,K \;,
\end{equation}
form the $K$ linearly independent vectors with the shortest (Euclidean) norms.

Closely related to the \emph{successive-minima problem} (SMP)---the determination of the $K$
shortest independent vectors in a lattice---is the \emph{shortest independent vector problem}
(SIVP). Here, only the maximum of the (squared) norms has to be short as possible, cf.\ the SBP
in~\eqref{eq:sbp}. Hence, the (generalized) SIVP is a weakened variant of the SMP and
defined as%
\begin{equation}	\label{eq:sivp}
		\ve{T} = \argmin_{\substack{\ve{T}\in\I^{K \times K}\\ 
				\rank(\ve{T})=K}} \max_{k=1,\dots,K} \|\ve{G}\ve{t}_k\|^2 \; .
\end{equation}
Obviously, every optimal solution for the successive-minima problem is also optimal w.r.t.\ the
SIVP \cite{Minkowski:89,Ding:15}.

As becomes apparent from~\eqref{eq:sivp}, the $K$ independent lattice vectors do not necessarily
form a basis of the lattice spanned by $\ve{G}$. In particular, if the $K$ integer vectors
$\ve{t}_k$ are combined into the integer transformation matrix
$\ve{T}=\big[\ve{t}_1,\dots,\ve{t}_K\big]$, the transformed generator matrix
$\ve{G}_\mathrm{tra}=\ve{G}\ve{T}$ may only define a \emph{sublattice} of the original one: the
integer vectors $\ve{t}_k$ are only required to be linearly independent, but they do not have
form a unimodular transformation matrix in combination. Hence, $\ve{T}$ may only have full rank,
resulting in a ``thinned'' lattice when multiplied by $\ve{G}$, depending on the particular
determinant $\det(\ve{T})$. Further details can, e.g., be found in \cite{Stern:19,Fischer:19}.
Consequently, the successive minima can be used as lower bounds for the norms of the basis
vectors of any (alternative) basis for a lattice spanned by a particular generator
matrix~$\ve{G}$.

\subsubsection{Generalized Successive Minima}

Even though the successive minima have initially been considered for real-valued lattices over
$\Z$, e.g., in \cite{Minkowski:91,Lagarias:90}, it is quite obvious that they can also be
determined if other integer rings~$\I$ are present. The only condition is that the shortest
$K$ linearly independent  lattice vectors are found---no additional constraints as, e.g., a
Euclidean property of the particular ring are imposed.

In contrast to LLL reduction, the isomorphism of a complex or quaternionic matrix $\ve{G}$
and its equivalent real-valued representation $\ve{G}_\mathsf{r}$ according
to \eqref{eq:equiv_real_compl} and \eqref{eq:equiv_real_quat}, respectively, holds for the
successive minima as shown in the following Theorem.
\begin{thm}[Successive Minima of Complex and Quaternionic Lattices]	\label{th:smp}
	Given an $N \times K$ generator matrix $\ve{G}$ for which
	the successive minima over $\I$ are given as%
	\begin{equation}	\label{eq:suc_minima}
		\ve{\mu}_\I=\big[\mu_1,\mu_{2},\dots,\mu_{K}] \in\R^{K}\; ,
	\end{equation}
	the successive minima of the equivalent
	$D_\mathsf{r} N \times D_\mathsf{r} K$
	real-valued generator matrix
		$\ve{G}_{\mathsf{r},\I} = \ve{G}_\mathsf{r} \ve{G}_\I$,
	where $\ve{G}_\mathsf{r}$ denotes the $D_\mathsf{r} N \times D_\mathsf{r} K$
	real-valued representation of
	$\ve{G}$, $\ve{G}_\I$ the $D_\mathsf{r} K \times D_\mathsf{r} K$
	real-valued generator matrix of the integer ring~$\I$, and
	$D_\mathsf{r}$ the number of (real-valued) components per scalar,
	cf.~\eqref{eq:equiv_real_compl} and \eqref{eq:equiv_real_quat},
	are obtained as%
	\begin{equation}	
		\ve{\mu}_{\mathsf{r},\I}=\big[
			\underbrace{\mu_1,\dots,\mu_{1}}_{D_\mathsf{r}\text{~times}},
			\underbrace{\mu_{2},\dots,\mu_{2}}_{D_\mathsf{r}\text{~times}},
			\dots,\underbrace{\mu_{K},\dots,\mu_{K}}_{D_\mathsf{r}\text{~times}}\big]
			\in \R^{D_\mathrm{r}K}\; ,
	\end{equation}
	i.e., each successive minimum from~\eqref{eq:suc_minima} occurs $D_\mathsf{r}$ times. This
	property is valid for complex lattices ($D_\mathsf{r}=2$) over both Gaussian and Eisenstein
	integers, as well as for quaternionic lattices ($D_\mathsf{r}=4$) over both Lipschitz and
	Hurwitz integers.
\end{thm}
\begin{proof}
	We start with the complex case ($D_\mathsf{r}=2$). Given $\ve{G}_\mathsf{r}$ according
	to~\eqref{eq:equiv_real_compl}, pairs of orthogonal (and, thus, linearly independent) column
	vectors occur at the indices $k$ and $k+K$, which additionally posses the norms
	$\|\ve{g}_{\mathsf{r},k}\|=\|\ve{g}_{\mathsf{r},k+K}\|=\|\ve{g}_{k}\|$, $k=1,\dots,K$.
	The same is obviously valid for the matrix
	$\ve{G}_{\mathsf{r,\G}}=\ve{G}_\mathsf{r}\ve{I}_{2K}$, but also for the matrix
	$\ve{G}_{\mathsf{r,\E}}=\ve{G}_\mathsf{r}\ve{G}_{\E}$: Since $\ve{G}_{\E}$
	from~\eqref{eq_equiv_real_E} is a non-singular matrix,
	$\rank(\ve{G}_{\mathsf{r,\E}})=\rank(\ve{G}_{\mathsf{r}})$, i.e., the columns of
	$\ve{G}_{\mathsf{r,\E}}$ are still independent. Concerning its column norms, the left half
	(index $1,\dots,K$) does not change in comparison to $\ve{G}_{\mathsf{r}}$ due to the identity
	part in $\ve{G}_{\E}$. Regarding the right half of the columns ($K+1,\dots,2K$), we obtain the
	squared norms
	$\|\ve{g}_{\mathsf{r},\E,k}\|^2=\frac{1}{4}\|\ve{g}_{k}^{(1)}\|^2
		+\frac{3}{4}\|\ve{g}_{k}^{(2)}\|^2+\frac{1}{4}\|\ve{g}_{k}^{(2)}\|^2
		+\frac{3}{4}\|\ve{g}_{k}^{(1)}\|^2 =
		\|\ve{g}_{k}^{(1)}+\ve{g}_{k}^{(2)}\|^2=\|\ve{g}_{k}\|^2$,
	i.e., these norms do not change either.\footnote{%
			Interpretation: The right part in~\eqref{eq_equiv_real_E} represents the
			\emph{Eisenstein unit}
			$\omega=\mathrm{e}^{\frac{2\pi}{3}\mathrm{i}}=-\frac{1}{2} + \frac{\sqrt{3}}{2}\i$
			with $|\omega|=1$, i.e., only a rotation in the complex plane is performed. The left
			(identity) part corresponds to the $0$\textsuperscript{th} root of unity
			$\mathrm{e}^0=1$.%
	}
	Hence, each lattice vector in $\Lat(\ve{G}_\mathsf{r,\I})$ has a linearly independent
	counterpart with the same length; both of them isomorphically represent one lattice vector
	of $\Lat(\ve{G})$ (with the same norm).	As a consequence, for $\ve{G}_\mathsf{r,\I}$, pairs of
	linearly independent lattice vectors
	$\ve{\lambda}_{\mathrm{m},k}$ and $\ve{\lambda}_{\mathrm{m},k+1}$ are obtained that yield
	successive minima with the same value, i.e., $\mu_{k}=\mu_{k+1}$, $k=1,3,\dots,2K-1$.
	
	For quaternion-valued matrices ($D_\mathsf{r}=4$), the proof is similar:
	In $\ve{G}_\mathsf{r}$ according to~\eqref{eq:equiv_real_quat},	orthogonal vectors with
	$\|\ve{g}_{\mathsf{r},k}\|=\|\ve{g}_{\mathsf{r},k+K}\|=\|\ve{g}_{\mathsf{r},k+2K}\|
		=\|\ve{g}_{\mathsf{r},k+3K}\|=\|\ve{g}_{k}\|$
	are present. Given the ring of Lipschitz integers, the same holds for the matrix
	$\ve{G}_\mathsf{r,\L}=\ve{G}_\mathsf{r}\ve{I}_{4K}$. Considering lattices over Hurwitz
	integers, we have $\ve{G}_\mathsf{r,\Hu}=\ve{G}_\mathsf{r}\ve{G}_{\Hu}$, where
	$\rank(\ve{G}_{\mathsf{r,\Hu}})=\rank(\ve{G}_{\mathsf{r}})$. The non-singular matrix
	$\ve{G}_{\Hu}$ defined in~\eqref{eq:equiv_real_hur} consists of an identity part for the
	column indices $1,\dots,3K$. For the rightmost part ($3K+1,\dots,4K$), we obtain the squared
	norms
	$\|\ve{g}_{\mathsf{r},\Hu,k}\|^2=4\cdot \frac{1}{4}\cdot
		\|\ve{g}_k^{(1)}+\ve{g}_k^{(2)}+\ve{g}_k^{(3)}+\ve{g}_k^{(4)}\|^2 =\|\ve{g}_k\|^2$,
	i.e., they remain the same.\footnote{%
		Interpretation: The rightmost part in~\eqref{eq:equiv_real_hur} represents the
		\emph{Hurwitz unit} $\frac{1}{2}(1+\i+\j+\k)$ with $|\frac{1}{2}(1+\i+\j+\k)|=1$, i.e.,
		only a rotation in quaternionic space is performed. The left (identity) part corresponds
		to the units $1$, $\i$, and $\j$.%
	}
	Hence, $\mu_{k}=\mu_{k+1}=\mu_{k+2}=\mu_{k+3}$, $k=1,5,\dots,4K-3$, is obtained.
\end{proof}
\begin{remark}
	Given lattices over the Gaussian integers, this isomorphism has been recognized	earlier, e.g.,
	in~\cite{Ding:15}. To the best of our knowledge, all other integer rings have not been
	considered before.
	
	In the quaternion-valued case, an isomorphism based on the equivalent complex-valued
	representation~\eqref{eq:equiv_compl_quat} can be derived in an analogous way. Then, the
	successive minima are repeated two times in the complex description.
\end{remark}

Due to Theorem~\ref{th:smp}, the successive minima and the related lattice points and integer
vectors can isomorphically be determined by solving the problem for the equivalent real-valued
representation. Hence, a permutation of the real or complex integer vectors can be found such
that an integer transformation matrix $\ve{T}_\mathsf{r}$ is formed that possesses the particular
structure defined in~\eqref{eq:equiv_real_compl} or \eqref{eq:equiv_real_quat}. Then, this matrix
can be reconverted to a complex or quaternion-valued representation, see also
\cite[Example~4.3]{Stern:19}. Both the ``direct'' determination of the successive minima over
$\C$ or $\Ha$, and their ``indirect'' determination over $\R$, finally lead to the same result.
Please note that, given lattices over Eisenstein or Hurwitz integers, a reconversion is required
at the end to compensate for the matrices $\ve{G}_\E$ and $\ve{G}_{\Hu}$ in~\eqref{eq_equiv_real_E}
and~\eqref{eq:equiv_real_hur}, respectively. This reconversion process will be described below
(and, additionally, in Appendix~\ref{app:helperSMP}).

\subsection{Generalized Determination of the Successive Minima}

For the determination of the successive minima, a special $K$-dimensional variant of the
shortest-vector problem has to be solved. It is well-known that already the classical
shortest-vector problem is NP-hard, i.e., that the computational complexity grows exponentially
with the dimension $K$. Nevertheless, at least for small dimensions, algorithms have been
proposed which can efficiently solve the shortest-vector problem. The most prominent one is the
so-called \emph{sphere decoder} \cite{Agrell:02}.

To solve the SMP, several algorithms have been proposed within the last
few years \cite{Ding:15,Fischer:16,Wen:19}. One possible strategy applied in \cite{Ding:15} is
to solve the shortest-vector problem $K$ times. Another strategy employed in
\cite{Fischer:16,Wen:19} is to solve an adapted variant thereof once in the beginning, afterwards
only operating on the initial result.

In this work, we address the list-based approach initially proposed in \cite{Fischer:16} which
has been shown to perform well w.r.t.\ runtime behavior if moderate dimensions are present
($K < 20$ in the real-valued case), see also the comparison in \cite{Wen:19}. In particular, in
that approach, the  sphere decoder \cite{Agrell:02} is initially applied to generate a list that
contains all lattice points within a hyperball of a predefined search radius. Then, among those
candidates, the shortest linearly independent ones are selected. 

\begin{algorithm}[tb]{
		\small
		\caption{\label{alg:smp} List-Based Determination of Successive Minima.} %
		$[\ve{G}_{\mathrm{tra}},\ve{T}] = \textsc{SMP}(\ve{G},\I)$
		\begin{algorithmic}[1]
			\State {$\ve{G}_{\mathsf{r},\I}=\textsc{RingToZ}(\ve{G},\I)$ } 
				{\color{\commentcolor}\Comment{real-valued representation}}
			\State {$\big[\ve{G}_\mathrm{red},\ve{T}_{\mathrm{LLL}}\big]=
				\textsc{LLL}(\ve{G}_{\mathsf{r},\I},1,\Z)$ }
				{\color{\commentcolor}\Comment{reduced basis}}
			\State {$\ve{C}_\mathrm{t}=\textsc{ListSphereDecoder}(
				\ve{G}_\mathrm{red},\max_k \|\ve{g}_{\mathrm{red},k}\|^2)$}
			\State {$\ve{C}=\ve{T}_{\mathrm{LLL}}\ve{C}_\mathrm{t}$}
				{\color{\commentcolor}\Comment{convert to original basis}}
			\State {$\ve{C}_\mathrm{u}=\textsc{ZToRing}(\ve{C},\I)$}
				{\color{\commentcolor}\Comment{go back to ring $\I$}}
			\State {$\ve{C}_\mathrm{s}=\textsc{Sort}(\ve{C}_\mathrm{u},
				\ve{G}\ve{C}_\mathrm{u})$}
				{\color{\commentcolor}\Comment{sort w.r.t.\ norm}}
			\State {$\ve{i}=\textsc{RowEchelon}(\ve{C}_\mathrm{s})$}
				{\color{\commentcolor}\Comment{indices of row-echelon form}}
			\State {$\ve{T}=\big[\ve{c}_{i_1},\dots,\ve{c}_{i_K}\big]$}
				{\color{\commentcolor}\Comment{shortest independent vectors}}
			\State {$\ve{G}_\mathrm{tra}=\ve{G}\ve{T}$}
				{\color{\commentcolor}\Comment{transformed generator matrix}}
		\end{algorithmic}
	}
\end{algorithm}

The generalized concept is provided in Algorithm~\ref{alg:smp}. As mentioned above, a
list-based variant of the sphere decoder is initially applied. As the sphere decoder
\cite{Agrell:02} in combination with Schnorr--Euchner enumeration \cite{Schnorr:94} only
operates over real numbers, the generator matrix $\ve{G}$ has to be converted to its
equivalent real-valued representation (given the integer ring $\I$). This is done with the
procedure \mbox{\textsc{RingToZ}} which is listed in Algorithm~\ref{alg:toZ} in
Appendix~\ref{app:helperSMP}. For the list sphere decoder, an initial search radius has to
be provided. The naive approach would be to use the maximum (squared) column norm of
$\ve{G}_{\mathsf{r},\mathbb{I}}$ as the (squared) radius, since the (squared) solution to the
SMP can never be worse than that value. However, the search radius (and the related
complexity) can significantly be decreased by applying a reduced basis via the LLL algorithm
(Line~2) instead, using the quality parameter $\delta=1$. Since this call has a low complexity
\cite{Akhavi:03}, it is negligible in comparison to the call of the list sphere decoder,
which is subsequently applied \cite[Algorithm~\textsc{AllClosest\-Points}]{Agrell:02}.
It results in a matrix of integer candidate vectors
$\ve{C}_\mathrm{t}\in\Z^{D_\mathsf{r} K \times N_\mathrm{c}}$, where $N_\mathrm{c}$ denotes
the list size. A list of respective candidate vectors w.r.t.\ the unreduced basis
$\ve{G}_{\mathsf{r},\mathbb{I}}$ is subsequently calculated by the multiplication with
$\ve{T}_\mathrm{LLL}$ (Line~4). The integer candidate vectors over the original ring $\I$ are
finally obtained by the procedure \textsc{ZToRing} which is listed in Algorithm~\ref{alg:fromZ}
in Appendix~\ref{app:helperSMP}. In particular, in that procedure, the original complex- or
quaternion-valued representations are reconstructed from the real-valued ones. In Line~6 of
Algorithm~\ref{alg:smp}, the candidate vectors are then sorted in ascending order w.r.t.\ their
norms. This can be done by a standard sorting algorithm \cite{Press:07}. In the procedure
\textsc{RowEchelon}, which is listed in Appendix~\ref{app:helperSMP} (Algorithm~\ref{alg:re}),
the matrix of sorted candidate vectors is transformed to row-echelon form. At the particular
indices $\ve{i}=\big[i_1,\dots,i_K\big]$ where a new dimension is established (``steps'' in the
row-echelon form), the vector resulting in the $k$\textsuperscript{th} successive minimum is
found as it leads to the shortest vector that is independent from the previous $k-1$ ones.
The related transformation matrix $\ve{T}$ is formed in Line~8, and the respective transformed
generator matrix $\ve{G}_\mathrm{tra}$ finally in Line~9.

The reconversion to the original ring $\I$ (as performed in Line~5 in Algorithm~\ref{alg:smp})
can alternatively be performed \emph{after} the calculation of the matrices $\ve{T}$ and
$\ve{G}_\mathrm{tra}$ (defined over $\Z$ and $\R$) takes place. However, then, for complex-
and quaternion-valued lattices, the calculation of the row-echelon form has to be performed
over real numbers with the equivalent $2K\times 2N_\mathrm{c}$ and $4K\times 4N_\mathrm{c}$
representation $\ve{C}$, respectively.

\section{Generalized Quality Bounds and Asymptotic Computational Complexity}	\label{sec:bounds}

\noindent Based on the generalized criteria and algorithms discussed previously, quality bounds
are derived and compared to each other in this section. In addition, the asymptotic complexity of
both above-mentioned (generalized) algorithms is evaluated.

\subsection{Bounds on the Norms}

The norms of the basis vectors are suited quantities to assess the quality of a lattice basis.
Since the respective successive minima are given as the norms of the shortest independent vectors
in the particular lattice, they serve as lower bounds on the norms resulting from any
lattice-basis-reduction scheme. 

To solve the SIVP~\eqref{eq:sivp} and the SBP~\eqref{eq:sbp}, respectively, it is required that
the \emph{maximum} of the $K$ norms has to be as small as possible. However, as stated in
\cite[p.~35]{Nguyen:09}, it is generally not possible to derive bounds for the
$K$\textsuperscript{th} successive minimum and any other of them except from the first. In
particular, $\mu_2,\dots,\mu_K$ can become arbitrarily large in comparison to the lattice volume.
It is quite obvious that the same holds for the norms of the basis vectors
$\ve{g}_2,\dots,\ve{g}_K$. Nevertheless, for the first successive minimum $\mu_1$ and the related
basis vector $\ve{g}_1$, bounds can in general be given.

\subsubsection{First Successive Minimum}

Given a real-valued lattice ($\I=\Z$) with generator matrix $\ve{G}\in\R^{N \times K}$, the
squared first successive minimum is bounded according to \emph{Minkowski's first theorem}
\cite{Hermite:50,Nguyen:09,Zhang:12}%
\begin{equation}	\label{eq:sm_bound_Z}
	\mu_{\Z,1}^2 \leq \eta_K\, \vol^{\frac{2}{K}}(\Lat(\ve{G})) \; ,
\end{equation}
where the factor $\eta_K$, which depends on the particular dimension, is called Hermite's
constant. It is only known for dimensions up to $K=8$ as well as $K=24$,
cf.\ \cite[Table on p.~33]{Nguyen:09}. However, it has been shown that Hermite's constant can
be upper-bounded by the term \cite{Blichfeldt:29}%
\begin{equation}	\label{eq:hermite}
	\eta_K \leq \frac{2}{\pi}\, \Gamma\left(2+\frac{K}{2}\right)^{\frac{2}{K}} \; ,
\end{equation}
where $\Gamma(x)=(x-1)!$ denotes the Gamma function.

Minkowski's first theorem is generalized in the following Theorem.
\begin{thm}[Generalized Bound on the First Successive Minimum]	\label{th:fsm}
	In lattices with $N \times K$ generator matrix $\ve{G}$ that are defined over the
	integer ring $\I$ and for which $D_\mathsf{r}$ denotes the number of real-valued components
	per scalar, the first successive minimum is bounded by%
	\begin{equation}	\label{eq:sm_bound}
		\mu_{\I,1}^2 \leq 
			\eta_{D_\mathsf{r} K} |\det\nolimits(\ve{G}_{\I})|^{\frac{2}{D_\mathsf{r} K}}
			\vol^{\frac{2}{K}}(\Lat(\ve{G})) \; .
	\end{equation}
\end{thm}
\begin{proof}
	According to Theorem~\ref{th:smp}, we have $\mu_{\I,1}^2=\mu_{\mathsf{r},\Z,1}^2$, the latter
	denoting the first successive minimum of the equivalent $D_\mathrm{r}N \times D_\mathrm{r}K$
	real-valued representation $\ve{G}_\mathsf{r}$.	Moreover,
	$\vol^{\frac{2}{D_\mathsf{r}K}}(\Lat(\ve{G}_\mathsf{r}))=\vol^{\frac{2}{K}}(\Lat(\ve{G}))$,
	cf.~\eqref{eq:det_real_compl}, \eqref{eq:det_real_quat}, and~\eqref{eq:volume}. Hence,
	applying Minkowski's first theorem~\eqref{eq:sm_bound_Z} w.r.t.\ $\mu_{\mathsf{r},\Z,1}^2$, we
	obtain%
	\begin{equation}	 \label{eq:proofSM}
		\begin{aligned}
			\mu_{\I,1}^2 &\leq \eta_{D_\mathsf{r}K}\, 
				\vol^{\frac{2}{D_\mathsf{r}K}}(\Lat(\ve{G}_{\mathsf{r},\I})) \\
			&= \eta_{D_\mathsf{r}K} 
			\det\nolimits^{\frac{1}{D_\mathsf{r}K}}
				(\ve{G}_{\mathsf{r},\I}^\T\ve{G}_{\mathsf{r},\I}^{})\\
			&= \eta_{D_\mathsf{r}K}\det\nolimits^{\frac{1}{D_\mathsf{r}K}}
				(\ve{G}_{\I}^\T\ve{G}_{\mathsf{r}}^\T\ve{G}_{\mathsf{r}}^{}\ve{G}_{\I}^{})\\
			&= \eta_{D_\mathsf{r}K}|\det\nolimits(\ve{G}_{\I})|^{\frac{2}{D_\mathsf{r}K}}
				\det\nolimits^{\frac{1}{D_\mathsf{r}K}}(\ve{G}_{\mathsf{r}}^\T\ve{G}_{\mathsf{r}})\\
			&= \eta_{D_\mathsf{r}K}|\det\nolimits(\ve{G}_{\I})|^{\frac{2}{D_\mathsf{r}K}} 
				\vol^{\frac{2}{K}}(\Lat(\ve{G}))
		\end{aligned}
	\end{equation}
	since the equivalent real-valued generator matrix of the integer ring $\I$, denoted as $\ve{G}_\I$,
	is a non-singular matrix.
\end{proof}
\begin{remark}
	Given lattices over $\Z$ ($D_\mathsf{r}=1$), $\G$ ($D_\mathsf{r}=2$), or $\L$ ($D_\mathsf{r}=4$),
	the matrix $\ve{G}_\I$ is an identity matrix, i.e., $\det(\ve{G}_\I)=1$. For complex lattices
	defined over $\E$ we have $\det(\ve{G}_\E)=({\sqrt{3}}/{2})^K$, cf.~\eqref{eq_equiv_real_E}, and
	for quaternionic ones over $\Hu$, $\det(\ve{G}_\Hu) = ({1}/{2})^K$, see~\eqref{eq:equiv_real_hur}.
	Hence, we obtain $|\det\nolimits(\ve{G}_{\E})|^{\frac{2}{D_\mathsf{r}K}}=\sqrt{3}/2$ and
	$|\det\nolimits(\ve{G}_{\Hu})|^{\frac{2}{D_\mathsf{r}K}}=1/\sqrt{2}$, respectively.
	
	Noteworthy, for the set of Eisenstein integers, \eqref{eq:sm_bound} is equivalent to the bound
	that has been derived for the special case of imaginary quadratic fields in \cite{Lyu:20}.
\end{remark}

When comparing the two types of complex lattices considered in this work, it can be recognized that
the bound~\eqref{eq:sm_bound} is smaller for $\I=\E$ than for $\I=\G$ since
$\frac{\sqrt{3}}{2}\approx0.866$. Hence, in general, the first successive minimum is expected to be
smaller. Given quaternion-valued lattices, the bound for $\Hu$ is lowered by a factor of
$\frac{1}{\sqrt{2}}\approx0.707$ in comparison to the bound for $\L$.

\subsubsection{First Basis Vector of an LLL Basis}	\label{subsec:lll_bounds}

In the initial publication on LLL reduction over $\Z$ \cite{Lenstra:82}, it has been shown that the
squared norm of the first basis vector can be bounded. Below, this bound is generalized for lattices
over integer rings~$\I$.
\begin{thm}[Generalized Bound on the First Vector Norm of an LLL Basis]
	Given a lattice with $N \times K$ generator matrix $\ve{G}$ which is LLL-reduced over the
	particular Euclidean integer ring $\I$, the first basis vector is bounded as%
	\begin{equation}	\label{eq:LLLbound_general}
		\mu_{\I,1}^2\leq \|\ve{g}_{\I,1}\|^2 \leq 
		\left(\frac{1}{\delta - \epsilon_{\I}^2}\right)^{\frac{K-1}{2}} 
		\vol^{\frac{2}{K}}(\Lat(\ve{G})) \; ,
	\end{equation}
	where $\epsilon_{\I}^2$ denotes the particular maximum squared quantization error.
\end{thm}
\begin{proof}
	The lower bound is obvious as $\mu_{\I,1}$ is the norm of the shortest non-zero lattice vector.
	The upper bound is, similar to~\cite{Lenstra:82}, derived as follows: Due to the Lov{\'a}sz
	condition~\eqref{eq:lovasz_general},%
	\begin{equation}
		\|\ve{q}_k\|^2 \leq \frac{1}{\delta - \epsilon_{\I}^2} \, \|\ve{q}_{k+1}\|^2\; .
	\end{equation}
	Hence, since in the (unsorted) GSO $\ve{q}_1=\ve{g}_1$,%
	\begin{equation}
			\|\ve{g}_{\I,1}\|^2 \leq
			\left(\frac{1}{\delta - \epsilon_{\I}^2}\right)^{k-1} \|\ve{q}_k\|^2 \; ,
			\qquad k=1,\dots,K \, .
	\end{equation}
	Taking the product over all indices $K$, we obtain%
	\begin{equation}
		\|\ve{g}_{\I,1}\|^{2K} \leq \left(\frac{1}{\delta - \epsilon_{\I}^2}\right)^{\frac{K(K-1)}{2}}
			\prod_{k=1}^{K}\|\ve{q}_k\|^2 \; .
	\end{equation}
	Since $\prod_{k=1}^{K}\|\ve{q}_k\|^2=\vol^{2}(\Lat(\ve{G}))$, see, e.g., \cite{Lenstra:82},%
	\begin{equation}
		\|\ve{g}_{\I,1}\|^2 \leq \left(\frac{1}{\delta - \epsilon_{\I}^2}\right)^{\frac{K-1}{2}}
			\vol^{\frac{2}{K}}(\Lat(\ve{G})) \; .
	\end{equation}
\end{proof}
\begin{remark}
	The bounds derived in \cite{Lenstra:82} ($\Z$) and \cite{Gan:09} ($\G$) cover special cases in
	which the particular maximum squared quantization errors $\epsilon_\Z^2=\frac{1}{4}$ and
	$\epsilon_{\G}^2=\frac{1}{2}$ have already been incorporated. Via~\eqref{eq:LLLbound_general}, the
	additional bounds for $\E$ ($\epsilon_\E^2=\frac{1}{3}$) and $\Hu$ ($\epsilon_\Hu^2=\frac{1}{2}$)
	are readily obtained.
\end{remark}
\subsubsection{Comparison of the Bounds}

The above upper bounds on the first successive minimum and the first vector norm of an LLL basis
are---normalized to the volume of the lattice---illustrated in Fig.~\ref{fig:boundsNorms}. In the
upper plot, the complex case is considered. In the bottom plot, quaternion-valued matrices are
regarded. To this end, the exact value for Hermite's constant has been chosen for all dimensions where
it is exactly known; otherwise, the approximation~\eqref{eq:hermite} has been used. For LLL reduction,
the (optimal) parameter $\delta=1$ is assumed.

Considering the first successive minimum, the superiority of lattices over $\E$ and $\Hu$, in
comparison to lattices over $\G$ and $\L$, or their isomorphic real- and complex-valued representations
over $\Z$ and $\G$, respectively, is clearly visible. For $N \times 2$ matrices, in both the complex-
and quaternion-valued case, the LLL-reduced first vector norms are identical to the respective first
successive minima. This is not a surprise since an LLL reduction with $\delta=1$ is equivalent to a
Gaussian reduction, cf.~\cite{Yao:02,Gan:09}. It is quite obvious that this relation does not hold any
more if the LLL reduction is applied to the equivalent $2N\times 4$ or $4N \times 8$ isomorphic
representations. When increasing the dimensions, the loss of the LLL approach in comparison to the
successive minima becomes more and more apparent. Among all variants of LLL reduction, given $\delta=1$,
the QLLL one ($\I=\Hu$) performs the best, followed by the ELLL one ($\I=\E$). The RLLL approach
($\I=\Z$) already shows a significant gap, and the CLLL one ($\I=\G$) performs the worst.

\begin{figure}[t]
	\centerline{%
		\includegraphics{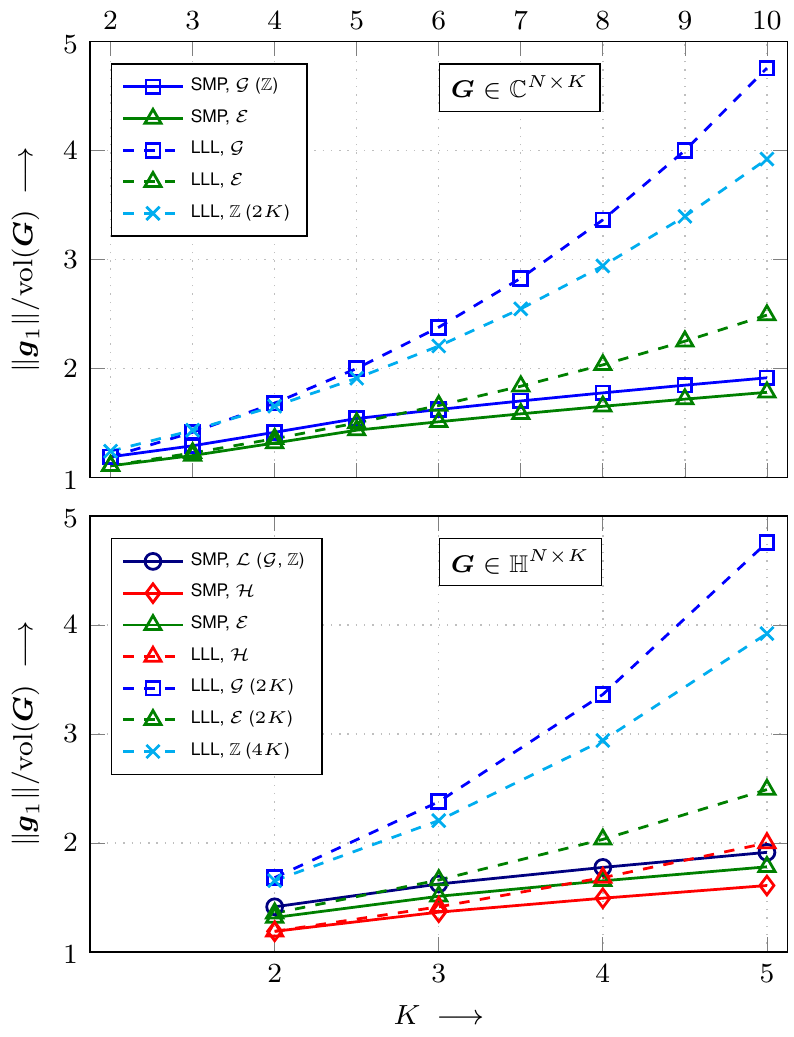}%
	}%
	\caption{\label{fig:boundsNorms}Upper bounds on the normalized first successive minimum
		(solid lines) and the normalized first vector norm of an LLL-reduced basis	with parameter
		$\delta=1$ (dashed lines) over the dimension $K$. Top: complex-valued lattices. Bottom:
		quaternion-valued lattices.}%
\end{figure}

In Fig.~\ref{fig:boundsNorms}, the quality parameter $\delta=1$ is considered. Still the question remains
which type of LLL reduction is---depending on $\delta$---the best-performing one in an asymptotic
manner. This comparison is expressed by the following Corollary.
\begin{cor}[Comparison of the First Vector Norm of LLL Variants]	\label{cor:norms}
	Given two variants of LLL reduction operating over lattices of rank $K_1$ and $K_2$, and related
	Euclidean integer rings with the maximum squared quantization errors
	$\epsilon_{1}^2$ and $\epsilon_{2}^2$, the first one has a smaller asymptotic ($K_1,K_2\to\infty$)
	upper bound~\eqref{eq:LLLbound_general} for the first basis vector than the second one if%
	\begin{equation}	\label{cor:norm}
		\frac{(\delta - \epsilon_{2}^2)^{{K_2-1}{}}}{({\delta - \epsilon_{1}^2)^{{K_1-1}{}}}} < 1 \; .
	\end{equation}
\end{cor}

In Appendix~\ref{app:LLLbounds}, \eqref{cor:norm} is briefly derived and the particular comparisons are
listed for all considered LLL variants. The main conclusions are: Given complex matrices
($\ve{G}\in\C^{N\times K}$), the ELLL approach generally performs better than the CLLL one, and it also
leads to a smaller bound when compared with the RLLL one (within the relevant range of $\delta$). In
contrast, the CLLL algorithm performs worse than the RLLL one, except from the case when
$\delta=\frac{3}{4}$, see also~\cite{Gan:09}. Given quaternionic lattices ($\ve{G}\in\Ha^{N \times K}$),
it can be stated that the QLLL algorithm shows better bounds than its CLLL, ELLL, and RLLL counterparts
within the relevant range of $\delta$. Hence, its use is highly recommendable in the quaternion-valued
case.

\subsection{Bounds on the Product of the Norms}

Even though the above bounds are restricted to the first vectors, it is possible to gain some indirect
knowledge on all vectors. In particular, bounds on their \emph{product} can be given.

\subsubsection{Product of the Successive Minima}

According to \emph{Minkowski's second theorem} \cite{Siegel:98,Nguyen:09}, for $L=1,\dots,K$, the product
of the first $L$ squared successive minima is bounded~by%
\begin{equation}	\label{eq:boundPSMZ}
	\prod\nolimits_{l=1}^{L} \mu_{\Z,l}^2 \leq \eta_K^L\,\vol^{\frac{2L}{K}}(\Lat(\ve{G})) \; .
\end{equation}
Since, for $L=K$,%
\begin{equation}	\label{eq:boundPSMZK}
	\prod\nolimits_{k=1}^{K} \mu_{\Z,k}^2 \leq \eta_K^K\,\vol^{2}(\Lat(\ve{G})) \; ,
\end{equation}
a bound on the orthogonality defect~\eqref{eq:orth_defect} of the transformed matrix
$\ve{G}_\mathrm{tra,\Z}=\ve{G}\ve{T}_\Z$ containing the shortest independent lattice vectors
(cf.\ Sec.~\ref{sec:algorithms}) is readily obtained.
\begin{thm}[Bound on the Orthogonality Defect of the Shortest Independent Vectors
		in a Real-Valued Lattice]	\label{th:orthZ}
	The orthogonality defect of the transformed matrix $\ve{G}_\mathrm{tra}$ is upper bounded by%
	\begin{equation}	\label{eq:orthDZ}
		\Omega(\ve{G}_\mathrm{tra,\Z}) \leq \eta_K^\frac{K}{2} \; .
	\end{equation}
\end{thm}
\begin{proof}
	If $\ve{T}_\Z$ is unimodular, \eqref{eq:orthDZ} directly follows from~\eqref{eq:boundPSMZK}	since
	$\Lat(\ve{G}_\mathrm{tra,\Z})=\Lat(\ve{G})$ is valid. If $\ve{T}_\Z$ is non-unimodular,
	$\Lat(\ve{G}_\mathrm{tra,\Z})\neq\Lat(\ve{G})$. However, since then,
	$\sqrt{\det(\ve{T}_\Z^\H\ve{T}_\Z)}>1$, we have
	$\vol(\Lat(\ve{G}_\mathrm{tra,\Z})) > \vol(\Lat(\ve{G}))$. Consequently,%
	\begin{equation}
		\Omega^2(\ve{G}_\mathrm{tra,\Z}) =
		\frac{\prod\nolimits_{k=1}^{K} \mu_{\Z,k}^2}{\vol^2(\Lat(\ve{G}_\mathrm{tra,\Z}))}
		< \frac{\prod\nolimits_{k=1}^{K} \mu_{\Z,k}^2}{\vol^2(\Lat(\ve{G}))} \; ,
	\end{equation}
	i.e., \eqref{eq:orthDZ} still follows from~\eqref{eq:boundPSMZK}.
\end{proof}
\begin{thm}[Generalized Bound on the Product of the Successive Minima and the Related Orthogonality Defect]
	In lattices with $N \times K$ generator matrix $\ve{G}$ that are defined over the integer ring $\I$ and
	for which $D_\mathsf{r}$ denotes the number of real-valued components per scalar, the product of the
	first $L$ successive minima is bounded by%
	\begin{equation}
		\prod\nolimits_{l=1}^{L} \mu_{\I,l}^2 \leq 
			\eta_{D_\mathsf{r}K}^L  |\det\nolimits(\ve{G}_{\I})|^{\frac{2L}{D_\mathsf{r} K}} 
			\vol^{\frac{2L}{K}}(\Lat(\ve{G})) \; ,
	\end{equation}
	and the orthogonality defect by%
	\begin{equation}
		\Omega(\ve{G}_\mathrm{tra,\I})=
		\eta_{D_\mathsf{r} K}^\frac{K}{2} |\det\nolimits(\ve{G}_{\I})|^{\frac{1}{D_\mathsf{r}}} \; .
	\end{equation}
\end{thm}
\begin{proof}
	The argumentation can be done equivalently to the derivations which are given in the proofs for
	Theorem~\ref{th:fsm} and Theorem~\ref{th:orthZ}.
\end{proof}
\begin{remark}
	Given lattices over $\Z$ ($D_\mathsf{r}=1$), $\G$ ($D_\mathsf{r}=2$), or $\L$ ($D_\mathsf{r}=4$),
	$|\det(\ve{G}_\I)|^{\frac{2L}{D_\mathsf{r}K}}=|\det(\ve{G}_\I)|^{\frac{1}{D_\mathsf{r}}}=1$.
	For complex lattices defined over $\E$ we have
	$|\det(\ve{G}_\E)|^{\frac{2L}{D_\mathsf{r}K}}=({\sqrt{3}}/{2})^L$
	and $|\det(\ve{G}_\E)|^{\frac{1}{D_\mathsf{r}}}=(3/4)^{\frac{K}{4}}$.
	For quaternionic ones over $\Hu$, $|\det(\ve{G}_\Hu)|^{\frac{2L}{D_\mathsf{r}K}}=(1/\sqrt{2})^L$, and
	$|\det(\ve{G}_\Hu)|^{\frac{1}{D_\mathsf{r}}}=(1/2)^{\frac{K}{4}}$.
\end{remark}

When employing complex lattices over $\E$, the upper bound on the orthogonality defect shrinks by a
factor of $(3/4)^{\frac{K}{4}}$ in comparison to lattices over $\G$. In contrast to the bound on the
first successive minimum in~\eqref{eq:sm_bound}, the expected gain grows with the dimension $K$. For
quaternion-valued lattices over $\I=\Hu$, the gap to the bound over $\I=\L$ grows by a factor of
$(1/2)^{\frac{K}{4}}$.

\subsubsection{Product of the Vector Norms of an LLL Basis}	\label{subsec:lll_boundsP}

A bound for the product of the norms of an LLL-reduced basis can be derived in a similar way. In
particular, it is known that at least the product over \emph{all} basis vectors is bounded
\cite{Lenstra:82,Nguyen:09}. This fact is generalized in the following Theorem.

\begin{thm}[Generalized Bound on the Product of the Norms of an LLL-Reduced Basis and its
	Related Orthogonality Defect]	\label{th:LLLPbound_general}
	Given a lattice $\Lat(\ve{G})$ spanned by the $N \times K$ generator matrix $\ve{G}$,
	the product of the norms of the basis vectors of an equivalent LLL-reduced matrix
	$\ve{G}_{\mathrm{red},\I}$ obtained over the integer ring $\I$ w.r.t.\ the quality parameter
	$\delta$ is bounded as%
	\begin{equation}	\label{eq:LLLPbound_general}
		\prod_{k=1}^{K}\mu_{\I,k}^2\leq \prod_{k=1}^{K}\|\ve{g}_{\I,k}\|^2 \leq \vol^{2}(\Lat(\ve{G}))
			\cdot \underbrace{\left(\frac{1}{\delta - \epsilon_{\I}^2}\right)^{\frac{K(K-1)}{2}} }_
			{\geq \Omega^2(\ve{G}_\mathrm{red,\I})} \;.
	\end{equation}
	Thereby, the rightmost term defines an upper bound on the squared orthogonality	defect of
	$\ve{G}_\mathrm{red,\I}$.
\end{thm}
\begin{proof}
	The lower bound is readily obtained since $\mu_{\I,k}$ represents the norms of the shortest
	independent vectors in the particular lattice. The upper bound is, similar to the proof of the
	original bound in \cite{Lenstra:82}, given as follows: After LLL reduction, the
	$k$\textsuperscript{th} basis vector can be written as%
	\begin{equation}
		\begin{aligned}
			\|\ve{g}_{\I,k}\|^2 &= \|\ve{q}_k\|^2 + \sum_{l=1}^{k-1} |r_{k,l}|^2 \|\ve{q}_l\|^2\\
			&\leq \|\ve{q}_k\|^2 + \sum_{l=1}^{k-1}\epsilon_{\I}^2 
				\left(\frac{1}{\delta-\epsilon_{\I}^2}\right)^{k-l} \|\ve{q}_k\|^2 \\
			&=\left(1 + \epsilon_{\I}^2\sum_{l=1}^{k-1}
				\left(\frac{1}{\delta-\epsilon_{\I}^2}\right)^{k-l}\right) \cdot \|\ve{q}_k\|^2 \\
			&\leq \left(\frac{1}{\delta-\epsilon_{\I}^2}\right)^{k-1}\cdot\|\ve{q}_k\|^2 \; .
		\end{aligned}
	\end{equation}
	Forming the product over all basis vectors, we obtain%
	\begin{equation}
		\begin{aligned}
			\prod_{k=1}^{K} \|\ve{g}_{\I,k}\|^2 \leq 
				\left(\frac{1}{\delta-\epsilon_{\I}^2}\right)^{\frac{K(K-1)}{2}}
				\underbrace{\prod_{k=1}^{K}\|\ve{q}_k\|^2}_{\vol^2(\Lat(\ve{G}))} \; .
		\end{aligned}
	\end{equation}
\end{proof}
\begin{figure}
	\centerline{%
		\includegraphics{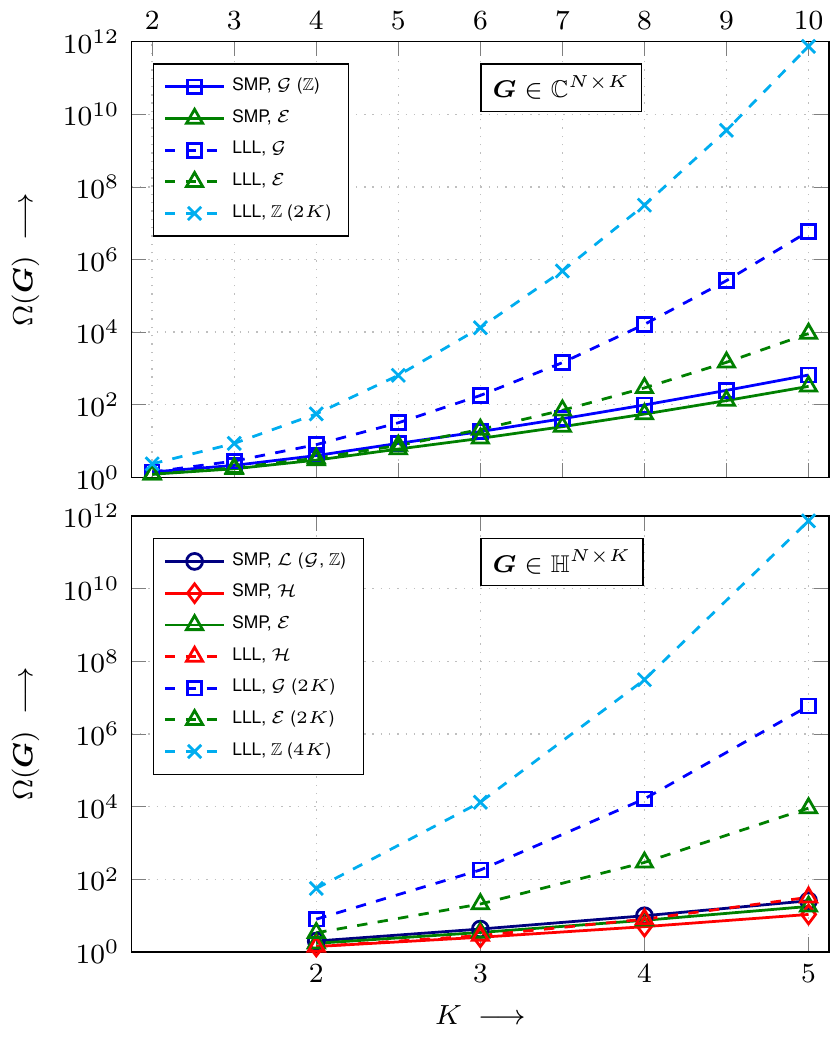}%
	}%
	\caption{\label{fig:boundsProd}Upper bounds on the orthogonality defect of the transformed
		matrix $\ve{G}_\mathrm{tra,\I}$ formed by the shortest independent vectors of the
		lattice (solid lines) and on the orthogonality defect of the LLL-reduced basis
		$\ve{G}_\mathrm{red,\I}$ with parameter $\delta=1$ (dashed lines) over the dimension $K$.
		Top: complex-valued lattices. Bottom: quaternion-valued	lattices.}%
\end{figure}
\subsubsection{Comparison of the Bounds}

In Fig.~\ref{fig:boundsProd}, the different bounds on the orthogonality defect are depicted
(Top: complex lattices; Bottom: quaternionic lattices). Again, the exact Hermite's constant has
been chosen if it is known for the particular dimension; otherwise, it has been approximated
by~\eqref{eq:hermite}. The quality parameter $\delta=1$ is assumed for LLL reduction.

Considering the orthogonality defect of the transformed matrix $\ve{G}_\mathrm{tra,\I}$ formed by
the shortest independent vectors, the superiority of lattices over $\E$
(Fig.~\ref{fig:boundsProd} Top) and $\Hu$ (Fig.~\ref{fig:boundsProd} Bottom) in comparison to
lattices defined over $\G$ and $\L$, respectively, is clearly visible. In the quaternion-valued
case, the QLLL reduction approaches the orthogonality defect of $\ve{G}_\mathrm{tra,\Hu}$ quite
well. In general, the QLLL reduction shows the best reduction quality, followed by the ELLL one.
Interestingly, the CLLL approach ($\I=\G$) performs better than the RLLL one ($\I=\Z$). This is
contrary to the behavior for the bounds on the first vector depicted in Fig.~\ref{fig:boundsNorms},
for which the RLLL reduction performed better. The reason is the additional factor $K$ in the
exponent of~\eqref{eq:LLLPbound_general} in comparison to~\eqref{eq:LLLbound_general}, i.e.,
the doubled dimension carries great weight in the RLLL approach.

Given a particular quality parameter $\delta$, the asymptotic behavior of the different types of
LLL reduction still has to be analyzed. This can be done with the following Corollary.
\begin{cor}[Comparison of the Orthogonality Defect of LLL Variants]	\label{cor:prod}
	Given two variants of LLL reduction operating over lattices of rank $K_1$ and $K_2$, and
	related Euclidean integer rings with the maximum squared quantization errors $\epsilon_{1}^2$
	and $\epsilon_{2}^2$, the first one has a smaller asymptotic ($K_1,K_2\to\infty$) upper bound
	on the orthogonality defect in~\eqref{eq:LLLPbound_general}	than the second one if%
	\begin{equation}
		\frac{(\delta - \epsilon_{2}^2)^{{K_2(K_2-1)}{}}}
			{({\delta - \epsilon_{1}^2)^{{K_1(K_1-1)}{}}}} < 1 \; .
\end{equation}
This inequality is straightforwardly obtained from~\eqref{eq:LLLPbound_general}.
\end{cor}

The particular comparisons are provided in Appendix~\ref{app:LLLpbounds}. We briefly summarize
the main results w.r.t.\ the asymptotic behavior of the orthogonality defect: Regarding complex
matrices, the CLLL approach performs better than the RLLL one, cf.~the discussion above. This is
a novel result that has not been recognized in the original CLLL paper \cite{Gan:09}. It is also
quite obvious that the ELLL approach performs better than the RLLL one for most $\delta$ except
from $\delta\approx\frac{1}{3}$. Besides, the ELLL strategy generally performs better than the
CLLL one. Given quaternionic lattices, the QLLL reduction is more powerful than the CLLL one.
Moreover, it performs better than the ELLL reduction and also better than the RLLL reduction,
except from a small range where $\delta\approx\frac{1}{2}$.

\subsection{Asymptotic Computational Complexity}

The asymptotic computational complexity is studied next. This includes a general discussion on
the complexity as well as a comparison of the different variants, i.e., of the complexity if
different types of integer rings are applied.

\subsubsection{Complexity of the List-Based Determination of the Successive Minima}
	\label{subsubsec:compl_smp}

In Algorithm~\ref{alg:smp}, three main steps can be identified: i) the call of the (real-valued)
LLL algorithm to find short initial basis vectors, ii) the call of the list sphere decoder that
provides all points within a hypersphere where the maximum norm of the basis vectors defines the
radius, and iii) the calculation of the row-echelon form.

Even if $\delta=1$, the LLL algorithm has a polynomial complexity given a particular dimension
\cite{Akhavi:03}; hence, it can efficiently be performed, see also the discussion below. Moreover,
the transformation to row-echelon form (Algorithm~\ref{alg:re}) that applies a simple Gaussian
elimination has a polynomial complexity over $K$ as well as the list size $N_\mathrm{c}$. In
particular, its asymptotic complexity reads $\mathcal{O}(K^2 N_\mathrm{c})$, since (less than) $K$
rows of $\ve{C}_\mathrm{s}$ have to be updated $K$ times, particularly each time when an independent vector
was found.

Hence, the crucial point in the algorithm is the call of the sphere decoder \cite{Agrell:02},
which is known to have an exponential complexity (over $K$). In~\cite{Fischer:16}, it has been
stated that the number of candidates, i.e., the number of points within a \emph{real-valued}
hypersphere, can be approximated by%
\begin{equation}
	N_{\mathrm{c},\Z} \approx \frac{(\pi \psi^2)^{\frac{K}{2}} }{\frac{K}{2}!\, \vol(\Lat(\ve{G}))}
\end{equation}
for the real-valued generator matrix $\ve{G}\in\R^{N\times K}$, see also~\cite{Conway:99}, where
$\psi^2$ denotes the squared search radius which is defined as
$\psi^2=\max_k\|\ve{g}_{\mathrm{red},k}\|^2$. Since the maximum norm of an LLL-reduced basis
(as well as the $K$\textsuperscript{th} successive minimum) cannot be bounded, the list size cannot
be bounded, too. This list size is generalized in the following Corollary.
\begin{cor}[Approximative List Size in the Generalized Determination of the Successive Minima]
	In lattices with $N \times K$ generator matrix $\ve{G}$ that are defined over the integer ring
	$\I$ and for which $D_\mathsf{r}$ denotes the number of real-valued components per scalar, the
	list size in the determination of the successive minima is approximated by%
	\begin{equation}
		\begin{aligned}
			N_{\mathrm{c},\I} &\approx 
			\frac{(\pi \psi^2)^{\frac{D_\mathsf{r}K}{2}} }{{\frac{D_\mathsf{r}K}{2}}!\,
					\vol(\Lat(\ve{G}_\mathsf{r,\I}))} \\
			&= \frac{(\pi \psi^2)^{\frac{D_\mathsf{r}K}{2}} }{{\frac{D_\mathsf{r}K}{2}}!\, 
					|\det(\ve{G}_\I)|\, \vol^{D_\mathsf{r}}(\Lat(\ve{G}))}  \; . 
		\end{aligned}
	\end{equation}
\end{cor}
\begin{remark}
	In the denominator, $\vol(\Lat(\ve{G}_\mathsf{r,\I}))$ has to be incorporated since the search
	is actually performed with the $D_\mathsf{r}N\times D_\mathsf{r}K$ real-valued representation
	$\ve{G}_\mathsf{r,\I}$ of $\ve{G}$ with
	$\vol(\ve{G}_\mathsf{r,\I})= |\det(\ve{G}_\I)|\,\vol^{D_\mathsf{r}}(\ve{G})$, cf.\
	Secs.~\ref{sec:extensions} and~\ref{sec:algorithms}.
\end{remark}

For lattices over Eisenstein and Hurwitz integers, we have $\det(\ve{G}_{\E})=(\sqrt{3}/2)^K$ and
$\det(\ve{G}_{\Hu})=(1/2)^{K}$, respectively. Hence, the number of candidates is increased in
comparison to lattices over Gaussian and Lipschitz integers, for which the determinant is one. This
is quite obvious as the Eisenstein as well as the Hurwitz integers constitute denser
packings---within the same hypervolume, more points are located. However, please note that the
initial RLLL reduction is then performed with $\ve{G}_{\mathsf{r,\E}}$ or $\ve{G}_{\mathsf{r,\Hu}}$.
Thereby, lower search radii may be obtained (see the above bounds), counteracting the increase in
list size.

\subsubsection{Complexity of LLL Reduction}	\label{subsubsec:LLLcompl}

The computational complexities of the different LLL approaches are finally assessed and compared to
each other. They are evaluated w.r.t.\ the number of \emph{real-valued multiplications} since,
in hardware implementation, multiplications are usually much more costly than additions.

It has been shown in the literature that an upper bound on the number of iterations in the LLL
algorithm, i.e., on the number of runs of the code lines within the while-loop in
Algorithm~\ref{alg:lll}, can be given \cite{Daude:94,Ling:07}. To this end, the elements of the
generator matrix have to be assumed to be drawn from a real-valued unit-variance uniform
\cite{Daude:94} or Gaussian \cite{Ling:07} distribution. Then, the asymptotic number of iterations
reads $\mathcal{O}(K \log_{\frac{1}{\delta}}(K))$ if $\delta < 1$, where the base of the logarithm
is the inverse of the quality parameter~$\delta$. For the case when $\delta=1$, see~\cite{Akhavi:03}.
It has been derived in \cite{Ling:07} that the same behavior holds for the complex case if a
circular-symmetric unit-variance Gaussian distribution is present. Adapting the derivation
in~\cite{Ling:07} to the quaternion-valued case, it can straightforwardly be shown that a
circular-symmetric unit-variance quaternionic Gaussian distribution leads to the same result.

It is well-known that the complexity inside the while loop of the LLL algorithm---assuming an
efficient implementation of the Gram--Schmidt update---is dominated by the for-loop for size
reduction (Lines 11--13 in Algorithm~\ref{alg:lll}), which has an asymptotic complexity of
$\mathcal{O}(NK)$, cf., e.g., \cite{Gan:09}. Hence, in total,\footnote{%
	The initial GSO (with pivoting) according to Algorithm~\ref{alg:gso} is not relevant in
	the analysis of the asymptotic behavior of the LLL algorithm as it only has a complexity of
	$\mathcal{O}(NK)$ required once in the beginning. Besides, even if a naive update of the GSO is
	performed within the while-loop, the total asymptotic complexity of one iteration still reads
	$\mathcal{O}(NK)$.%
}
this leads to the famous result that the LLL reduction as implemented in Algorithm~\ref{alg:lll}
has the asymptotic complexity $\mathcal{O}(K^3N\log_{\frac{1}{\delta}}(K))$. Noteworthy, this
complexity analysis holds for all types of LLL reduction described in this work---under the
assumption that the operations are performed in the particular real-, complex-, or
quaternion-valued arithmetic.

A complexity comparison of different LLL variants w.r.t\ the number of totally required
\emph{real-valued} multiplications, denoted as $M_\mathsf{r}$, is of interest. In \cite{Gan:09},
such a comparison was given for complex matrices and CLLL vs.\ RLLL reduction. This comparison is
generalized in the following Corollary.
\begin{cor}[Comparison of the Complexity of LLL Variants]
	Given two variants of LLL reduction operating over lattices of rank $K_1$ and $K_2$, and
	related Euclidean integer rings $\I_1$ and $\I_2$ for which $N_\mathsf{r,\I}$ denotes the
	number of real-valued multiplications required to implement a multiplication in the respective
	field, their ratio of the number of real-valued multiplications in the particular LLL algorithm
	$M_{\mathsf{r},\I}$ reads%
	\begin{equation}\label{eq:LLLnumber}
		\frac{M_\mathsf{r,\I_1}}{M_\mathsf{r,\I_2}} \approx
			\frac{N_\mathsf{r,\I_1}}{N_\mathsf{r,\I_2}} \cdot 
			\frac{K_1^3N_1\log_{\frac{1}{\delta}}(K_1)}
			{K_2^3 N_2\log_{\frac{1}{\delta}}(K_2)}\cdot\xi_{\I_1,\I_2} \; ,
	\end{equation}
	where
	$\xi_{\I_1,\I_2}={\mathrm{Pr}\{\Q_{\I_1}\{r_{l,k}\}\neq0\}}/
	{\mathrm{Pr}\{\Q_{\I_2}\{r_{l,k}\}\neq0\}}$
	denotes the ratio between the probabilities	that the particular size-reduction operation has to
	be performed.
	
	In an asymptotic view, \eqref{eq:LLLnumber} simplifies to%
		\begin{equation}\label{eq:LLLnumberAs}
			\lim_{K_1,K_2 \to \infty}\frac{M_\mathsf{r,\I_1}}{M_\mathsf{r,\I_2}} \approx
				\frac{N_\mathsf{r,\I_1}}{N_\mathsf{r,\I_2}} \cdot
				\frac{K_1^3N_1}{K_2^3 N_2}\cdot\xi_{\I_1,\I_2} \; ,
	\end{equation}
	since $\log_{\frac{1}{\delta}}(K_1)/\log_{\frac{1}{\delta}}(K_2)\to 1$.
\end{cor}
\begin{remark}
	To account for the fact that the size reduction only has to be performed if
	$\Q_{\I}\{r_{l,k}\}\neq 0$, cf.\ Algorithm~\ref{alg:size}, the ratio $\xi_{\I_1,\I_2}$ is
	incorporated. As stated in \cite{Ling:09} for complex matrices, the real value $r_{l,k}$
	in the RLLL algorithm and the complex components $r_{l,k}^{(1)}$ and $r_{l,k}^{(2)}$ in the CLLL
	one, which can be assumed to be independent in a radial-symmetric model, have quite similar
	statistics. Hence, $\mathrm{Pr}\{\Q_{\G}\{r_{l,k}\}\neq0\}\approx 
	2 \,\mathrm{Pr}\{\Q_{\Z}\{r_{l,k}\}\neq 0\}$, i.e., $\xi_{\G,\Z}\approx2$.
	
	Considering the ELLL approach, in the size-reduction step, the quantization is not performed
	w.r.t.\ a square Voronoi cell as in the CLLL approach but w.r.t.\ a hexagonal one. The latter
	covers less space ($\vol(\E)=\vol(\ve{A}_2)=\frac{\sqrt{3}}{2}$) than the former
	($\vol(\G)=\vol(\Z^2)=1$). However, the hexagonal cell is more similar to a two-dimensional
	hypersphere, i.e., a circle, than the square one. Hence, it covers a circular-symmetric
	(Gaussian) distribution more precisely.	It can be expected that both effects roughly compensate
	each other.	Thus, $\xi_{\E,\G}\approx 1$.
	
	For quaternionic lattices, four independent components $r_{l,k}^{(1)}$, $r_{l,k}^{(2)}$,
	$r_{l,k}^{(3)}$, and $r_{l,k}^{(4)}$ can be assumed if the QLLL algorithm is applied. However,
	here it has to be taken into account that the Voronoi cell of $\Hu$	forms a 24-cell as mentioned
	in Sec.~\ref{sec:extensions}, which covers less volume ($\vol(\Hu)=\frac{1}{2}$) than a
	four-dimensional hypercube ($\vol(\L)=\vol(\Z^4)=1$). Nevertheless, the former covers a
	circular-symmetric (Gaussian) distribution in a better way, cf.~above. Hence, we can again
	assume that $\xi_{\Hu,\L}\approx 1$ and, as a consequence, that $\xi_{\Hu,\G}\approx 2$ and
	$\xi_{\Hu,\Z}\approx 4$.
\end{remark}

We briefly evaluate the asymptotic complexity ratios according to~\eqref{eq:LLLnumberAs}. To this
end, please note that for the straightforward multiplication of complex numbers
in~\eqref{eq:compl_mul}, four real-valued multiplications are required. For the quaternionic
multiplication as defined in~\eqref{eq:quat_mul}, 16 real-valued multiplications are necessary.
Hence, we obtain $N_{\mathsf{r},\Z}=1$, $N_{\mathsf{r},\G}=N_{\mathsf{r},\E}=4$, and
$N_{\mathsf{r},\Hu}=16$. Assuming complex matrices, the asymptotic ratios of the required real
multiplications calculate to
$\frac{M_\mathsf{r,\G}}{M_\mathsf{r,\Z}}\approx
	\frac{M_\mathsf{r,\E}}{M_\mathsf{r,\Z}}\approx \frac{1}{2}$
for the comparison of complex and real-valued processing (given $K_2=2K_1$, $N_2=2N_1$, and
$\xi_{\G,\Z}\approx\xi_{\E,\Z}\approx 2$). Hence, for the former, the complexity is roughly halved,
see also~\cite{Gan:09}. Assuming quaternionic matrices, we obtain the ratios
$\frac{M_\mathsf{r,\Hu}}{M_\mathsf{r,\G}}\approx
	\frac{M_\mathsf{r,\Hu}}{M_\mathsf{r,\E}}\approx \frac{1}{2}$
for quaternionic vs.\ complex processing as $\xi_{\Hu,\G}\approx\xi_{\Hu,\E}\approx 2$. Comparing
quaternionic and real-valued processing, $\frac{M_\mathsf{r,\Hu}}{M_\mathsf{r,\Z}}\approx \frac{1}{4}$,
since $K_2=4K_1$, $N_2=4N_1$, and $\xi_{\Hu,\Z}\approx 4$. Thus, the complexity is significantly
reduced by using the QLLL algorithm instead of its real/complex counterparts.

Finally, we briefly consider the complexity ratios for the case when ``advanced'' multiplication
schemes are applied. As mentioned in Sec.~\ref{sec:extensions}, a complex-valued multiplication
can be implemented by only $N_\mathsf{r}=3$ real-valued multiplications, and a quaternion-valued
multiplication by only $N_\mathsf{r}=8$ ones. Hence, for complex- vs.\ real-valued processing, the
ratios
$\frac{M_\mathsf{r,\G}}{M_\mathsf{r,\Z}}\approx
	\frac{M_\mathsf{r,\E}}{M_\mathsf{r,\Z}}\approx \frac{3}{8}$
are obtained. For quaternion-valued vs.\ complex-valued processing, the complexity ratios are given
as
$\frac{M_\mathsf{r,\Hu}}{M_\mathsf{r,\G}}\approx
	\frac{M_\mathsf{r,\Hu}}{M_\mathsf{r,\E}}\approx \frac{1}{3}$,
and for quaternion-valued vs.\ real-valued processing, the number of required multiplications is
reduced to
$\frac{M_\mathsf{r,\Hu}}{M_\mathsf{r,\Z}}\approx \frac{1}{8}$.
Please note that these decreased numbers of multiplications are accompanied by increased numbers
of required additions. Hence, the best-performing strategy largely depends on the particular
hardware architecture.

\subsection{Numerical Evaluation and Comparison}	\label{subsec:numBounds}

To complement the theoretical derivations provided in this section, numerical simulations have
been performed. In particular, for both the complex- and the quaternion-valued case, $10^6$
i.i.d.\ unit-variance complex or quaternionic Gaussian random matrices have been considered.
Since we are interested in upper bounds on the quality and complexity, the assessment is
performed by evaluating the $0.99$-quantiles of the particular quantities, i.e., the values which
are surpassed by exactly $1~\%$ of the observed realizations. Again, for LLL reduction, the
optimal quality parameter $\delta=1$ is assumed.

\begin{figure}
	\centerline{%
		\includegraphics{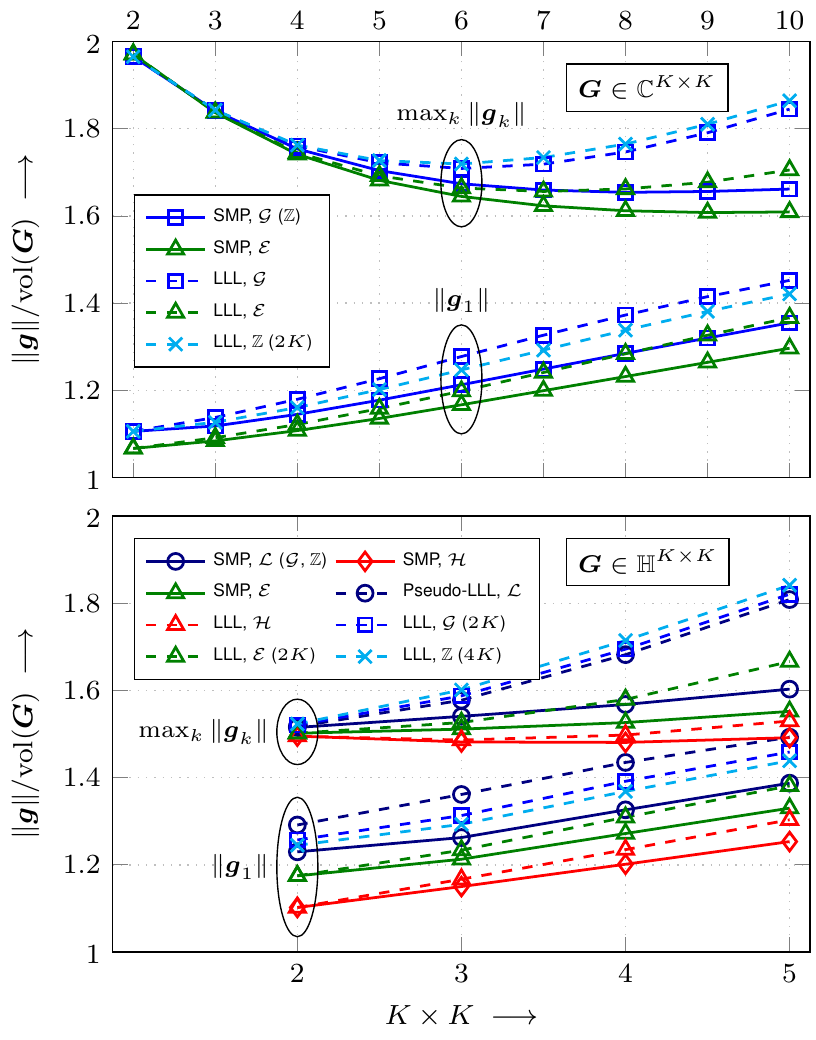}%
	}%
	\caption{\label{fig:statNorms}$0.99$-quantiles for the normalized first successive minimum
		(solid lines) and the normalized first vector norm of an LLL-reduced basis with parameter
		$\delta=1$ (dashed lines) as well as the related maximum values among all vectors over
		the dimensions $K \times K$ obtained from numerical simulations. Top: complex-valued
		lattices. Bottom: quaternion-valued lattices.}%
\end{figure}
\subsubsection{Norms of the Vectors}

In Fig.~\ref{fig:statNorms}, the $0.99$-quantiles of the normalized norms are illustrated,
cf.\ Fig.~\ref{fig:boundsNorms} (Top: complex matrices; Bottom: quaternionic matrices). In
addition to the norms of the first vectors, their maximum values are shown. The maximum among the
vector norms, which can only be evaluated numerically (see above), is the relevant quantity for
the SBP \eqref{eq:sbp} and the SIVP~\eqref{eq:sivp}.

Restricting the considerations to the norms of the first vectors in Fig.~\ref{fig:statNorms}, the
conclusions follow the ones that were drawn for the theoretical upper bounds in
Fig.~\ref{fig:boundsNorms}: Lattices over $\E$ and $\Hu$ possess lower first successive minima
than lattices over $\G$ and $\L$, respectively. The LLL reductions over $\E$ and $\Hu$ show the
best quality; their respective quantiles may even fall below the ones of the successive minima
over $\G$ and $\L$ for small dimensions. Among the ``genuine'' LLL approaches, the reduction over
$\G$ performs the worst. In the quaternion-valued case, only the curve for pseudo-QLLL reduction
(over $\L$) possesses higher quantiles. However, for statistical models like the i.i.d.\ Gaussian
one at hand, its application still results in a reasonable performance, though alternative
approaches may be more appropriate if the first vector is relevant.

For the maximum of the vector norms in Fig.~\ref{fig:statNorms}, similar conclusions can be
drawn---except for two important observations: First, in contrast to the first vector and similar
to the theoretical bounds on the orthogonality defect in Fig.~\ref{fig:boundsProd}, the CLLL
reduction performs better than the RLLL reduction. Second, in the quaternion-valued case, the
pseudo-QLLL reduction even possesses lower quantiles than the CLLL or RLLL reduction. Hence,
concerning an approximate solution for the SBP and the SIVP, respectively, the application of the
pseudo-QLLL reduction over $\L$ may even be beneficial in practice, if and only if the
``genuine'' QLLL reduction over $\Hu$ is not desired.

\subsubsection{Orthogonality Defect}

In Fig.~\ref{fig:statProd}, the $0.99$-quantile of the orthogonality defect is plotted as the
statistical quantity. Again, the complex case is shown at the top and the quaternionic case at
the bottom.

\begin{figure}
	\centerline{%
		\includegraphics{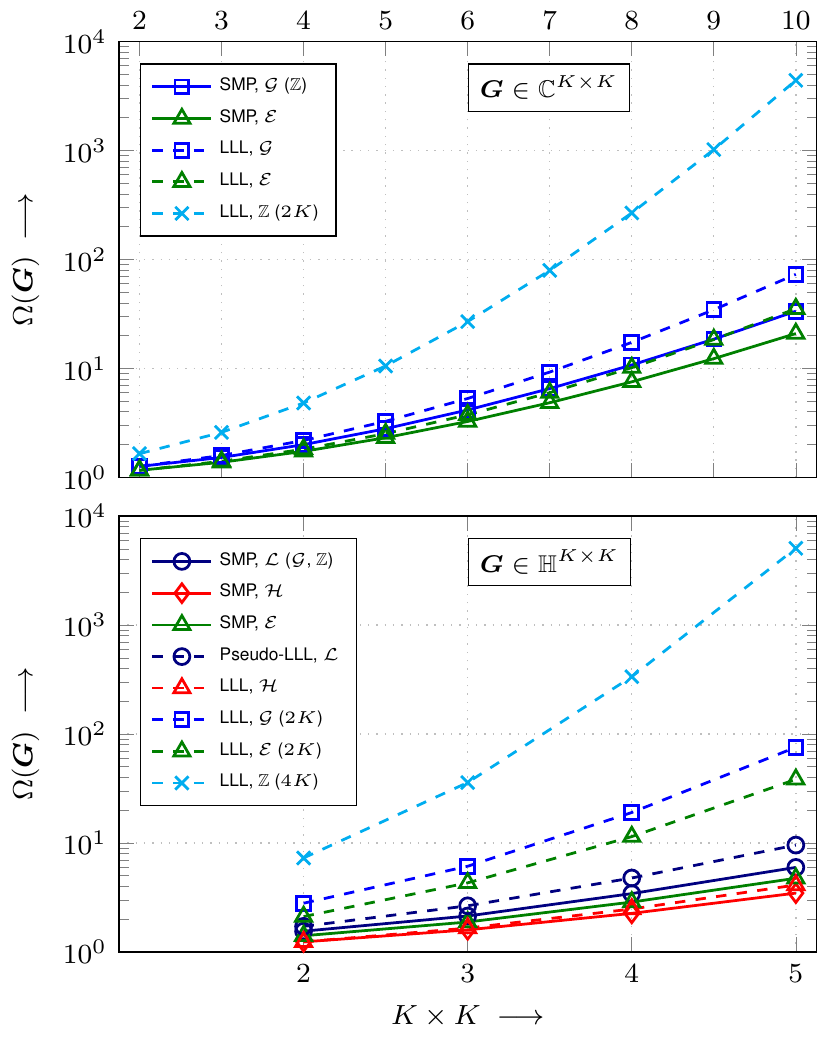}%
	}%
	\caption{\label{fig:statProd}$0.99$-quantiles for the orthogonality defect of the matrix
		formed by the shortest independent vectors of the lattice (solid lines) as well as the
		related basis vectors of an LLL-reduced basis with parameter $\delta=1$ (dashed lines)
		over the dimensions $K \times K$ obtained from numerical simulations. Top:
		complex-valued lattices. Bottom: quaternion-valued lattices.}%
\end{figure}

The shapes of the curves correspond to the behavior that can be expected when the theoretical
bounds from Fig.~\ref{fig:boundsProd} are regarded. For the complex case, the orthogonality
defect of the matrix formed by the shortest independent vectors over $\E$ is the lowest one,
whereas the application of the RLLL algorithm results in the worst quality.

For quaternion-valued lattices, the successive minima over $\Hu$ as well as the related LLL
reduction are accompanied by the smallest orthogonality defects. Again, the classical LLL
reduction over $\Z$ shows the worst quality. Surprisingly, at least w.r.t.\ the orthogonality
defect, the pseudo-QLLL reduction over $\L$ results in quite a good performance. In particular,
the orthogonality defect falls significantly below the one of any equivalent complex or
real-valued strategy.

\subsubsection{List Sizes of the Successive-Minima Algorithm}

Fig.~\ref{fig:statList} depicts the $0.99$-quantiles of the list size $N_\mathrm{c}$ for the
list-based determination of the successive minima of a lattice as described in
Sec.~\ref{subsec:smp}. Both the complex-valued ($\ve{G}\in\C^{K \times K}$) and the
quaternion-valued case ($\ve{G}\in\Ha^{K \times K}$) are considered.

\begin{figure}
	\centerline{%
		\includegraphics{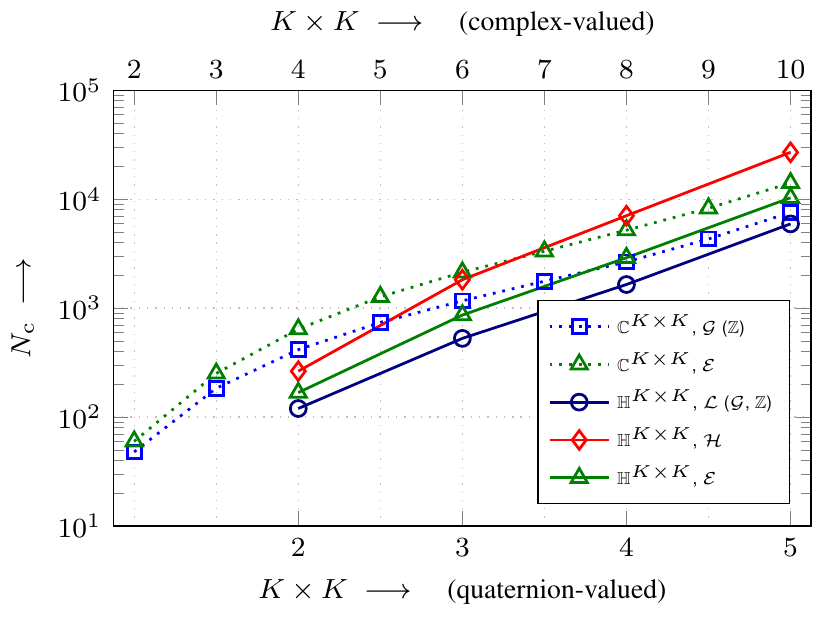}%
	}%
	\caption{\label{fig:statList}$0.99$-quantiles of the list size $N_\mathrm{c}$ for the
		determination of the successive minima according to Algorithm~\ref{alg:smp} over
		different integer rings $\I$ obtained by numerical simulations. Dashed curves:
		complex-valued lattices with generator matrix $\ve{G}\in\C^{K \times K}$. Solid curves:
		quaternion-valued lattices with generator matrix $\ve{G}\in\Ha^{K \times K}$.}%
\end{figure}

With regard to complex lattices, the determination of the successive minima over $\E$ is
accompanied by a slightly increased list size in comparison to $\G$. As explained in
Sec.~\ref{subsubsec:compl_smp}, this increase in complexity is caused by an increased number
of lattice points within a given search radius due to the denser packing. Obviously, even
though the initial search radius obtained by \emph{real-valued} LLL reduction may be lowered
a little bit for lattices over Eisenstein integers, the denser packing seems to be the more
dominating point.

The same holds for the quaternion-valued case: here, lattices over $\L$ result in the lowest
quantiles for the list size, whereas lattices over $\Hu$ are---due to the densest packing in
four dimensions---accompanied by the largest quantities. If the successive minima w.r.t.\ $\E$
are calculated for the equivalent complex-valued representation, the list size is located in
between the one of $\I=\L$ and $\I=\Hu$, as the packing is denser than the one of $\G^2$, but
sparser than the one of $\Hu$.

\subsubsection{Multiplications within the LLL Algorithm}

Finally, we assess the $0.99$-quantiles of the numbers of real-valued multiplications
$M_\mathsf{r}$ that are required to run the (generalized) LLL reduction as defined in
Algorithm~\ref{alg:lll}. To this end, the numbers are shown in Fig.~\ref{fig:mulStat}
(Top: complex lattices; Bottom: quaternionic lattices). Both the standard multiplication
approaches~\eqref{eq:compl_mul} and \eqref{eq:quat_mul}, respectively, with $N_\mathrm{r}=4$
and $N_\mathsf{r}=16$ real-valued multiplications (solid lines), and the reduced-complexity
variants with $N_\mathrm{r}=3$ and $N_\mathsf{r}=8$ multiplications (dashed-dotted lines),
are evaluated.

\begin{figure}
	\centerline{%
		\includegraphics{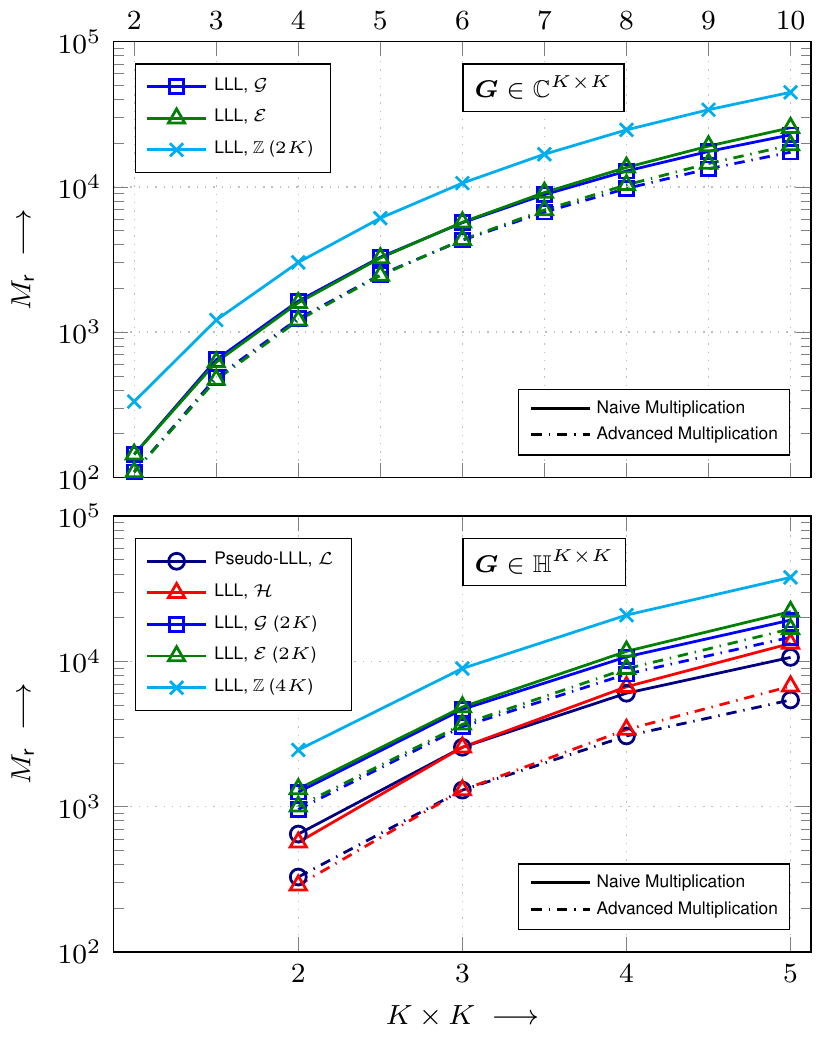}%
	}%
	\caption{\label{fig:mulStat}$0.99$-quantiles for the number of real-valued multiplications
		$M_\mathsf{r}$ within the generalized LLL reduction approach as stated in
		Algorithm~\ref{alg:lll} obtained by numerical simulations. Over each integer ring, the
		quality parameter \mbox{$\delta=1$} was used. Solid	lines indicate the naive
		implementation of the complex or quaternion-valued multiplication operations,
		dashed-dotted ones the implementation with reduced numbers of real-valued
		multiplications. Top: complex-valued lattices. Bottom: quaternion-valued lattices.}%
\end{figure}

Given the complex case in Fig.~\ref{fig:mulStat} (Top), it is clearly visible that the CLLL
and the ELLL algorithm roughly possess the same number of multiplications, which are halved
in comparison to the RLLL approach if the standard multiplication strategy is applied,
cf.\ Sec.~\ref{subsubsec:LLLcompl}. If the advanced complex multiplication with
$N_\mathrm{r}=3$ is used instead, these ratios are reduced to about $\frac{3}{8}$ like
predicted by the asymptotic assessment of the complexity.

For the quaternion-valued case in Fig.~\ref{fig:mulStat} (Bottom), the following conclusions
can be drawn: Pseudo-QLLL reduction over $\L$ and QLLL reduction over $\Hu$ possess about the
same number of real-valued multiplications. The same is valid for the CLLL and the ELLL
algorithm, using the equivalent complex-valued representation. Restricting to the standard
quaternion-valued multiplication approach, the number of multiplications can roughly be halved
in comparison to the CLLL and the ELLL approach, and roughly be reduced to one fourth in
comparison to the RLLL approach, if the reduction is performed over $\L$ or $\Hu$. When using
the advanced multiplication scheme for quaternionic numbers, these ratios can further be
reduced to one third and one eight, respectively, as already derived in the theoretical
analysis in Sec.~\ref{subsubsec:LLLcompl}.

\section{Application to MIMO Transmission}	\label{sec:mimo}

\noindent In this section, it is studied in which ways the generalized algorithms
derived and analyzed in the previous sections can be applied to the field of MIMO
communications. To this end, the system model of complex MIMO transmission is reviewed
and extended to the quaternion-valued case. Particular scenarios for the use of
quaternion-valued arithmetic are identified. In addition, lattice-based channel
equalization is adapted to the situation at hand. 

\subsection{SISO Fading Channel}

In wireless transmission, the fading model is quite popular to represent non-line-of-sight
connections. In the simplest case, only one transmit and one receive antenna are used, i.e.,
a single-input/single-output (SISO) fading channel is present. It can be modeled by the
system equation\footnote{%
	In order to simplify the notation, the time index is omitted in all system equations,
	i.e., one particular time step (modulation step) is considered.%
}%
\begin{equation}	\label{eq:siso_model}
	y = h \cdot x + n \; .
\end{equation}
Thereby, $x$ is a transmit symbol taken from the finite set of symbols $\mathcal{A}$
(signal constellation), $h$ denotes the fading coefficient (multiplicative distortion),
$n$ represents Gaussian noise (additive distortion), and $y$ is the disturbed
receive~symbol.

\subsubsection{Complex-Valued Transmission}

Most often, \eqref{eq:siso_model} is considered to be a complex-valued equation that
models radio-frequency transmission in the equivalent complex baseband
\cite{Fischer:02,vanTrees:04}. Then, a (zero-mean) QAM constellation
$\mathcal{A}\subset\C$ with the variance $\sigma_{x,\mathsf{c}}^2$ and the cardinality
$M_\mathsf{c}=|\mathcal{A}|$ is most commonly applied. These constellations form
(shifted) subsets of the Gaussian integers \cite{Fischer:16,Fischer:19}. Alternatively,
constellations based on the Eisenstein integers can be employed
\cite{Huber:94b,Tunali:15,Stern:19,Fischer:19}. In accordance, the fading coefficient
is complex (\emph{Rayleigh} fading \cite{Proakis:08}, usually normalized to
$\sigma_{h,\mathsf{c}}^2=1$), and we have to deal with (radial-symmetric) complex
Gaussian noise with some variance $\sigma_{n,\mathsf{c}}^2$.

Taking advantage of~\eqref{eq:equiv_real_compl}, the SISO Rayleigh-fading channel can
equivalently be modeled by a real-valued system with the system equation\footnote{%
	For $x$ and $n$, only the left column of the equivalent real-valued
	representation~\eqref{eq:equiv_real_compl} is employed, as the right
	column is completely redundant and not required to obtain the final
	result (left column) of $y$.%
}%
\begin{equation}
	\begin{bmatrix}
		y^{(1)}\\
		y^{(2)}
	\end{bmatrix}=
	\begin{bmatrix}
		h^{(1)} & -h^{(2)}\\
		h^{(2)} & \phantom{-}h^{(1)} 
	\end{bmatrix}
	\begin{bmatrix}
		x^{(1)}\\
		x^{(2)}
	\end{bmatrix}
	+
	\begin{bmatrix}
		n^{(1)}\\
		n^{(2)}
	\end{bmatrix} \; ,
\end{equation}
where the noise components $n^{(1)}$ and $n^{(2)}$ have the variance
$\sigma_{n,\mathsf{r}}^2 = \frac{1}{2}\sigma_{n,\mathsf{c}}^2$ and the variances of
the constellation's components read
$\sigma_{x,\mathsf{r}}^2 = \frac{1}{2}\sigma_{x,\mathsf{c}}^2$ if the components
are independent (e.g., in QAM).

\subsubsection{Quaternion-Valued Transmission}

On the basis of~\eqref{eq:siso_model}, a \emph{quaternion-valued} SISO fading channel
model can be defined. Then, the transmit symbols are consistently drawn from a
quaternion-valued signal constellation $\mathcal{A}\subset\Ha$ with the (4D)
cardinality $M_\mathsf{q}=|\mathcal{A}|$. The signal points can, e.g., be chosen as a
subset of the Lipschitz or the Hurwitz integers \cite{Stern:18,Frey:20}. The related
variance reads $\sigma_{x,\mathsf{q}}^2$. The fading coefficient is given as%
\begin{equation}
	h = (\underbrace{h^{(1)} + h^{(2)}\,\i}_{h^{\{1\}}}) +
		(\underbrace{h^{(3)} + h^{(4)}\,\i}_{h^{\{2\}}})\,\j \, ,
\end{equation}
i.e., it consists of \emph{four independent real-valued} or
\emph{two independent complex-valued} ones, respectively. Since two independent
unit-variance Rayleigh-fading coefficients are present, the variance reads
$\sigma_{h,\mathsf{q}}^2=2\,\sigma_{h,\mathsf{c}}^2=4\,\sigma_{h,\mathsf{r}}^2=2$.
In the same way, the noise is represented by%
\begin{equation}
		n = (\underbrace{n^{(1)} + n^{(2)}\,\i}_{n^{\{1\}}}) +
		(\underbrace{n^{(3)} + n^{(4)}\,\i}_{n^{\{2\}}})\,\j \, ,
\end{equation}
where $\sigma_{n,\mathsf{q}}^2=2\,\sigma_{n,\mathsf{c}}^2=4\,\sigma_{n,\mathsf{r}}^2$.
A disturbed receive symbol $y\in\Ha$ is finally obtained.

\begin{figure}[t]
	\centerline{%
		\includegraphics{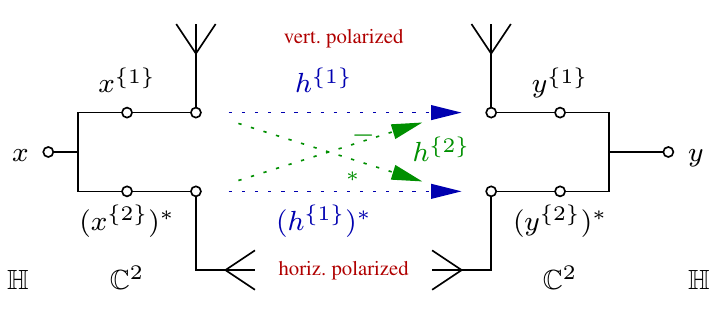}%
	}%
	\caption{\label{fig:dualPol}%
		Transmission with dual-polarized antennas over the SISO fading channel
		according to~\eqref{eq:equiv_compl_quat_siso}. At the transmitter and
		receiver, antenna pairs are present that transmit and receive electromagnetic
		waves which carry information in both the vertical polarization plane (top) and
		the horizontal polarization plane (bottom).	In the quaternion-valued fading
		factor $h=h^{\{1\}}+h^{\{2\}}\j$, both the direct gain $h^{\{1\}} $and the
		cross-polar gain $h^{\{2\}}$ are contained.}%
\end{figure}

Benefiting from~\eqref{eq:equiv_compl_quat}, the quaternion-valued SISO fading channel
is equivalently expressed by the complex-valued $2 \times 2$ system equation%
\begin{equation}	\label{eq:equiv_compl_quat_siso}
	\begin{bmatrix}
		y^{\{1\}}\\
		(y^{\{2\}})^*
	\end{bmatrix}=
	\begin{bmatrix}
		h^{\{1\}} & -h^{\{2\}}\\
		(h^{\{2\}})^* & (h^{\{1\}})^* 
	\end{bmatrix}
	\begin{bmatrix}
		x^{\{1\}}\\
		(x^{\{2\}})^*
	\end{bmatrix}
	+
	\begin{bmatrix}
		n^{\{1\}}\\
		(n^{\{2\}})^*
	\end{bmatrix} \; .
\end{equation}
This complex $2 \times 2$ system equation and its quaternion-valued (scalar)
representation, respectively, are well-suited to model particular transmission
scenarios:
\begin{enumerate}[label=\roman*)]
	\item Transmission with \emph{dual-polarized antennas}:
		as illustrated in Fig.~\ref{fig:dualPol}, \eqref{eq:equiv_compl_quat_siso}
		models the (SISO) transmission with one dual-polarized antenna at both the
		transmitter and the receiver side. In particular, both horizontal and vertical
		polarization of the electromagnetic wave are then used for the orthogonal
		transmission of \emph{two} complex-valued symbols at the same time and on the
		same frequency band. Thereby, the first complex fading factor $h^{\{1\}}$
		describes the \emph{direct gain} within the same polarization plane, whereas
		$h^{\{2\}}$ represents the \emph{cross-polar gain}, i.e., the crosstalk to the
		other polarization plane. Moreover, the noise samples $n^{\{1\}}$ and $n^{\{2\}}$
		describe the additive noise which is present at the vertically and horizontally
		polarized receive antenna, respectively. Further details can be found in
		\cite{Isaeva:95,Wysocki:06,Stern:18}.
	\item
	 Transmission with Alamouti (space-time) coding:
		the system equation~\eqref{eq:equiv_compl_quat_siso} corresponds to the one
		of the Alamouti scheme as a diversity technique \cite{Alamouti:98}.	Then, the
		complex symbols $x^{\{1\}}$ and $x^{\{2\}}$ may actually be radiated at
		different time steps (or, e.g., frequencies), but are processed jointly.
		Diversity is obtained if the channel gains $h^{\{1\}}$ and $h^{\{2\}}$
		(representing the two time steps, frequencies, \dots) are independent. The
		quaternion-valued SISO fading model can be seen as an alternative (scalar)
		representation of the $2\times 2$ space-time coding scheme. For further details
		and a deeper comparison of dual-polarized transmission and Alamouti coding,
		see~\cite{Qureshi:18}.
\end{enumerate}

\subsection{MIMO Fading Channel (Uplink Transmission)}

The complex- or quaternion-valued SISO fading models can be generalized to respective
MIMO ones. In this work, \emph{MIMO uplink transmission}---aka
\emph{MIMO multiple-access channel}---is treated in an exemplary way.

$K$ uncoordinated single-antenna user devices transmit their data---at the same time and
on the same frequency---to \emph{one} central receiver which is equipped with $N \geq K$
antennas. The related system equation is given as%
\begin{equation}
	\ve{y} = \ve{H} \cdot \ve{x} + \ve{n} \; .
\end{equation}
Here, $\ve{x}=[x_1,\dots,x_K]^\T$ denotes the vector of transmit symbols sent by the
users, $\ve{H}$ the $N \times K$ MIMO channel matrix, $\ve{n}=[n_1,\dots,n_N]^\T$ a
vector with $N$ noise samples, and $y=[y_1,\dots,y_N]^\T$ the vector of $N$ symbols which
are received at the central unit.

\subsubsection{Complex-Valued Transmission}

Most often, a complex-valued MIMO channel is considered. A popular model is that the
channel matrix%
\begin{equation}	\label{eq:MIMO_channel}
	\ve{H} = \big[\,h_{n,k}\,\big]_{\substack{n=1,\dots,N\\ k=1,\dots,K}} 
		\in \C^{ N\times K}
\end{equation}
contains i.i.d.\ complex Gaussian unit-variance channel gains $h_{n,k}\in\C$ that
describe (in equivalent complex baseband representation) the links between the user $k$
and the receive antenna $n$. At each receive antenna, complex Gaussian noise with the
(same) variance $\sigma_{n,\mathsf{c}}^2$ is assumed.

Given the assumption that the transmit symbols $x_1,\dots,x_K$ are drawn from a
complex-valued integer ring $\I$, i.e., $x_k\in\G$ or $x_k\in\E$, $k=1,\dots,K$, and
neglecting the noise, the receive symbols are drawn from (a subset of) a complex-valued
lattice as defined in~\eqref{eq:lattice}. Thereby, the generator matrix is given as
$\ve{G}=\ve{H}$.

The complex-valued system equation is equivalently expressed by the real-valued equation%
\begin{equation}
	\underbrace{
	\begin{bmatrix}
		\ve{y}^{(1)}\\
		\ve{y}^{(2)}
	\end{bmatrix}
	}_{\ve{y}_\mathsf{r}}
	=
	\underbrace{
	\begin{bmatrix}
		\ve{H}^{(1)} & -\ve{H}^{(2)}\\
		\ve{H}^{(2)} & \phantom{-}\ve{H}^{(1)} 
	\end{bmatrix}
	}_{\ve{H}_\mathsf{r}}
	\underbrace{
	\begin{bmatrix}
		\ve{x}^{(1)}\\
		\ve{x}^{(2)}
	\end{bmatrix}
	}_{\ve{x}_\mathsf{r}}
	+
	\underbrace{
	\begin{bmatrix}
		\ve{n}^{(1)}\\
		\ve{n}^{(2)}
	\end{bmatrix}
	}_{\ve{n}_\mathsf{r}} \; .
\end{equation}
Consequently, given the case that $\ve{x}\in\G^K$, an equivalent $2N \times 2K$
real-valued lattice (over $\I=\Z$) with the generator matrix $\ve{H}_\mathsf{r}$ is
spanned, where $\ve{x}_{\mathsf{r}}\in\Z^{2K}$. Lattices over $\I=\E$ can be expressed
according to~\eqref{eq_equiv_real_E}. 

\subsubsection{Quaternion-Valued Transmission}

It is possible to extend the complex-valued MIMO uplink model to the quaternion-valued
case. Then, the channel matrix is represented as%
\begin{equation}	\label{eq:MIMO_channel_H}
	\ve{H} = \big[\,h_{n,k}\,\big]_{\substack{n=1,\dots,N\\ k=1,\dots,K}} 
		\in \Ha^{ N\times K} \; ,
\end{equation}
i.e., it contains i.i.d.\ quaternion-valued Gaussian channel gains $h_{n,k}\in\Ha$
(Gaussian distribution in each component) with the total variance
$\sigma_{h,\mathsf{q}}^2=2\,\sigma_{h,\mathsf{c}}^2=2$, i.e., unit-variance channel
gains per complex component as often considered in MIMO communications. Under the
assumption that the transmit symbols $x_1,\dots,x_K$ are now chosen from a
quaternion-valued integer ring $\I$, i.e., $x_k\in\L$ or $x_k\in\Hu$, $k=1,\dots,K$,
and neglecting the noise, the receive symbols form a subset of a quaternion-valued
lattice according to~\eqref{eq:lattice} with $\ve{G}=\ve{H}$.

Given the scenario of MIMO uplink transmission via dual-polarized antennas as depicted
in Fig.~\ref{fig:dualMIMO}, the quaternion-valued channel gains describe the links
between the \emph{transmit-antenna pairs} $k$, $k=1,\dots,K$, and the
\emph{receive-antenna pairs} $n$, $n=1,\dots,N$. At each receive-antenna pair,
quaternion-valued noise samples with the (total) variance $\sigma_{n,\mathsf{q}}^2$ are
assumed that contain the noise of both polarization planes. When using the quaternionic
model for the representation of the Alamouti coding scheme, the horizontally-polarized
antennas can be replaced by ``virtual'' antennas that represent the transmission at
another time step or in another frequency band instead.

\begin{figure}[t]
	\centerline{%
		\includegraphics{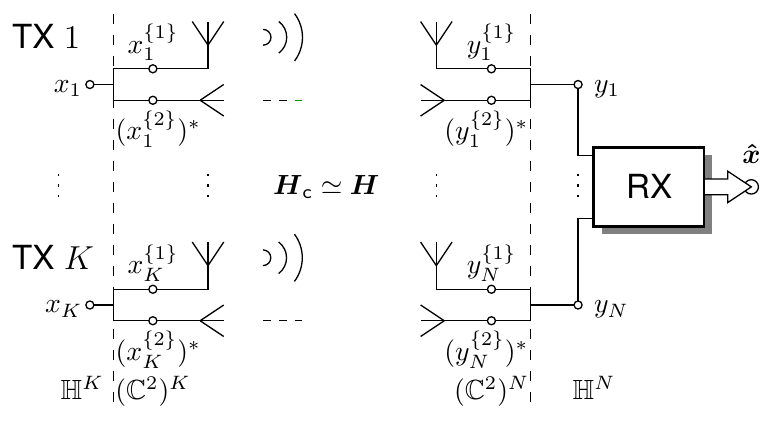}%
	}%
	\caption{\label{fig:dualMIMO}%
		Quaternion-valued MIMO uplink transmission (MIMO multiple-access channel)
		with dual-polarized antennas. $K$ uncoordinated pairs of transmit antennas
		(in each TX, one for vertical and one for horizontal polarization) radiate the
		users' data symbols to the $N$ pairs of receive antennas. The quaternion-valued
		receive symbols $y_1,\dots,y_N$ are processed jointly in a central receive unit
		(RX) in order to obtain estimates of the transmit symbols $x_1,\dots,x_K$.}%
\end{figure}

As becomes apparent from Fig.~\ref{fig:dualMIMO}, the quaternion-valued transmission
is actually realized by the equivalent $2N\times 2K$ complex-valued system model%
\begin{equation}	\label{eq:equiv_compl_MIMO}
	\underbrace{
	\begin{bmatrix}
		\ve{y}^{\{1\}}\\
		(\ve{y}^{\{2\}})^*
	\end{bmatrix}
	}_{\ve{y}_\mathsf{c}}
	=
	\underbrace{
	\begin{bmatrix}
		\ve{H}^{\{1\}} & -\ve{H}^{\{2\}}\\
		(\ve{H}^{\{2\}})^* & (\ve{H}^{\{1\}})^* 
	\end{bmatrix}
	}_{\ve{H}_\mathsf{c}}
	\underbrace{
	\begin{bmatrix}
		\ve{x}^{\{1\}}\\
		(\ve{x}^{\{2\}})^*
	\end{bmatrix}
	}_{\ve{x}_\mathsf{c}}
	+
	\underbrace{
	\begin{bmatrix}
		\ve{n}^{\{1\}}\\
		(\ve{n}^{\{2\}})^*
	\end{bmatrix}
	}_{\ve{n}_\mathsf{c}} \; .
\end{equation}
Hence, if $\ve{x}\in\L^K$, the quaternion-valued lattice $\Lat(\ve{H})$ is
isomorphically represented by a complex lattice ($\I=\G$) with the generator matrix
$\ve{G}=\ve{H}_\mathsf{c}$. In the same way, on the basis of~\eqref{eq:equiv_real_quat},
an equivalent real-valued lattice ($\I=\Z$) can be defined. A quaternion-valued
transmission with symbols drawn from $\I=\Hu$ can be represented by equivalent lattices
over $\I=\Z$ and $\I=\G$ according to~\eqref{eq:equiv_real_hur} and
\eqref{eq:equiv_compl_hur}, respectively.

\subsection{Lattice-Reduction-Aided and Integer-Forcing Equalization}	
	\label{subsec:lra}

In the MIMO (uplink) scenario, handling the (multi-user) interference is a crucial point.
It is well-known that (purely) linear channel equalization according to the zero-forcing
(ZF)  or the minimum mean-square error (MMSE) criterion does not achieve a satisfactory
performance since it does not exploit the MIMO channel's (spatial) diversity, see, e.g.,
\cite{Zheng:03,Tse:05}.

In contrast, the concepts of \emph{lattice-reduction-aided} (LRA) linear equalization
\cite{Yao:02,Windpassinger:03,Wuebben:11} and the closely related integer-forcing (IF)
linear equalization \cite{Zhan:14} are suited in order to achieve the receive diversity
of the MIMO channel \cite{Taherzadeh:07}, see also the discussion below. In the following,
a universal description thereof is provided that enables a lattice-based channel
equalization over all above-mentioned integer rings~$\I$. To that end, please note that
both the LRA and the IF receiver share the same equalization task; they only differ in the
way the integer interference is treated in combination with the channel code,
see~\cite{Fischer:19} or the section on soft-decision decoding given below.

The generalized concept is described as follows: The transmit symbols $x_1,\dots,x_K$ are
chosen from a subset $\mathcal{A}$ of the integer ring~$\I$. Then, as illustrated in
Fig.~\ref{fig:if_lra}, the channel matrix $\ve{H}$ is \emph{not directly} equalized by its
pseudoinverse\footnote{%
	The $K \times N$ (left) pseudoinverse of an $N \times K$ matrix $\ve{G}$ is calculated
	as $\ve{G}^{+}=(\ve{G}^\H\ve{G})^{-1}\ve{G}^\H$. If $N=K$, $\ve{G}^{+}=\ve{G}^{-1}$.
	The Hermitian of $\ve{G}^{+}$ is denoted as $\ve{G}^{+\H}=(\ve{G}^{+})^{\H}$, and as
	$\ve{G}^{-\H}$ if $\ve{G}$ is a square matrix.%
}
$\ve{H}^{+}$. Instead, a \emph{transformed} channel $\ve{H}_\mathrm{tra}$ is first handled
via the $K \times N$ filter matrix $\ve{F}=\big[\ve{f}_1^\H,\dots,\ve{f}_K^\H\big]^\H$
with the rows $\ve{f}_1,\dots,\ve{f}_K$. Assuming the ZF criterion, the filter matrix is
obtained as\footnote{%
	The lattice spanned by $\ve{G}=\ve{H}^{+\H}$ is the lattice which is \emph{dual} to the
	one spanned by $\ve{G}=\ve{H}$, cf.~\cite{Conway:99,Ling:09,Fischer:19}.%
}
\cite{Taherzadeh:07,Fischer:19}%
\begin{equation}	\label{eq:fac_ZF}
	\ve{F}^\H=\ve{H}_\mathrm{tra}^{+\H}=\ve{H}^{+\H}
		\underbrace{\ve{Z}^{\H}}_{\ve{T}} \; ,
\end{equation}
where the transformation is expressed by the \emph{integer matrix}
$\ve{T}\in\I^{K \times K}$. In particular, the matrix $\ve{F}$ \emph{shapes} the
interference in such a way that only \emph{integer interference} is left before
decoding---this integer interference is particularly expressed by the integer matrix
$\ve{Z}=\ve{T}^\H\in\I^{K \times K}$ and finally reversed via the matrix
$\ve{Z}^{-1} = \ve{T}^{-\H}$ as illustrated in Fig.~\ref{fig:if_lra}.

\begin{figure}[t]
	\centerline{%
		\includegraphics{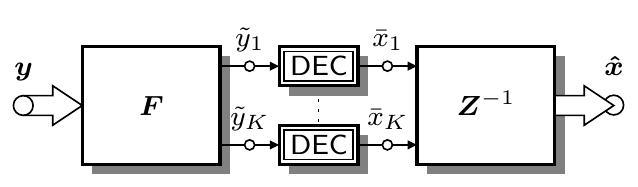}%
	}%
	\caption{\label{fig:if_lra}%
		General receiver structure of LRA/IF linear equalization \cite{Fischer:19}.
		The receive symbols in $\ve{y}$ are linearly equalized via the filter matrix
		$\ve{F}$ before channel decoding ($\mathsf{DEC}$) w.r.t.\ linear
		combinations $\bar{x}_k$ of the transmit symbols is applied. The remaining
		interference is resolved via the inverse of the integer matrix~$\ve{Z}$,
		resulting in estimates $\hat{x}_k$ of the original transmit symbols.}%
\end{figure}

In~\eqref{eq:fac_ZF}, the transformation matrix $\ve{T}$ (and, thus, $\ve{Z}$) should be
chosen in a way that the row norms of $\ve{F}$ are minimized. In the presence of
additive white Gaussian noise with the variance $\sigma_{n}^2$, they determine the
\emph{noise variances}%
\begin{equation}	\label{eq:noise_var}
	\sigma_{n,k}^2 = \|\ve{f}_k\|^2 \cdot \sigma_{n}^2 \; , \qquad k=1,\dots,K \; ,
\end{equation}
and, thus, the individual signal-to-noise ratios (SNRs) and the related mean-square errors
(MSEs) before decoding. In order to minimize the MSEs, the ZF criterion applied
in~\eqref{eq:fac_ZF} is not the optimum strategy. They can be lowered by employing the
\emph{MMSE criterion}. To this end, the matrices in~\eqref{eq:fac_ZF} can be replaced by
their \emph{augmented ones} \cite{Hassibi:00,Wuebben:04} according to
\cite{Fischer:16,Fischer:19}%
\begin{equation}	\label{eq:fac_MMSE}
	\ve{\mathcal{F}}^\H=\ve{\mathcal{H}}_\mathrm{tra}^{+\H}=
	\underbrace{\begin{bmatrix}
			\ve{H}\\
			\frac{\sigma_{n}}{\sigma_{x}}\ve{I}_{K}
		\end{bmatrix}^{+\H}
	}_{\ve{\mathcal{H}}^{+\H}}
	\underbrace{\ve{Z}^{\H}}_{\ve{T}} \; ,
\end{equation}
where $\ve{\mathcal{H}}$ denotes the $(N+K) \times K$ augmented channel matrix in
which the (square root of the) \emph{inverse SNR} is incorporated.\footnote{%
	Instead of using the augmented matrix $\ve{\mathcal{H}}$, often the
	Cholesky square root $\ve{L}^\H$ of
	$\ve{L}\ve{L}^\H=(\ve{H}^\H\ve{H}+{\sigma_{n}^2}/{\sigma_{x}^2})$
	is applied to calculate the MMSE variant of the	channel transformation
	in~\eqref{eq:fac_ZF}. Both approaches are equivalent since
	$\ve{\mathcal{H}}$ is an alternative square-root of
	$\ve{L}\ve{L}^\H$ \cite{Stern:19,Fischer:19}.%
}
The filter matrix $\ve{F}$ for MMSE linear equalization is then given as the
$K \times N$ \emph{left part} of the $K \times (N+K)$ \emph{augmented filter matrix}
$\ve{\mathcal{F}}$, and the noise variances in \eqref{eq:noise_var} are determined
by the (complete) rows of $\ve{\mathcal{F}}$ instead of $\ve{F}$ \cite{Stern:19}.

The crucial performance criterion is usually the \emph{worst-link SNR}, i.e., the
lowest SNR among the $K$ data streams \emph{before decoding}. It dominates the error
curves in case of uncoded transmission (spatial \emph{diversity order} \cite{Tse:05})
as well as the coded performance expressed in achievable bit rate according to
Shannon \cite{Zhan:14}. Considering the worst-link SNR as the performance criterion
and applying the MMSE criterion according to~\eqref{eq:fac_MMSE}, the optimization
problem reads%
\begin{equation}	\label{eq:SBPfac}
	\underbrace{\ve{Z}}_{\ve{T}^\H}=\argmin_{{\ve{Z} \in \I^{K \times K}}}
	\max_{k=1,\dots,K} \big\{ 
	\|\underbrace{\ve{\mathcal{H}}^{+\H}}_{\ve{G}} \ve{z}_k^\H \|^2 \big\} \; ,
\end{equation}
where the rows of $\ve{Z}$ correspond to the columns of its Hermitian matrix
$\ve{Z}^\H= [\ve{z}_1^\H,\dots,\ve{z}_K^\H]$.
 
In the initial publications on lattice-based equalization
\cite{Yao:02,Windpassinger:03}, see also \cite{Wuebben:11}, the integer matrix
$\ve{Z}$ from~\eqref{eq:SBPfac} is restricted to the set of unimodular matrices,
i.e., $\det(\ve{Z}^\H\ve{Z})=1$, in order to ensure the existence of an inverse
matrix $\ve{Z}^{-1}\in\I^{K \times K}$ for integer equalization. Then,
\eqref{eq:SBPfac} corresponds to the SBP as defined in~\eqref{eq:sbp} which can be
solved (approximately) by lattice-basis-reduction algorithms like the LLL one.

However, this unimodularity constraint is actually not required
\cite{Zhan:14,Fischer:19}. If $\ve{Z}$ describes a
\emph{full-rank integer linear combination} of the transmit symbols, i.e., if the
constraint $\rank(\ve{Z})=K$ is imposed, $\ve{G}=\ve{\mathcal{F}}^\H$ may only
define a \emph{sublattice} of $\ve{G}=\ve{\mathcal{H}}^{+\H}$,
cf.\ Sec.~\ref{subsec:generalized_lattice}. Hence, after linear equalization
via~$\ve{F}$, the vector of linear combinations of the transmit symbols,
$\ve{Z}\ve{x}$, may be drawn from a \emph{subspace} of $\I^K$. Nevertheless, only
valid lattice points are obtained. Using a suited lattice-decoding strategy, the
linear combinations can still successfully be reconstructed, and the non-unimodular
relaxation does not impair the equalization approach. Consequently, the
SIVP~\eqref{eq:sivp} has to be solved w.r.t.\ the generator matrix
$\ve{G}=\ve{\mathcal{H}}^{+\H}$. As discussed in Sec.~\ref{sec:algorithms}, this is
optimally performed by the calculation of the successive-minima vectors.

\subsection{Diversity Orders and Asymptotic Rates}	\label{subsec:div_rate}

The \emph{diversity order} describes the slope of the symbol error curve of
\emph{uncoded} transmission if
\emph{the average over all possible channel realizations and users}
is considered. Given the symbol error ratio (SER) over the SNR
$\sigma_{x}^2/\sigma_{n}^2$, it is defined as \cite{Tse:05}%
\begin{equation}
	\Delta = -\lim_{\frac{\sigma_{x}^2}{\sigma_{n}^2}\to \infty} 
		\frac{\log_{10}\mathrm{SER}}
		{\log_{10}(\frac{\sigma_{x}^2}{\sigma_{n}^2})} \; . 
\end{equation}
Hence, the SER drops by $\Delta$ decades per $10~\dB$ SNR increase.

It has been proven in~\cite{Taherzadeh:07} that LRA equalization exploits the MIMO
channel's receive diversity. To this end, a bound on the product of the norms of the
basis vectors over the volume has to exist. According to
Theorem~\ref{th:LLLPbound_general}, this property is fulfilled for all (generalized)
variants of LLL reduction and the related successive-minima vectors.

The achievable receive diversity can be generalized to%
\begin{equation}	\label{eq:diversity}
	\Delta_{\mathrm{rec}} = \frac{D_\mathsf{r}}{2} N  \; ,
\end{equation}
where $D_\mathsf{r}=1$ is valid for real, $D_\mathsf{r}=2$ for complex, and
$D_\mathsf{r}=4$ for quaternionic Gaussian MIMO channels. As a consequence, the diversity
is doubled for quaternionic channels in comparison to complex ones. This can either be
derived from the fact that the number of independent Gaussian random variables per scalar
coefficient is doubled as well \cite{Stierstorfer:09}, or from the point of view that
its equivalent complex-valued MIMO representation~\eqref{eq:equiv_compl_MIMO} actually 
forms an Alamouti code---if the complex channel gains are independent in both planes,
the diversity in the Alamouti scheme is doubled \cite{Alamouti:98}, see also the
discussion in~\cite{Qureshi:18}.

The \emph{asymptotic (bit) rate} often serves as a quantity for quality assessment in
\emph{coded transmission}. It describes the maximum bit rate
\emph{for the worst-case user}\footnote{%
	If all users are assumed/forced to employ the same channel code and the same signal
	constellation, they usually share the same rate, see, e.g., \cite{Zhan:14}.%
}
and \emph{one particular channel realization}. In the LRA/IF setting, it is given as
the Shannon capacity for an infinite-dimensional modulo channel
\cite{Forney:00,Zhan:14}. In the setup at hand, it is universally expressed by%
\begin{equation}	\label{eq:rate}
	R = \frac{D_\mathsf{r}}{2} \max_{k=1,\dots,K} 
		\log_2\left(\frac{\sigma_{x}^2}{\|\ve{f}_k\|^2\cdot\sigma_{n}^2}\right) \; .
\end{equation}

\subsection{Lattice Constellations and Soft-Decision Decoding}

We briefly review different types of lattice constellations, i.e., finite subsets
$\mathcal{A}\subset\I$ that are used as signal sets in lattice-based MIMO transmission,
including their properties and gains.

In complex transmission, most often QAM constellations are employed that form (shifted)
subsets of $\G$. Besides, especially in optical communications, \emph{dual-polarized QAM}
\cite{Karlsson:16} is popular which corresponds to (shifted) subsets of $\L$ if
interpreted over quaternion space. Alternatively, lattice-based constellations can be
defined as $\mathcal{A}\subset\E$ \cite{Tunali:15,Stern:15,Stern:19,Lyu:20} and
$\mathcal{A}\subset\Hu$
\cite{Karlsson:16,Stern:18,Frey:20,Guzeltepe:13,Guzeltepe:14,Rohweder:21}, respectively.
Since these rings constitute denser packings than their related $\Z^{D_\mathsf{r}}$-based
ones, the number of signal points within a certain hypervolume is increased, i.e., a
\emph{packing gain} is achieved. This gain can either be used to increase the cardinality
(data rate) keeping the SNR fixed, or to lower the SNR (constellation's variance) for a
fixed cardinality/rate. Noteworthy, in LRA/IF equalization, this packing gain directly
adds up to the SNR gain obtained from solving the SBP/SMP over the particular ring $\I$
(\emph{factorization} or \emph{equalization gain}), cf.\ Sec.~\ref{sec:bounds}.

The packing gain has to be distinguished from the \emph{shaping gain} that is present if
the finite constellation is not bounded by a hypercube but by other geometric shapes
like a hexagonal Voronoi cell or even hyperspheres. This gain does, in general, not
depend on the particular signal lattice $\I$ (and is not even restricted to lattice-based
constellations). It additionally adds up to the other two gains mentioned above.

If lattice-based equalization is combined with hard-decision decoding, no additional
constraint on the (lattice) constellation is imposed; the decoding operation in
Fig.~\ref{fig:if_lra} is realized by the quantization $\Q_\I\{\cdot\}$ followed by the
decoding algorithm after integer equalization. The situation changes when soft-decision
decoding is applied: in the LRA/IF receiver, linear combinations of transmit symbols have
to be decoded directly, see also \cite{Zhan:14,Fischer:19,Stern:19}. Two different
strategies have been derived to handle this situation: As initially proposed in the IF
concept \cite{Zhan:14}, an isomorphism can be established between the constellation points
and the code symbols. It is enabled by $p$-ary \emph{algebraic constellations}, $p$ prime,
that are constructed by a special interaction between signal lattice and shaping region
(see above). Such constellations have been proposed and assessed over Gaussian and
Eisenstein integers in \cite{Huber:94a,Huber:94b,Stern:15,Stern:19}. Moreover, first
results on Hurwitz-based algebraic constellations are available
\cite{Guzeltepe:13,Guzeltepe:14,Rohweder:21}. An alternative approach is based on a special
kind of multilevel-coding scheme \cite{Chae:17,Fischer:18} in which the addition of
Gaussian integers corresponds to the addition with carry over binary arithmetic
\cite{Fischer:18} and the addition of Eisenstein integers to the addition with carry over
ternary arithmetic \cite{Stern:19b}. An adaption to other rings like the Hurwitz integers
is still an open point.

\subsection{Numerical Evaluation and Comparison}

The performance of the generalized lattice algorithms in the MIMO uplink scenario is
finally assessed by means of numerical simulations. To this end, $10^6$ i.i.d.\ complex
and quaternionic Gaussian \emph{channel matrices} have been generated according
to~\eqref{eq:MIMO_channel} and~\eqref{eq:MIMO_channel_H}, respectively. Please note
that---in contrast to the numerical evaluations in Sec.~\ref{subsec:numBounds}---those
matrices do not directly form the generator matrices of the lattices to be considered.
Instead, since we assume that the MMSE criterion is applied and that the filter matrix is
calculated as defined in~\eqref{eq:fac_MMSE}, the generator matrices are given by the
pseudoinverses of the respective augmented channel matrices
(\emph{MMSE dual-lattice approach}, cf.~\cite{Fischer:19}). Hence, the statistical
distributions differ from the straightforward evaluation of i.i.d.\ Gaussian matrices.
The channel equalization is performed using the LRA/IF receiver as explained in
Sec.~\ref{subsec:lra}. For all kinds of LLL reduction, the optimal quality parameter
$\delta=1$ is assumed once again.

\subsubsection{Symbol Error Ratios in Uncoded Transmission}

We start with the assessment of the different SERs which are obtained in case of uncoded
transmission. These curves are not only relevant in terms of the quality of the different
lattice algorithms, but are also suited the evaluate the different diversity behavior,
cf.\ Sec.~\ref{subsec:div_rate}. In order to simulate a block-fading environment, bursty
transmissions with $10^3$ symbols/noise samples per channel matrix have been performed.
Due to uncoded transmission, the decoders in Fig.~\ref{fig:if_lra} are realized by the
quantization to the particular integer ring $\I$, see above. The data symbols are finally
estimated by an ML decision based on the signal constellation $\mathcal{A}$. Please note
that we focus on the comparison of the \emph{equalization performance}, i.e., we are not
interested in shaping or packing gains obtained by different constellations. Since, in the
comparisons, all constellations posses the same variance $\sigma_{x}^2$, all gains are
enabled by a superior performance of the algorithms given particular rings.

First, we restrict to the complex-valued case. For lattices defined over the Gaussian
integers, the zero-mean $4$-ary QAM constellation $\mathcal{A}_\G=\{\pm1\pm\i\}\subset\G$
with variance $\sigma_{x}^2=2$ has been employed. Since its components are independent,
they can individually be processed using the equivalent real-valued representation
according to~\eqref{eq:equiv_real_compl}, i.e., if algorithms over $\Z$ are
considered.\footnote{%
	Please note that, in this work, one symbol error is always defined w.r.t.\ the
	original complex or quaternion-valued signal constellation $\mathcal{A}$.%
}
Hence, the same constellation is present in both representations. Unfortunately,
$\mathcal{A}_\G\not\subset\E$, i.e., we cannot employ this constellation for lattices
over $\E$ any more. Thus, the zero-mean signal constellation
$\mathcal{A}_\E=\{1,-1,\sqrt{3}\,\i,-\sqrt{3}\,\i\}$, with $\sigma_{x}^2=2$, has been
used. Since both signal constellations possess the same variance, a fair comparison is
enabled, i.e., a packing or shaping gain is not present.

\begin{figure}
		\centerline{%
			\includegraphics{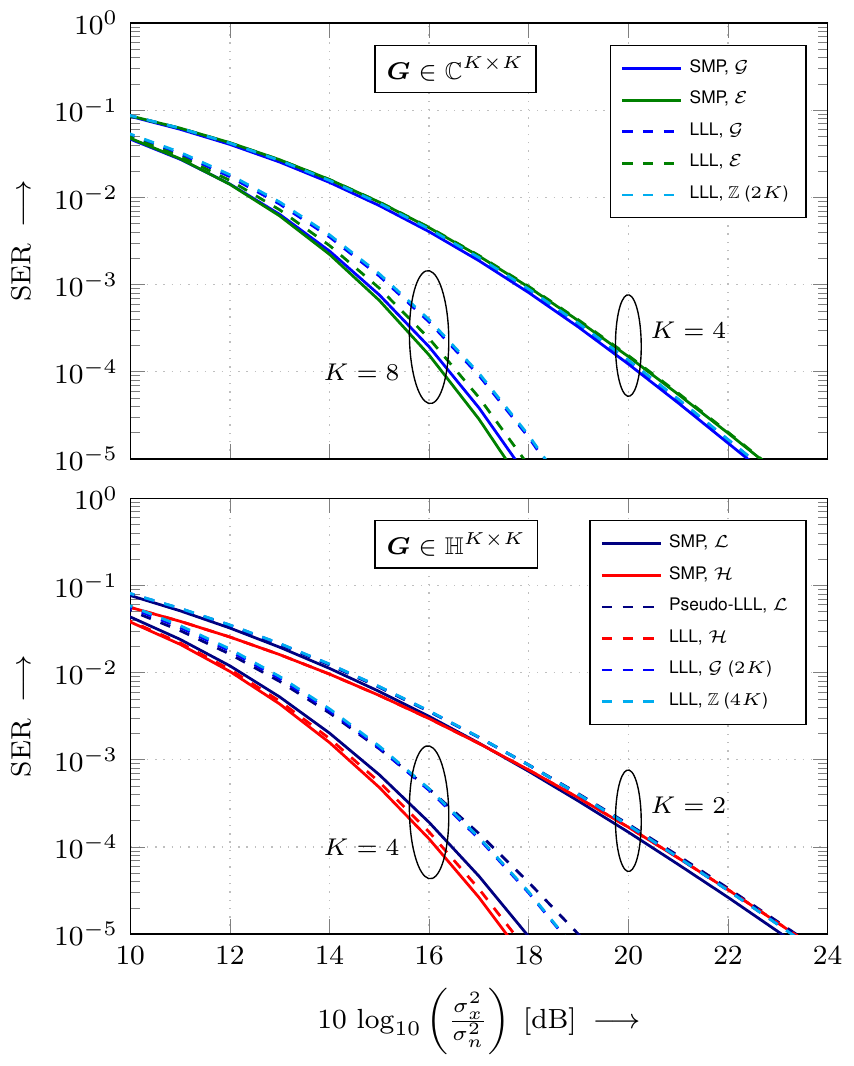}%
		}%
	\caption{\label{fig:sim}Symbol error ratio over the SNR in dB for $K \times K$
		uncoded MIMO uplink transmission and LRA/IF equalization (MMSE criterion)
		over different integer rings and algorithms obtained by numerical simulations.
		For LLL reduction, the optimal parameter $\delta=1$ is assumed. Top:
		complex-valued transmission with the signal constellations
		$\mathcal{A}_\G=\{\pm1\pm\i\}$ and
		$\mathcal{A}_\E=\{1,-1,\sqrt{3}\,\i,-\sqrt{3}\,\i\}$. Bottom:
		quaternion-valued transmission with the constellation
		$\mathcal{A}_\L=\{\pm 1  \pm\i  \pm\j \pm \k\}$.}%
\end{figure}

In Fig.~\ref{fig:sim} (Top), the SER is plotted over the SNR for complex transmission
and two scenarios with $N=K=4$ and $N=K=8$, respectively. The expected diversity orders
$\Delta_\mathrm{rec}=4$ and $\Delta_\mathrm{rec}=8$ are clearly visible. If $N=K=4$,
all curves are still quite similar. However, in the high-SNR regime, the use of the
Eisenstein lattice may even be disadvantageous. This is due to the fact that the
application of the particular algorithms yields almost the same result since the
dimensions are low, i.e., nearly the same (maximum) noise enhancement is present over
$\G$ and $\E$. However, the quantization to $\E$ results in a higher probability for
false quantization due to six nearest neighbors instead of four
(cf.\ Sec.~\ref{sec:extensions}). In coded transmission, this effect becomes irrelevant.
Regarding the case $N=K=8$, the decreased maximum noise enhancement for $\E$ is now the
dominating part. It has a positive impact such that a horizontal (SNR) gain is achieved if
the ELLL algorithm is applied or if the respective successive-minima vectors are calculated.
In summary, as expected from the results provided in Sec.~\ref{sec:bounds}, the potential
gains increase with the dimensions even though the statistical model slightly differs from a
Gaussian one (MMSE dual-lattice approach~\eqref{eq:fac_MMSE}).

Next, we consider the case of quaternion-valued transmission. To this end, the $16$-ary
constellation $\mathcal{A}_\L=\{\pm 1  \pm\i  \pm\j \pm \k\}$ ($4$QAM per complex component)
with $\sigma_{x}^2=4$ has been employed for which
$\mathcal{A}_\L\subset\L \simeq \G^2\simeq\Z^4$ as well as $\mathcal{A}_\L\subset\Hu$ are
valid, since $\L\subset\Hu$. Hence, this constellation can be used in combination with lattices
over $\L$, $\Hu$, $\G$, and $\Z$. The related SER curves are illustrated in Fig.~\ref{fig:sim}
(Bottom) for two particular scenarios with $N=K=2$ and $N=K=4$. First, we see that, in
accordance with~\eqref{eq:diversity}, the diversity orders are the same as in Fig.~\ref{fig:sim}
(Top), even though the dimensions $K$ are halved. For $N=K=2$, the curves are again quite the
same. Moreover, we also have the effect that the quantization to $\Hu$ may be more erroneous
than the one to $\L$ due to the increased number of nearest neighbors (24 instead of 8). In
contrast, for $N=K=4$, the Hurwitz lattice achieves a significant gain over $\L$, or its complex
and real-valued equivalents, which perform nearly the same if applied in combination with LLL
reduction. The pseudo-LLL reduction shows the worst performance. However---even though
respective theoretical bounds cannot be derived---it still seems to approach the same diversity
behavior in practice.

\subsubsection{Bit Rates in Coded Transmission}

Finally, the performance is assessed w.r.t.\ coded transmission for which the rate according
to~\eqref{eq:rate}---depending on the particular noise enhancement (maximum squared norm) and
the SNR---is the relevant quantity. In Fig.~\ref{fig:rate}, both the expectations and the
$0.01$-quantiles of the achievable rates are illustrated, assuming
$\sigma_{x}^2/\sigma_{n}^2 \corr 20~\dB$ (MMSE criterion). In particular, the latter represent
the rates for which an outage of $1~\%$ can be expected.

\begin{figure}
	\centerline{%
		\includegraphics{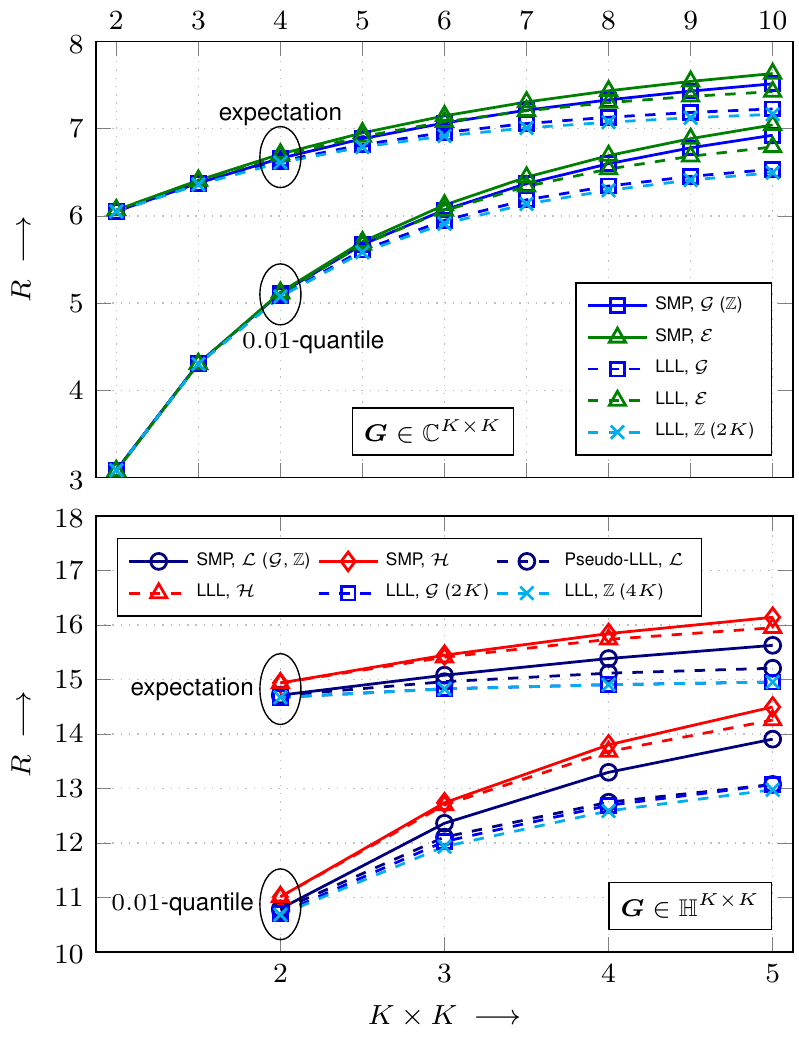}%
	}%
	\caption{\label{fig:rate}Achievable rates for lattice-based linear equalization
		according to~\eqref{eq:rate} in coded transmission 
		($\sigma_{x}^2/\sigma_{n}^2 \corr 20~\dB$, MMSE criterion) over the
		$K \times K$ MIMO channel for different integer rings and algorithms
		obtained by numerical simulations. Both the expectations and the
		$0.01$-quantiles are shown. Top: complex-valued transmission. Bottom:
		quaternion-valued transmission.}%
\end{figure}

In the complex-valued case (Fig.~\ref{fig:rate} (Top)), it is visible that the use of
lattices over Eisenstein integers enables a slight gain w.r.t\ the achievable rate, which
increases with the dimensions in accordance with Sec.~\ref{sec:bounds}. Just like in the
numerical analysis of the maximum norms in Fig.~\ref{fig:statNorms} for i.i.d.\ Gaussian
matrices, the RLLL algorithm performs the worst.

In quaternion-valued transmission (Fig.~\ref{fig:rate} (Bottom)), lattices over Hurwitz
integers significantly enhance the achievable rates. The QLLL reduction over $\Hu$
approximates the solutions to the SMP quite well. Moreover, it is quite interesting that---by
analogy with the curves in Fig.~\ref{fig:statNorms}---, the pseudo-QLLL reduction over $\L$
performs substantially better than its complex or real-valued ``genuine'' counterparts.
Hence, in practice, its usage may be taken into account---even though no theoretical
performance guarantees can be given.

Summarizing the above observations, we can conclude that the main results from
Sec.~\ref{sec:bounds} also hold for the channel model at hand. This especially concerns
the potential gains which increase with the dimensions of the MIMO channel. In addition,
by performing a related complexity analysis, one would obtain results almost identical
to the ones provided in Figs.~\ref{fig:statList} and~\ref{fig:mulStat}.

\section{Summary and Outlook}	\label{sec:sum}

\noindent In this paper, algorithms and bounds known from the field of \emph{real-valued}
lattice problems have been generalized and adapted to operate over complex and quaternion
numbers. To this end, a review of the particular arithmetic and the properties of these
number sets has been given first. Then, generalized variants of LLL reduction and a
list-based algorithm to determine the successive minima of a lattice have been given. In
addition to lattices over the set of (real-valued) integers $\Z$, they can operate over
the Gaussian (complex integers) as well as the Eisenstein integers for the complex case,
and the Lipschitz as well as the Hurwitz integers for the quaternion-valued case. For all
of these integers sets, bounds on the length of the first basis vector obtained from LLL
reduction as well as on the first successive minimum have been derived. These bounds were
complemented by bounds for the particular orthogonality defects, incorporating all basis
or successive-minima vectors. The provided results indicate that lattices over the
Eisenstein integers (instead of the Gaussian integers) may be beneficial in the
complex-valued case, and that lattices over the Hurwitz integers (instead of the Lipschitz
integers) may be of advantage in the quaternion-valued case. When running the presented
successive-minima algorithm over these rings, the expected complexity is only increased a
little bit in comparison to their counterparts. Moreover, when the LLL reduction is applied
over these complex or quaternion-valued integer rings, the expected complexity is
significantly decreased in comparison to an equivalent real-valued reduction. Finally,
particular application scenarios have been identified. These considerations included their
use in the field of MIMO communications, in particular in lattice-based equalization for
the MIMO uplink channel, where gains can be expected if lattices and related constellations
over Eisenstein and Hurwitz integers, respectively, are employed.

Future work could deal with the extension to the multi-dimensional, e.g., the
eight-dimensional case, in which several time steps or frequencies could be combined to one
symbol. In addition, the adaption of other criteria/algorithms for lattice basis reduction,
e.g., HKZ \cite{Korkine:73} or Minkowski \cite{Minkowski:91} reduction, could be studied and
related bounds could be derived for the complex and/or quaternion-valued case. Moreover,
since algebraic constellations are required for the use in the IF concept, the initial work
on four-dimensional ones in \cite{Guzeltepe:13,Guzeltepe:14,Rohweder:21} could be extended
to derive constellations similar to the complex ones in
\cite{Huber:94a,Huber:94b,Stern:15,Stern:19}. The same holds for quaternionic
multilevel-coding schemes similar to~\cite{Chae:17,Fischer:18,Stern:19b}. Besides, a deeper
analysis of the pseudo-QLLL algorithm and related bounds could be given and other
non-Euclidean rings or algebraic structures similar to~\cite{Lyu:20} be studied.

\appendices

\section{Helper Functions for the List-Based Determination of the Successive Minima}
\label{app:helperSMP}

\begin{algorithm}[t]{
		\small
		\caption{\label{alg:toZ} Equivalent Real-Valued Representation over $\Z$.}
		$\ve{G}_{\mathsf{r},\I} = \textsc{RingToZ}(\ve{G},\I)$
		\begin{algorithmic}[1]
			\Switch {$\I$}
				\Case{$\Z$}				
					\State{$\ve{G}_{\mathsf{r},\I} = \ve{G}$ }
				\EndCase
				\Case{$\G$}
					{\color{\commentcolor}\Comment{\eqref{eq:equiv_real_compl}}}
					\State{$\ve{G}_{\mathsf{r},\I} =
						\begin{bmatrix} \setlength{\tabcolsep}{3pt}
							\ve{G}^{(1)} & -\ve{G}^{(2)}\\
							\ve{G}^{(2)} & \phantom{-}\ve{G}^{(1)}
						\end{bmatrix}$ }
				\EndCase
				\Case{$\E$}
					{\color{\commentcolor}\Comment{\eqref{eq_equiv_real_E}}}
					\State{$\ve{G}_{\mathsf{r},\I} = 
						\begin{bmatrix} 
							\ve{G}^{(1)} & -\ve{G}^{(2)}\\
							\ve{G}^{(2)} & \phantom{-}\ve{G}^{(1)}
						\end{bmatrix}
						\begin{bmatrix}
							\ve{I}_K & -\frac{1}{2}\ve{I}_K\\
							\ve{0}_K & \frac{\sqrt{3}}{2} \ve{I}_K
						\end{bmatrix} $ }
				\EndCase
				\Case{$\L$}
					{\color{\commentcolor}\Comment{\eqref{eq:equiv_real_quat}}}
					\State{$\ve{G}_{\mathsf{r},\I} =
							\begin{bmatrix}
							\ve{G}^{(1)} & -\ve{G}^{(2)} & 
								-\ve{G}^{(3)} & -\ve{G}^{(4)} \\
							\ve{G}^{(2)} & \phantom{-}\ve{G}^{(1)} & 
								-\ve{G}^{(4)} & \phantom{-}\ve{G}^{(3)} \\
							\ve{G}^{(3)} & \phantom{-}\ve{G}^{(4)} & 
								\phantom{-}\ve{G}^{(1)} & -\ve{G}^{(2)} \\
							\ve{G}^{(4)} & -\ve{G}^{(3)} & 
								\phantom{-}\ve{G}^{(2)} & \phantom{-}\ve{G}^{(1)} \\
						\end{bmatrix}$ }
				\EndCase
				\Case{$\Hu$}
					{\color{\commentcolor}\Comment{\eqref{eq:equiv_real_hur}}}
					\State{$\ve{G}_{\mathsf{r},\I} =
							\begin{bmatrix}
								\ve{G}^{(1)} & -\ve{G}^{(2)} & 
									-\ve{G}^{(3)} & -\ve{G}^{(4)} \\
								\ve{G}^{(2)} & \phantom{-}\ve{G}^{(1)} & 
									-\ve{G}^{(4)} & \phantom{-}\ve{G}^{(3)} \\
								\ve{G}^{(3)} & \phantom{-}\ve{G}^{(4)} & 
									\phantom{-}\ve{G}^{(1)} & -\ve{G}^{(2)} \\
								\ve{G}^{(4)} & -\ve{G}^{(3)} & 
									\phantom{-}\ve{G}^{(2)} & \phantom{-}\ve{G}^{(1)} \\
							\end{bmatrix} \cdot$ }
					\Statex {\qquad \qquad\qquad\qquad$\begin{bmatrix}
							\ve{I}_K &\ve{0}_K &\ve{0}_K & \frac{1}{2}\ve{I}_K \\
							\ve{0}_K &\ve{I}_K &\ve{0}_K & \frac{1}{2}\ve{I}_K \\
							\ve{0}_K &\ve{0}_K &\ve{I}_K & \frac{1}{2}\ve{I}_K \\
							\ve{0}_K &\ve{0}_K &\ve{0}_K & \frac{1}{2}\ve{I}_K \\
						\end{bmatrix}$}
				\EndCase
			\EndSwitch
		\end{algorithmic}
	}
\end{algorithm}

\noindent In this appendix, some procedures are listed which are incorporated in
Algorithm~\ref{alg:smp} (determination of the successive minima).

Algorithm~\ref{alg:toZ} creates the equivalent real-valued representation of the
generator matrix $\ve{G}$ depending on the integer ring~$\I$. These representations
have been derived in Sec.~\ref{sec:extensions}.

\begin{algorithm}[tb]{
		\small
		\caption{\label{alg:fromZ} Reconversion from Representation over $\Z$.} %
		$\ve{C}_\mathrm{u} = \textsc{ZToRing}(\ve{C},\I)$
		\begin{algorithmic}[1]
			\Switch {$\I$}
			\Case{$\Z$}
			\State{$\ve{C}_\mathrm{u} = \ve{C}$ }
			\EndCase
			\Case{$\G$}
			\State{$\ve{C}_\mathrm{u}^{(1)}=\ve{C}_{1:K,1:N_\mathrm{c}}$}
			\State{$\ve{C}_\mathrm{u}^{(2)}=\ve{C}_{K+1:2K,1:N_\mathrm{c}}$}
			\EndCase
			\Case{$\E$}
			\State{$\ve{\tilde{C}}=
				\begin{bmatrix}
					\ve{I}_K & -\frac{1}{2}\ve{I}_K\\
					\ve{0}_K & \frac{\sqrt{3}}{2} \ve{I}_K
				\end{bmatrix} 
				\cdot \ve{C}$ }
			\State{$\ve{C}_\mathrm{u}^{(1)}=\ve{\tilde{C}}_{1:K,1:N_\mathrm{c}}$}
			\State{$\ve{C}_\mathrm{u}^{(2)}=\ve{\tilde{C}}_{K+1:2K,1:N_\mathrm{c}}$}
			\EndCase
			\Case{$\L$}
			\State{$\ve{C}_\mathrm{u}^{(1)}=\ve{C}_{1:K,1:N_\mathrm{c}}$}
			\State{$\ve{C}_\mathrm{u}^{(2)}=\ve{C}_{K+1:2K,1:N_\mathrm{c}}$}
			\State{$\ve{C}_\mathrm{u}^{(3)}=\ve{C}_{2K+1:3K,1:N_\mathrm{c}}$}
			\State{$\ve{C}_\mathrm{u}^{(4)}=\ve{C}_{3K+1:4K,1:N_\mathrm{c}}$}
			\EndCase
			\Case{$\Hu$}
			\State {$\ve{\tilde{C}}= \begin{bmatrix}
					\ve{I}_K &\ve{0}_K &\ve{0}_K & \frac{1}{2}\ve{I}_K \\
					\ve{0}_K &\ve{I}_K &\ve{0}_K & \frac{1}{2}\ve{I}_K \\
					\ve{0}_K &\ve{0}_K &\ve{I}_K & \frac{1}{2}\ve{I}_K \\
					\ve{0}_K &\ve{0}_K &\ve{0}_K & \frac{1}{2}\ve{I}_K \\
				\end{bmatrix}\cdot \ve{C}$}
			\State{$\ve{C}_\mathrm{u}^{(1)}=\ve{\tilde{C}}_{1:K,1:N_\mathrm{c}}$}
			\State{$\ve{C}_\mathrm{u}^{(2)}=\ve{\tilde{C}}_{K+1:2K,1:N_\mathrm{c}}$}
			\State{$\ve{C}_\mathrm{u}^{(3)}=\ve{\tilde{C}}_{2K+1:3K,1:N_\mathrm{c}}$}
			\State{$\ve{C}_\mathrm{u}^{(4)}=\ve{\tilde{C}}_{3K+1:4K,1:N_\mathrm{c}}$}
			\EndCase
			\EndSwitch
		\end{algorithmic}
	}
\end{algorithm}

In Algorithm~\ref{alg:fromZ}, the candidate integer vectors are reconverted from the
equivalent real-valued representation to the representation in the particular ring $\I$.
To this end, for the Gaussian and the Lipschitz integers, the components are simply
stacked in the matrix $\ve{C}$. For the Eisenstein and the Hurwitz integers, the
particular generator matrices according to~\eqref{eq_equiv_real_E}
and~\eqref{eq:equiv_real_hur}, respectively, have to be incorporated first.

\begin{algorithm}[tb]{
		\small
		\caption{\label{alg:re} Transformation to Row-Echelon Form.} %
		$\ve{i} = \textsc{RowEchelon}(\ve{C})$
		\begin{algorithmic}[1]
			\State {$k=1$, $l=1$, $\ve{i}=[0,\dots,0]$ }
			\While {$k \leq K$}
			\For{$m=k,\dots,K$}
			\If{$c_{m,l}\neq 0$}
			\State {$i_k = l$}
				{\color{\commentcolor}\Comment{independent vector found}}
			\State {$\tilde{\ve{c}}= \ve{C}_{k,1:N_\mathrm{c}}$}
				{\color{\commentcolor}\Comment{interchange rows}}
			\State {$\ve{C}_{k,1:N_\mathrm{c}}=\ve{C}_{m,1:N_\mathrm{c}}$}
			\State {$\ve{C}_{m,1:N_\mathrm{c}}=\tilde{\ve{c}}$}
			\State {$\ve{C}_{k,1:N_\mathrm{c}}=c_{k,l}^{-1} \cdot
				\ve{C}_{k,1:N_\mathrm{c}}$}
				{\color{\commentcolor}\Comment{normalize $k$\textsuperscript{th} row}}
			\For{$n=k+1,\dots,K$}
				{\color{\commentcolor}\Comment{eliminate successors}}
			\State {$\ve{C}_{n,1:N_\mathrm{c}}=
					\ve{C}_{n,1:N_\mathrm{c}} - c_{n,l}\ve{C}_{k,1:N_\mathrm{c}}$}
			\EndFor
			\State{$k = k + 1$}
			\State {\bf break}
			\EndIf
			\EndFor
			\State{$l = l + 1$}
			\EndWhile
		\end{algorithmic}
	}
\end{algorithm}

In Algorithm~\ref{alg:re}, the matrix of candidate vectors $\ve{C}$ is transformed
to row-echelon form. To this end, for each candidate with the index $l$, it is
checked if a new dimension is established. This is the case when one of the elements
$c_{k,l},\dots,c_{K,l}$ is not zero; then, the vector does not depend on the previous
ones. Given that case, the particular row with the non-zero element is interchanged
with the row $k$. After normalization\footnote{%
	Please note that linearly independent lattice vectors are required, i.e.,
	independent vectors have to be present over $\R$, $\C$, or $\Ha$. Hence,
	when calculating the row-echelon form for the integer vectors, non-integer
	elements may occur. Nevertheless, these non-integer elements are not relevant
	since only the ``steps'' within the row-echelon form are of interest.
	Alternatively, the calculation of the row-echelon form can directly be
	performed with the related lattice vectors.%
}
to $c_{k,l}=1$, all other elements $c_{k+1,l},\dots,c_{K,l}$ are set to zero by
subtracting $c_{n,l}$ times the row with index $k$. All multiplications are performed
in such a way that the skew-field property of quaternions is taken into account.

\section{Comparison of the Upper Bound on the First Vector Norm for Different LLL Variants} 
	\label{app:LLLbounds}

\noindent In this appendix, the upper bound on the squared length of the first basis
vector in case of LLL reduction as described in Sec.~\ref{subsec:lll_bounds} is compared
for different Euclidean integer rings.

\subsection{Derivation of Corollary~\ref{cor:norms}}

We start with the derivation of Corollary~\ref{cor:norms}. It contains a generalized
variant of the comparison between RLLL and CLLL algorithm that has been given
in~\cite{Gan:09}.

Considering the upper bound in~\eqref{eq:LLLbound_general}, we can first state that its
right-hand part, $\vol^{\frac{2}{K}}(\Lat(\ve{G}))$, is irrelevant for the comparison of
all considered LLL variants since
$\vol^{\frac{2}{2K}}(\Lat(\ve{G}_\mathsf{r}))=\vol^{\frac{2}{K}}(\Lat(\ve{G}))$ holds for
the $2N\times 2K$ real-valued representation of complex matrices and
$\vol^{\frac{2}{4K}}(\Lat(\ve{G}_\mathsf{r}))=
	\vol^{\frac{2}{2K}}(\Lat(\ve{G}_\mathsf{c}))=\vol^{\frac{2}{K}}(\Lat(\ve{G}))$
for the $4N\times4K$ real-valued and the $2N\times2K$ complex-valued representations of
quaternionic matrices.

Hence, given the first LLL variant with the maximum squared quantization error
$\epsilon_{1}^2$ that operates on a matrix of rank $K_1$, and the second variant with the
error $\epsilon_{2}^2$ operating over a matrix of rank $K_2$, the first one performs
better if%
\begin{equation}
	 \frac{(\delta - \epsilon_{2}^2)^{K_2-1}}{(\delta - \epsilon_{1}^2)^{K_1-1}} <1 
\end{equation}
since both terms within the braces are positive.

\subsection{Complex Matrices}

Based on Corollary~\ref{cor:norms}, the different LLL variants can be compared in an
asymptotic manner. We start with the comparison given complex generator matrices
$\ve{G}\in\C^{N \times K}$.

When comparing the ELLL ($\I=\E$) with the CLLL ($\I=\G$), we have $K_1=K_2$,
$\epsilon_{1}^2=\frac{1}{3}$, and $\epsilon_{2}^2=\frac{1}{2}$. Hence, we obtain%
\begin{equation}
		 \lim_{K_1 \to \infty}\frac{(\delta - \frac{1}{2})^{K_1 - 1}}
		 	{(\delta - \frac{1}{3})^{K_1 - 1}} < 1 \; ,
\end{equation}
i.e., the ELLL performs better independently from $\delta$.

Given the CLLL ($\I=\G$) and the RLLL ($\I=\Z$), $K_2=2K_1$,
$\epsilon_{1}^2=\frac{1}{2}$, and $\epsilon_{2}^2=\frac{1}{4}$.
Corollary~\ref{cor:norms} is given as%
\begin{equation}	\label{eq:clllvsrlll}
	\frac{(\delta - \frac{1}{4})^{2K_1}}{(\delta - \frac{1}{2})^{K_1}} \cdot
	\underbrace{\frac{(\delta - \frac{1}{2})}{(\delta - \frac{1}{4})}}_{<1}	\; .
\end{equation}
Hence, the CLLL bound is smaller if $(\delta-\frac{1}{4})^2< \delta - \frac{1}{2}$.
This condition is, however, never fulfilled for $\delta\in(\frac{1}{2},1]$, but at
least equality can be achieved if and only if $\delta=\frac{3}{4}$, see
also~\cite{Gan:09}.

In the same way, the comparison ELLL ($\I=\E$) vs.\ RLLL ($\I=\Z$), with $K_2=2K_1$,
$\epsilon_{1}^2=\frac{1}{3}$, and $\epsilon_{2}^2=\frac{1}{4}$ can be performed. By
analogy with \eqref{eq:clllvsrlll}, the condition
$(\delta-\frac{1}{4})^2< \delta - \frac{1}{3}$ is obtained. It is fulfilled if
$\frac{3}{4}-\frac{1}{\sqrt{6}}\approx0.34 < \delta \leq 1$.

\subsection{Quaternionic Matrices}

We continue with the comparison for the case when quaternionic matrices
$\ve{G}\in\Ha^{N \times K}$ are present.

For the comparison QLLL ($\I=\Hu$) vs.\ CLLL ($\I=\G$), the parameters $K_2=2K_1$
and $\epsilon_{1}^2=\epsilon_{2}^2=\frac{1}{2}$ are given, leading to%
\begin{equation}
	\lim_{K_1 \to \infty}\frac{(\delta - \frac{1}{2})^{2K_1 - 1}}
		{(\delta - \frac{1}{2})^{K_1 - 1}} < 1 \; .
\end{equation}
Hence, independent from $\delta$, the QLLL gives a better result.

Comparing the QLLL ($\I=\Hu$) with the ELLL ($\I=\E$), we have $K_2=2K_1$,
$\epsilon_{1}^2=\frac{1}{2}$, and $\epsilon_{2}^2=\frac{1}{3}$. Hence, similar
to~\eqref{eq:clllvsrlll}, we obtain the condition
$(\delta-\frac{1}{3})^2 < \delta - \frac{1}{2}$, which is fulfilled if
$\frac{5}{6}- \frac{1}{2\sqrt{3}}\approx0.54 < \delta \leq 1$.

Finally, for the comparison QLLL ($\I=\Hu$) vs.\ RLLL ($\I=\Z$), the parameters
$K_2=4K_1$, $\epsilon_{1}^2=\frac{1}{2}$, and $\epsilon_{2}^2=\frac{1}{4}$ are
present. Similar to above, we obtain the condition
$(\delta-\frac{1}{4})^4 < \delta - \frac{1}{2}$ that is fulfilled if
$0.504 < \delta \leq 1$.

\section{Comparison of the Upper Bound on the Orthogonality Defect for Different LLL~Variants}
	\label{app:LLLpbounds}

\noindent In this appendix, the upper bound on the orthogonality defect of the basis obtained
from LLL reduction (Sec.~\ref{subsec:lll_boundsP}) is asymptotically compared for different
Euclidean integer rings. To this end, the inequality condition from Corollary~\ref{cor:prod}%
\begin{equation}
	\frac{(\delta - \epsilon_{2}^2)^{{K_2(K_2-1)}{}}}
		{({\delta - \epsilon_{1}^2)^{{K_1(K_1-1)}{}}}} < 1
\end{equation}
is analyzed that follows directly from~\eqref{eq:LLLPbound_general} since both terms within
the braces are positive. Noteworthy, the same condition also holds for the bound on the product
of the basis vectors in~\eqref{eq:LLLPbound_general} in which the volume of the lattice is
additionally incorporated, see also the discussion in Appendix~\ref{app:LLLbounds}.

\subsection{Complex Matrices}

First, the performance of different LLL variants is considered for the case when complex
generator matrices $\ve{G}\in\C^{N \times K}$ are given.

When comparing the ELLL ($\I=\E$) with the CLLL ($\I=\G$), with $K_1=K_2$,
$\epsilon_{1}^2=\frac{1}{3}$, and $\epsilon_{2}^2=\frac{1}{2}$, we directly see that%
\begin{equation}
	\lim_{K_1 \to \infty}\frac{(\delta - \frac{1}{2})^{K_1(K_1 - 1)}}
		{(\delta - \frac{1}{3})^{K_1(K_1 - 1)}} < 1 \; ,
\end{equation}
i.e., the ELLL performs better independently from $\delta$.

Given the CLLL ($\I=\G$) and the RLLL ($\I=\Z$), $K_2=2K_1$, $\epsilon_{1}^2=\frac{1}{2}$,
and $\epsilon_{2}^2=\frac{1}{4}$. Then, the term reads%
\begin{equation}	\label{eq:clllvsrlllO}
	\frac{(\delta - \frac{1}{4})^{4K_1^2}}{(\delta - \frac{1}{2})^{K_1^2}} \cdot
		\underbrace{\frac{(\delta - \frac{1}{2})^{K_1}}
		{(\delta - \frac{1}{4})^{2K_1}}}_{\leq 1}	\; ,
\end{equation}
where the right part is less than or equal to one, cf.~\eqref{eq:clllvsrlll}. The CLLL bound
is lower if $(\delta-\frac{1}{4})^4< \delta - \frac{1}{2}$. This condition is fulfilled for
$0.504 < \delta \leq 1$.

The comparison ELLL ($\I=\E$) vs.\ RLLL ($\I=\Z$), with $K_2=2K_1$,
$\epsilon_{1}^2=\frac{1}{3}$, and $\epsilon_{2}^2=\frac{1}{4}$ can be performed similar
to~\eqref{eq:clllvsrlllO}. Then, for the left-hand term,
$(\delta - \frac{1}{4})^4 < \delta - \frac{1}{3}$ has to hold. However, the right-hand term is
not necessarily less than one, i.e., the product of both terms will be less than one if the
left one has this property and if
${(\delta-\frac{1}{4})^4}/({{\delta-\frac{1}{3}}}) <
	{(\delta-\frac{1}{4})}^2/({{\delta-\frac{1}{3}}})$.
The latter condition is always true since $\delta - \frac{1}{4} < 1$. The left-hand one holds
when $0.3334 < \delta \leq 1$.

\subsection{Quaternionic Matrices}

Finally, the bounds are compared for the case when quaternionic matrices
$\ve{G}\in\Ha^{N \times K}$ are present.

Comparing the QLLL ($\I=\Hu$) with the CLLL ($\I=\G$), the parameters $K_2=2K_1$
and $\epsilon_{1}^2=\epsilon_{2}^2=\frac{1}{2}$ result in%
\begin{equation}
	\frac{(\delta - \frac{1}{2})^{4K_1^2}}
		{(\delta - \frac{1}{2})^{K_1^2}} \cdot 
		\frac{(\delta - \frac{1}{2})^{K_1}}
		{(\delta - \frac{1}{2})^{2K_1}} \; .
\end{equation}
The QLLL generally results in a better bound: similar to above, the left-hand term is less
than one; the same holds for the product of both terms due to
$(\delta - \frac{1}{2})^3 < (\delta - \frac{1}{2})$ as $\delta - \frac{1}{2} < 1$.

The comparison QLLL ($\I=\Hu$) vs.\ ELLL ($\I=\E$) is done with $K_2=2K_1$,
$\epsilon_{1}^2=\frac{1}{2}$, and $\epsilon_{2}^2=\frac{1}{3}$.
Here, the conditions $(\delta - \frac{1}{3})^4 < \delta - \frac{1}{2}$ and
$(\delta - \frac{1}{3})^4 <  (\delta - \frac{1}{3})^2$ are obtained. The latter is true
and the former is fulfilled if $0.5008 < \delta \leq 1$.

For the comparison QLLL ($\I=\Hu$) vs.\ RLLL ($\I=\Z$) with $K_2=4K_1$,
$\epsilon_{1}^2=\frac{1}{2}$, and $\epsilon_{2}^2=\frac{1}{4}$, the conditions read
$(\delta - \frac{1}{4})^{16} < \delta - \frac{1}{2}$ and
$(\delta-\frac{1}{4})^{16} < (\delta - \frac{1}{4})^4$. Since the latter is again fulfilled,
see above, the former is relevant. It holds if $0.5000 < \delta \leq 1$.

\section*{Acknowledgment}
\noindent
The authors would like to thank the associate editor and the
anonymous reviewers for their valuable comments and suggestions
that helped a lot to improve the quality of this~work.
\ifCLASSOPTIONcaptionsoff
  \newpage
\fi
\bibliographystyle{IEEEtran}
\bibliography{IEEEabrv,ABCQL}

\begin{thebibliography}{10}
\providecommand{\url}[1]{#1}
\csname url@samestyle\endcsname
\providecommand{\newblock}{\relax}
\providecommand{\bibinfo}[2]{#2}
\providecommand{\BIBentrySTDinterwordspacing}{\spaceskip=0pt\relax}
\providecommand{\BIBentryALTinterwordstretchfactor}{4}
\providecommand{\BIBentryALTinterwordspacing}{\spaceskip=\fontdimen2\font plus
\BIBentryALTinterwordstretchfactor\fontdimen3\font minus
  \fontdimen4\font\relax}
\providecommand{\BIBforeignlanguage}[2]{{%
\expandafter\ifx\csname l@#1\endcsname\relax
\typeout{** WARNING: IEEEtran.bst: No hyphenation pattern has been}%
\typeout{** loaded for the language `#1'. Using the pattern for}%
\typeout{** the default language instead.}%
\else
\language=\csname l@#1\endcsname
\fi
#2}}
\providecommand{\BIBdecl}{\relax}
\BIBdecl

\bibitem{Stern:18}
S.~Stern and R.~F.~H. Fischer, ``Quaternion-valued multi-user {MIMO}
  transmission via dual-polarized antennas and {QLLL} reduction,'' in
  \emph{25th Int.\ Conf.\ on Telecommun.\ (ICT)}, Saint Malo, France, Jun.
  2018, pp. 63--69.

\bibitem{Stern:19}
S.~Stern, ``Advanced equalization and coded-modulation strategies for
  multiple-input/multiple-output systems,'' Ph.D. dissertation, Ulm University,
  Ulm, Germany, May 2019.

\bibitem{Hermite:50}
C.~Hermite, ``Extraits de lettres de {M.Ch. Hermite} {\`a} {M. Jacobi} sur
  diff{\'e}rents objets de la th{\'e}orie des nombres,'' \emph{Journal f{\"u}r
  die reine und angewandte Mathematik}, vol.~40, pp. 261--315, 1850.

\bibitem{Korkine:73}
A.~N. Korkine and J.~I. Zolotareff, ``Sur les formes quadratiques,''
  \emph{Mathematische Annalen}, vol.~6, pp. 366--389, 1873.

\bibitem{Minkowski:91}
H.~Minkowski, ``\BIBforeignlanguage{ger}{{\"U}ber die positiven quadratischen
  {Formen} und {\"u}ber kettenbruch{\"a}hnliche {Algorithmen}},''
  \emph{\BIBforeignlanguage{ger}{Journal f{\"u}r die reine und angewandte
  Mathematik}}, vol. 107, pp. 278--297, 1891.

\bibitem{Lenstra:82}
A.~K. Lenstra, H.~W. Lenstra, and L.~Lov{\'a}sz, ``Factoring polynomials with
  rational coefficients,'' \emph{Mathematische Annalen}, vol. 261, pp.
  515--534, Dec. 1982.

\bibitem{Kannan:83}
R.~Kannan, ``Improved algorithms for integer programming and related lattice
  problems,'' in \emph{15th Annual ACM Symp.\ on Theory of Computing}, 1983,
  pp. 193--206.

\bibitem{Schnorr:94}
C.~P. Schnorr and M.~Euchner, ``Lattice basis reduction: Improved practical
  algorithms and solving subset sum problems,'' \emph{Mathematical
  Programming}, vol.~66, no.~1, pp. 181--199, Aug. 1994.

\bibitem{Zhang:12}
W.~Zhang, S.~Qiao, and Y.~Wei, ``{HKZ} and {Minkowski} reduction algorithms for
  lattice-reduction-aided {MIMO} detection,'' \emph{IEEE Trans.\ Signal
  Process.}, vol.~60, no.~11, pp. 5963--5976, Nov. 2012.

\bibitem{Helfrich:85}
B.~Helfrich, ``Algorithms to construct {Minkowski} reduced and {Hermite}
  reduced lattice bases,'' \emph{Theoretical Computer Science}, vol.~41, pp.
  125--139, 1985.

\bibitem{Micciancio:09}
D.~Micciancio and O.~Regev, \emph{Lattice-based Cryptography}.\hskip 1em plus
  0.5em minus 0.4em\relax Berlin, Heidelberg: Springer-Verlag, 2009, pp.
  147--191.

\bibitem{Tse:05}
D.~N. Tse and P.~Viswanath, \emph{Fundamentals of Wireless
  Communication}.\hskip 1em plus 0.5em minus 0.4em\relax New York, NY, USA:
  Cambridge University Press, 2005.

\bibitem{Agrell:02}
E.~Agrell, T.~Eriksson, A.~Vardy, and K.~Zeger, ``Closest point search in
  lattices,'' \emph{IEEE Trans.\ Inf.\ Theory}, vol.~48, no.~8, pp. 2201--2214,
  Aug. 2002.

\bibitem{Yao:02}
H.~Yao and G.~W. Wornell, ``Lattice-reduction-aided detectors for {MIMO}
  communication systems,'' in \emph{IEEE Global Telecomm.\ Conf.\ (GLOBECOM)},
  Taipei, Taiwan, Nov. 2002, pp. 424--428.

\bibitem{Windpassinger:03}
C.~Windpassinger and R.~F.~H. Fischer, ``Low-complexity near-maximum-likelihood
  detection and precoding for {MIMO} systems using lattice reduction,'' in
  \emph{IEEE Inf.\ Theory Workshop (ITW)}, Paris, France, Mar. 2003, pp.
  345--348.

\bibitem{Windpassinger:04}
C.~Windpassinger, ``Detection and precoding for multiple input multiple output
  channels,'' Ph.D. dissertation, University of Erlangen-Nuremberg (FAU),
  Erlangen, Germany, 2004.

\bibitem{Wuebben:04}
D.~W{\"u}bben, R.~B{\"o}hnke, V.~K{\"u}hn, and K.-D. Kammeyer,
  ``Near-maximum-likelihood detection of {MIMO} systems using {MMSE}-based
  lattice-reduction,'' in \emph{IEEE Int.\ Conf.\ on Commun.\ (ICC)}, vol.~2,
  Paris, France, Jun. 2004, pp. 798--802.

\bibitem{Wuebben:11}
D.~W{\"u}bben, D.~Seethaler, J.~Jald{\'e}n, and G.~Matz, ``Lattice reduction,''
  \emph{IEEE Signal Process.\ Mag.}, vol.~28, no.~3, pp. 70--91, May 2011.

\bibitem{Fischer:19}
R.~F.~H. Fischer, S.~Stern, and J.~B. Huber, ``Lattice-reduction-aided and
  integer-forcing equalization: Structures, criteria, factorization, and
  coding,'' \emph{Foundations and Trends® in Commun.\ and Inf.\ Theory},
  vol.~16, no. 1-2, pp. 1--155, 2019.

\bibitem{Taherzadeh:07}
M.~Taherzadeh, A.~Mobasher, and A.~K. Khandani, ``{LLL} reduction achieves the
  receive diversity in {MIMO} decoding,'' \emph{IEEE Trans.\ Inf.\ Theory},
  vol.~53, no.~12, pp. 4801--4805, Dec. 2007.

\bibitem{Zhan:14}
J.~Zhan, B.~Nazer, U.~Erez, and M.~Gastpar, ``Integer-forcing linear
  receivers,'' \emph{IEEE Trans.\ Inf.\ Theory}, vol.~60, no.~12, pp.
  7661--7685, Dec. 2014.

\bibitem{Fischer:02}
R.~F.~H. Fischer, \emph{Precoding and Signal Shaping for Digital
  Transmission}.\hskip 1em plus 0.5em minus 0.4em\relax New York, NY, USA: John
  Wiley \& Sons, 2002.

\bibitem{vanTrees:04}
H.~L. Van~Trees, \emph{Detection, Estimation, and Modulation Theory:
  Part~{I}---Detection, Estimation, and Filtering Theory}, 2nd~ed.\hskip 1em
  plus 0.5em minus 0.4em\relax New York, NY, USA: Wiley, 2013.

\bibitem{Huber:94a}
K.~Huber, ``Codes over {Gaussian} integers,'' \emph{IEEE Trans.\ Inf.\ Theory},
  vol.~40, no.~1, pp. 207--216, Jan. 1994.

\bibitem{Conway:99}
J.~H. Conway and N.~J. Sloane, \emph{Sphere Packings, Lattices and Groups},
  3rd~ed., ser. A Series of Comprehensive Studies in Mathematics.\hskip 1em
  plus 0.5em minus 0.4em\relax New York, NY, USA: Springer, 1999, vol. 290.

\bibitem{Jiang:13}
H.~Jiang and S.~Du, ``Complex {Korkine-Zolotareff} reduction algorithm for
  full-diversity {MIMO} detection,'' \emph{IEEE Commun.\ Lett.}, vol.~17,
  no.~2, pp. 381--384, Feb. 2013.

\bibitem{Ding:17}
L.~Ding, Y.~Wang, and J.~Zhang, ``Complex {Minkowski} reduction and a
  relaxation for near-optimal {MIMO} linear equalization,'' \emph{IEEE Wireless
  Commun.\ Lett.}, vol.~6, no.~1, pp. 38--41, Feb. 2017.

\bibitem{Tunali:15}
N.~E. Tunali, Y.-C. Huang, J.~J. Boutros, and K.~R. Narayanan, ``Lattices over
  {Eisenstein} integers for compute-and-forward,'' \emph{IEEE Trans.\ Inf.\
  Theory}, vol.~61, no.~10, pp. 5306--5321, Oct. 2015.

\bibitem{Stern:16}
S.~Stern and R.~F.~H. Fischer, ``Advanced factorization strategies for
  lattice-reduction-aided preequalization,'' in \emph{IEEE Int.\ Symp. on Inf.\
  Theory (ISIT)}, Barcelona, Spain, Jul. 2016, pp. 1471--1475.

\bibitem{Huber:94b}
K.~Huber, ``Codes over {Eisenstein-Jacobi} integers,'' \emph{Contemporary
  Mathematics}, vol. 168, pp. 165--165, 1994.

\bibitem{Karlsson:16}
M.~Karlsson and E.~Agrell, \emph{Multidimensional Optimized Optical Modulation
  Formats}.\hskip 1em plus 0.5em minus 0.4em\relax John Wiley \& Sons, 2016,
  ch.~2, pp. 13--64.

\bibitem{Cui:14}
Y.~{Cui}, R.~{Li}, and H.~{Fu}, ``A broadband dual-polarized planar antenna for
  {2G/3G/LTE} base stations,'' \emph{IEEE Trans.\ Antennas Propag.}, vol.~62,
  no.~9, pp. 4836--4840, Sep. 2014.

\bibitem{Li:16}
M.~{Li}, Y.~{Ban}, Z.~{Xu}, G.~{Wu}, C.~{Sim}, K.~{Kang}, and Z.~{Yu},
  ``Eight-port orthogonally dual-polarized antenna array for {5G} smartphone
  applications,'' \emph{IEEE Trans.\ Antennas Propag.}, vol.~64, no.~9, pp.
  3820--3830, Sep. 2016.

\bibitem{Ghaedi:20}
F.~Ghaedi, J.~Jamali, and M.~Taghizadeh, ``A wideband dual-polarized antenna
  using magneto-electric dipoles for base station applications,''
  \emph{AE{\"U}---Int.\ J.\ of Electronics and Commun.}, vol. 126, p. 153395,
  2020.

\bibitem{Alamouti:98}
S.~M. {Alamouti}, ``A simple transmit diversity technique for wireless
  communications,'' \emph{IEEE J.\ Sel.\ Areas Commun.}, vol.~16, no.~8, pp.
  1451--1458, Oct. 1998.

\bibitem{Conway:03}
J.~H. Conway and D.~A. Smith, \emph{On Quaternions and Octonions}.\hskip 1em
  plus 0.5em minus 0.4em\relax Boca Raton, FL, USA: A K Peters/CRC Press, 2003.

\bibitem{Isaeva:95}
O.~M. {Isaeva} and V.~A. {Sarytchev}, ``Quaternion presentations polarization
  state,'' in \emph{2nd Topical Symp.\ on Combined Optical-Microwave Earth and
  Atmosphere Sensing}, Atlanta, GA, USA, Apr. 1995, pp. 195--196.

\bibitem{Wysocki:06}
B.~J. {Wysocki}, T.~A. {Wysocki}, and J.~{Seberry}, ``Modeling dual
  polarization wireless fading channels using quaternions,'' in \emph{Joint IST
  Workshop on Mobile Future and Symp.\ on Trends in Commun.}, Bratislava,
  Slovakia, Jun. 2006, pp. 68--71.

\bibitem{Napias:96}
H.~Napias, ``A generalization of the {LLL}-algorithm over {Euclidean} rings or
  orders,'' \emph{Journal de Th{\'e}orie des Nombres de Bordeaux}, vol.~8,
  no.~2, pp. 387--396, 1996.

\bibitem{Gan:09}
Y.~H. Gan, C.~Ling, and W.~H. Mow, ``Complex lattice reduction algorithm for
  low-complexity full-diversity {MIMO} detection,'' \emph{IEEE Trans.\ Signal
  Process.}, vol.~57, no.~7, pp. 2701--2710, Jul. 2009.

\bibitem{Stern:15}
S.~Stern and R.~F.~H. Fischer, ``Lattice-reduction-aided preequalization over
  algebraic signal constellations,'' in \emph{9th Int.\ Conf.\ on Signal
  Process.\ and Commun.\ Systems (ICSPCS)}, Cairns, QLD, Australia, Dec. 2015.

\bibitem{Lyu:20}
S.~{Lyu}, C.~{Porter}, and C.~{Ling}, ``Lattice reduction over imaginary
  quadratic fields,'' \emph{IEEE Trans.\ Signal Process.}, vol.~68, pp.
  6380--6393, 2020.

\bibitem{Ding:15}
L.~Ding, K.~Kansanen, Y.~Wang, and J.~Zhang, ``Exact {SMP} algorithms for
  integer-forcing linear {MIMO} receivers,'' \emph{IEEE Trans.\ Wireless
  Commun.}, vol.~14, no.~12, pp. 6955--6966, Dec. 2015.

\bibitem{Fischer:16}
R.~F.~H. Fischer, M.~Cyran, and S.~Stern, ``Factorization approaches in
  lattice-reduction-aided and integer-forcing equalization,'' in \emph{Int.\
  Zurich Seminar on Commun.}, Zurich, Switzerland, Mar. 2016, pp. 108--112.

\bibitem{Wen:19}
J.~{Wen}, L.~{Li}, X.~{Tang}, and W.~H. {Mow}, ``An efficient optimal algorithm
  for the successive minima problem,'' \emph{IEEE Trans.\ Commun.}, vol.~67,
  no.~2, pp. 1424--1436, Feb. 2019.

\bibitem{Nguyen:09}
P.~Nguyen and B.~Vall{\'e}e, \emph{The {LLL} Algorithm: Survey and
  Applications}, ser. Information Security and Cryptography.\hskip 1em plus
  0.5em minus 0.4em\relax Berlin, Heidelberg: Springer-Verlag, 2009.

\bibitem{Neumann:94}
P.~M. Neumann, G.~A. Stoy, and E.~C. Thompson, \emph{Groups and Geometry}, ser.
  Oxford Science Publications.\hskip 1em plus 0.5em minus 0.4em\relax Oxford,
  United Kingdom: Oxford University Press, 1994.

\bibitem{Karatsuba:62}
A.~A. Karatsuba and Y.~P. Ofman, ``Multiplication of many-digital numbers by
  automatic computers,'' in \emph{Doklady Akademii Nauk}, vol. 145,
  no.~2.\hskip 1em plus 0.5em minus 0.4em\relax Russian Academy of Sciences,
  1962, pp. 293--294.

\bibitem{Cariow:15}
A.~Cariow and G.~Cariowa, ``An unified approach for developing rationalized
  algorithms for hypercomplex number multiplication,'' \emph{Electric Review},
  vol.~91, no.~2, pp. 36--39, 2015.

\bibitem{Aslaksen:96}
H.~Aslaksen, ``Quaternionic determinants,'' \emph{The Mathematical
  Intelligencer}, vol.~18, no.~3, pp. 57--65, 1996.

\bibitem{Mueller:12}
R.~R. {M{\"u}ller} and B.~{Cakmak}, ``Channel modelling of {MU-MIMO} systems by
  quaternionic free probability,'' in \emph{IEEE Int.\ Symp. on Inf.\ Theory
  (ISIT)}, Cambridge, MA, USA, 2012, pp. 2656--2660.

\bibitem{Weintraub:08}
S.~H. Weintraub, \emph{Factorization: Unique and Otherwise}, ser. CMS Treatises
  in Mathematics.\hskip 1em plus 0.5em minus 0.4em\relax Boca Raton, FL, USA: A
  K Peters/CRC Press, 2008.

\bibitem{Cormen:09}
T.~H. Cormen, C.~E. Leiserson, R.~L. Rivest, and C.~Stein, \emph{Introduction
  to Algorithms}, ser. Computer Science.\hskip 1em plus 0.5em minus 0.4em\relax
  Cambridge, MA, USA: MIT Press, 2009.

\bibitem{Bossert:99}
M.~Bossert, \emph{Channel Coding for Telecommunications}.\hskip 1em plus 0.5em
  minus 0.4em\relax Chichester, United Kingdom: John Wiley \& Sons, 1999.

\bibitem{Conway:82}
J.~{Conway} and N.~{Sloane}, ``Fast quantizing and decoding algorithms for
  lattice quantizers and codes,'' \emph{IEEE Trans.\ Inf.\ Theory}, vol.~28,
  no.~2, pp. 227--232, Mar. 1982.

\bibitem{Zamir:14}
R.~Zamir, B.~Nazer, Y.~Kochman, and I.~Bistritz, \emph{Lattice Coding for
  Signals and Networks: A Structured Coding Approach to Quantization,
  Modulation and Multiuser Information Theory}.\hskip 1em plus 0.5em minus
  0.4em\relax Cambridge, United Kingdom: Cambridge University Press, 2014.

\bibitem{Daude:94}
H.~Daud{\'e} and B.~Vall{\'e}e, ``An upper bound on the average number of
  iterations of the {LLL} algorithm,'' \emph{Theoretical Computer Science},
  vol. 123, no.~1, pp. 95--115, 1994.

\bibitem{Press:07}
W.~H. Press, S.~A. Teukolsky, W.~T. Vetterling, and B.~P. Flannery,
  \emph{Numerical Recipes: The Art of Scientific Computing}, 3rd~ed.\hskip 1em
  plus 0.5em minus 0.4em\relax Cambridge, United Kingdom: Cambridge University
  Press, 2007.

\bibitem{Akhavi:03}
A.~Akhavi, ``The optimal {LLL} algorithm is still polynomial in fixed
  dimension,'' \emph{Theoretical Computer Science}, vol. 297, no.~1, pp. 3--23,
  2003.

\bibitem{Fischer:10}
R.~F.~H. Fischer, ``From {Gram--Schmidt} orthogonalization via sorting and
  quantization to lattice reduction,'' in \emph{Joint Workshop on Coding and
  Commun.\ (JWCC)}, Santo Stefano Belbo, Italy, Oct. 2010.

\bibitem{Lagarias:90}
J.~C. Lagarias, H.~W. Lenstra, and C.~P. Schnorr, ``{Korkin-Zolotarev} bases
  and successive minima of a lattice and its reciprocal lattice,''
  \emph{Combinatorica}, vol.~10, no.~4, pp. 333--348, Dec. 1990.

\bibitem{Minkowski:89}
H.~Minkowski, ``Diskontinuit{\"a}tsbereich f{\"u}r arithmetische
  {{\"A}quivalenz},'' in \emph{Ausgew{\"a}hlte Arbeiten zur Zahlentheorie und
  zur Geometrie, Mit D.\ Hilberts Ged{\"a}chtnisrede auf H.\ Minkowski,
  G{\"o}ttingen 1909}.\hskip 1em plus 0.5em minus 0.4em\relax Vienna, Austria:
  Springer, 1989, pp. 73--120.

\bibitem{Blichfeldt:29}
H.~F. Blichfeldt, ``The minimum value of quadratic forms, and the closest
  packing of spheres,'' \emph{Mathematische Annalen}, vol. 101, no.~1, pp.
  605--608, 1929.

\bibitem{Siegel:98}
C.~Siegel and K.~Chandrasekharan, \emph{Lectures on the Geometry of
  Numbers}.\hskip 1em plus 0.5em minus 0.4em\relax Berlin, Heidelberg:
  Springer-Verlag, 1989.

\bibitem{Ling:07}
C.~{Ling} and N.~{Howgrave-Graham}, ``Effective {LLL} reduction for lattice
  decoding,'' in \emph{IEEE Int.\ Symp. on Inf.\ Theory (ISIT)}, 2007, pp.
  196--200.

\bibitem{Ling:09}
C.~Ling, W.~H. Mow, and L.~Gan, ``Dual-lattice ordering and partial lattice
  reduction for {SIC}-based {MIMO} detection,'' \emph{IEEE J.\ Sel.\ Topics
  Signal Process.}, vol.~3, no.~6, pp. 975--985, Dec. 2009.

\bibitem{Proakis:08}
J.~G. Proakis and M.~Salehi, \emph{Digital Communications}, 5th~ed.\hskip 1em
  plus 0.5em minus 0.4em\relax New York, NY, USA: McGraw-Hill, 2008.

\bibitem{Frey:20}
F.~{Frey}, S.~{Stern}, J.~K. {Fischer}, and R.~F.~H. {Fischer}, ``Two-stage
  coded modulation for {Hurwitz} constellations in fiber-optical
  communications,'' \emph{J.\ Lightw.\ Technol.}, vol.~38, no.~12, pp.
  3135--3146, Jun. 2020.

\bibitem{Qureshi:18}
S.~S. {Qureshi}, S.~{Ali}, and S.~A. {Hassan}, ``Optimal polarization diversity
  gain in dual-polarized antennas using quaternions,'' \emph{IEEE Signal
  Process.\ Lett.}, vol.~25, no.~4, pp. 467--471, Apr. 2018.

\bibitem{Zheng:03}
L.~Zheng and D.~N. Tse, ``Diversity and multiplexing: A fundamental tradeoff in
  multiple-antenna channels,'' \emph{IEEE Trans.\ Inf.\ Theory}, vol.~49,
  no.~5, pp. 1073--1096, May 2003.

\bibitem{Hassibi:00}
B.~Hassibi, ``An efficient square-root algorithm for {BLAST},'' in \emph{IEEE
  Int.\ Conf.\ on Acoustics, Speech, and Signal Process.}, Istanbul, Turkey,
  Jun. 2000, pp. 737--740.

\bibitem{Stierstorfer:09}
C.~Stierstorfer, ``A bit-level-based approach to coded multicarrier
  transmission,'' Ph.D. dissertation, University of Erlangen-Nuremberg (FAU),
  Erlangen, Germany, 2009.

\bibitem{Forney:00}
G.~D. {Forney}, M.~D. {Trott}, and {Sae-Young Chung}, ``Sphere-bound-achieving
  coset codes and multilevel coset codes,'' \emph{IEEE Trans.\ Inf.\ Theory},
  vol.~46, no.~3, pp. 820--850, May 2000.

\bibitem{Guzeltepe:13}
M.~G{\"u}zeltepe, ``Codes over {Hurwitz} integers,'' \emph{Discrete
  Mathematics}, vol. 313, no.~5, pp. 704--714, 2013.

\bibitem{Guzeltepe:14}
M.~G{\"u}zeltepe and O.~Heden, ``Perfect {Mannheim, Lipschitz and Hurwitz}
  weight codes,'' \emph{Mathematical Communications}, vol.~19, no.~2, pp.
  253--276, 2014.

\bibitem{Rohweder:21}
D.~Rohweder, S.~Stern, R.~F. Fischer, S.~Shavgulidze, and J.~Freudenberger,
  ``Four-dimensional {Hurwitz} signal constellations, set partitioning,
  detection, and multilevel coding,'' \emph{IEEE Trans.\ Commun.}, vol.~69,
  no.~8, pp. 5079--5090, 2021.

\bibitem{Chae:17}
S.~H. Chae, M.~Jang, S.-K. Ahn, J.~Park, and C.~Jeong, ``Multilevel coding
  scheme for integer-forcing {MIMO} receivers with binary codes,'' \emph{IEEE
  Trans.\ Wireless Commun.}, vol.~16, no.~8, pp. 5428--5441, Aug. 2017.

\bibitem{Fischer:18}
R.~F.~H. Fischer, J.~B. Huber, S.~Stern, and P.~M. Guter, ``Multilevel codes in
  lattice-reduction-aided equalization,'' in \emph{Int.\ Zurich Seminar on
  Inf.\ and Commun.}, Zurich, Switzerland, Feb. 2018, pp. 133--137.

\bibitem{Stern:19b}
S.~Stern, D.~Rohweder, J.~Freudenberger, and R.~F. Fischer, ``Binary multilevel
  coding over {Eisenstein} integers for {MIMO} broadcast transmission,'' in
  \emph{23rd Int.\ ITG Workshop on Smart Antennas (WSA)}, 2019, pp. 198--205.

\end{thebibliography}
\vfill
\vspace*{-5.4mm}
\begin{IEEEbiographynophoto}{Sebastian~Stern}
(Member, IEEE) received the B.Sc.\ and M.Sc.\ degrees in Communications and Computer
Engineering from Ulm University, Ulm, Germany, in 2010 and 2012, respectively, and the
Dr.-Ing.\ (Doctor of Engineering) degree in 2019. He is currently working as a Senior
Researcher at the Institute of Communications Engineering, Ulm University, where he
received the teaching assignment for a lecture on multi-user communications and
multiple-input/multiple-output (MIMO) systems. His main research interests are in the
area of communication theory and especially concern lattice-based approaches and
algorithms for MIMO communications and related fields, coded-modulation techniques,
and massive MIMO schemes.
\end{IEEEbiographynophoto}
\vspace*{-5.4mm}
\begin{IEEEbiographynophoto}{Cong~Ling}
(Member, IEEE) received the B.Sc.\ and M.Sc.\ degrees in Electrical Engineering from the
Nanjing Institute of Communications Engineering, Nanjing, China, in 1995 and 1997,
respectively, and the Ph.D.\ degree in electrical engineering from Nanyang Technological
University, Singapore, in 2005. He is currently a Reader (Associate Professor) with the
Electrical and Electronic Engineering Department, Imperial College London. His research
interests include coding, signal processing, and security, with a focus on lattices.
Before joining Imperial College, he had been on the faculties of the Nanjing Institute
of Communications Engineering and King’s College. He served as an Associate Editor for
\textsc{IEEE Transactions on Communications} and
\textsc{IEEE Transactions on Vehicular Technology}.
\end{IEEEbiographynophoto}
\vspace*{-5.4mm}
\begin{IEEEbiographynophoto}{Robert~F.H.~Fischer}
(Senior Member, IEEE) received the Dr.-Ing.\ and Habilitation degrees from the
University of Erlangen-Nuremberg, Erlangen, Germany, in 1996 and 2001, respectively.

From 1992 to 1996 he was a Research Assistant with the Telecommunications Institute,
University of Erlangen-Nuremberg. In 1997, he was with the IBM Research Laboratory,
Z\"urich, Switzerland. In 1998, he returned to the University of Erlangen-Nuremberg.
In 2005, he spent a sabbatical with ETH Z\"urich. Since 2011 he has been a Full
Professor at Ulm University, Ulm, Germany. He is currently teaching the undergraduate
and graduate courses on signals and systems and digital communications. He authored
the textbook \emph{Precoding and Signal Shaping for Digital Transmission}
(John Wiley~\& Sons, 2002). His current research interests include fast, reliable, and
secure digital transmission, including single-carrier and multi-carrier modulation
techniques, information theory, coded modulation, digital communications, signal
processing, and especially precoding and shaping techniques.

Dr.\ Fischer was a recipient of the Dissertation Award from the Technical Faculty,
University of Erlangen-Nuremberg, in 1997, the Publication Award of the German Society
of Information Techniques in 2000, the Wolfgang Finkelnburg Habilitation Award in 2002,
and the Johann-Philipp-Reis-Preis in 2005.
\end{IEEEbiographynophoto}

\end{document}